 \documentclass[12pt, draftclsnofoot, onecolumn, journal]{IEEEtran}
\IEEEoverridecommandlockouts
\usepackage{xcolor}
\usepackage{tikz}
\usetikzlibrary{decorations}
\usepackage{pgfplots}
\usepackage[ruled]{algorithm2e}
\SetKwInput{kwInit}{Initialization}
\usepackage{amsmath,epsfig,amssymb,amsthm,cite,url,xcolor}
\allowdisplaybreaks
\usepackage{hyperref}
\usepackage{enumerate}
\usepackage[latin1]{inputenc}
\DeclareMathOperator{\tr}{tr}
\DeclareMathOperator{\diag}{diag}
\DeclareMathOperator{\var}{var}   
\theoremstyle{plain}
\newtheorem{theorem}{Theorem}
\newtheorem{assumption}{Assumption}
\newtheorem{lemma}[theorem]{Lemma}

\newtheorem{cor}{Corollary}
\theoremstyle{definition}

\usepackage{makecell}
\newtheorem{remark}{Remark}
\usepackage{makecell}

\newtheoremstyle{specialcasestyle}{1mm}{1mm}{\upshape}{}{\bfseries\upshape}{.}{0mm}{}
\theoremstyle{specialcasestyle}

\newcommand{\figref}[1]{Fig.~\protect\ref{#1}}

\newcommand{\bh}{{\bf h}}

\newcommand{\bD}{{\bf D}}

\newcommand{\bx}{{\bf x}}
\newcommand{\bv}{{\bf v}}

\newcommand{\by}{{\bf y}}

\newcommand{\bz}{{\bf z}}

\newcommand{\bX}{{\bf X}}
\newcommand{\bA}{{\bf A}}

\newcommand{\bR}{{\bf R}}
\newcommand{\bI}{{\bf I}}

\newcommand{\bC}{{\bf C}}

\newcommand{\ex}{{\mathbb E}}

\newcommand{\go}{{\mathcal{O}}}

\newcommand{\bnu}{{\boldsymbol\nu}}

\newcommand{\btt}{{\boldsymbol\theta}}

\newcommand{\bxi}{{\boldsymbol\xi}}

\newcommand{\bzero}{{\boldsymbol 0}}

\makeatother
\def\tablename{Table}
\usepackage{graphicx}
\usepackage{float}
\usepackage{tabularx}
\usepackage{subfigure}
\usepackage{amsmath,epsfig,amssymb}
\usepackage{amsmath}
\usepackage{bbm}
\usetikzlibrary{intersections}
\usepgfplotslibrary{fillbetween}

\usepackage{mathtools}
\pgfplotsset{compat=1.16}
\usepackage{multirow}

\begin{document}

\title{Over-The-Air Federated Learning Over Scalable Cell-free Massive MIMO
}

\author{Houssem Sifaou,~\IEEEmembership{Member,~IEEE}, Geoffrey Ye Li,~\IEEEmembership{Fellow,~IEEE}}
%

\maketitle

\begin{abstract}

Cell-free massive MIMO is emerging as a promising technology for future wireless communication systems, which is expected to offer uniform coverage and high spectral efficiency compared to classical cellular systems. We study in this paper how cell-free massive MIMO can support federated edge learning. Taking advantage of the additive nature of the wireless multiple access channel, over-the-air computation is exploited, where the clients send their local updates simultaneously over the same communication resource. Such an approach, known as over-the-air federated learning (OTA-FL), is proven to alleviate the communication overhead of federated learning over wireless networks. Considering channel correlation and only imperfect channel state information available at the central server, we propose a practical implementation of OTA-FL over cell-free massive MIMO. The convergence of the proposed implementation is studied analytically and experimentally, confirming the benefits of cell-free massive MIMO for OTA-FL.  
\end{abstract}

\section{Introduction}
With the exponential growth in data collection from distributed devices and remarkable advancements in AI applications, there is a growing interest in distributed learning, which leverages the computing capabilities of these devices. Distributed learning offers two primary benefits. Firstly, it eliminates the need to transfer massive volumes of high-dimensional data from the collecting devices to central servers (CSs) for further processing. Secondly, it inherently ensures privacy since local data is not shared. This work specifically focuses on federated learning (FL), a popular technique in distributed learning, where a CS coordinates global model training by communicating with multiple distributed clients \cite{pmlr-v54-mcmahan17a,kairouz2021advances,9530714}. Instead of sharing local data, the parameters of the machine learning (ML) model are transferred. The FL process involves multiple training rounds, during which the global model is shared with the clients. The clients then perform several local iterations using their respective local data and send their model updates back to the CS. These model updates are subsequently aggregated to obtain the new global model.

In this work, we focus specifically on over-the-air FL (OTA-FL), which has been recently proposed to alleviate the communication overhead of the FL process \cite{Kai2020,amiri2020federated}. By exploiting the additive nature of the multiple access channel, OTA-FL allows the clients to send their local updates simultaneously and the CS can obtain the sum directly from the received signal. Such an approach has been extensively studied in literature and has been proven to enhance the communication efficiency of FL over wireless networks. The main advantage of OTA-FL is that the communication resources needed to reach convergence do not scale with the number of clients. In contrast, the communication resources of digital transmission based FL scale with the number of clients \cite{harnessing}.

In this paper, we consider implementing OTA-FL over cell-free massive MIMO. The latter is a promising technique for future wireless communication systems, which is expected to bring several benefits such as ultra-high reliability, high energy efficiency, ultra-low latency, and ubiquitous and uniform coverage \cite{matthaiou2021road,interdonato2019ubiquitous,he2021cell}. The main idea of cell-free massive MIMO is deploying a large number of distributed APs, connected to a central processing unit (CPU), that serve all users in a wide coverage area~\cite{he2021cell}. Specifically, we consider in this work the scalable cell-free architecture recently proposed in \cite{Emilcellfree}. It is a user-centric architecture where a cluster composed of a finite number of APs is formed around each user terminal (UT). The latter will only communicate with the APs in its cluster. This is the main difference compared to the early architectures proposed for cell-free massive MIMO, where each UT is served by all the APs in the network \cite{venkatesan2007network,5594708}. The scalable architecture \cite{Emilcellfree} allows the number of UTs and APs to be very large while keeping all the signal processing algorithms and the fronthaul signaling practically implementable. The main objective of our work is to explore the benefits that cell-free massive MIMO can bring for OTA-FL.

Significant efforts have been recently dedicated for the design and performance analysis of FL over wireless networks \cite{tran2019federated,zhu2019broadband,amiri2021blind,harnessing,9042352,Kai2020,Liu2021,amiri2020federated,vu2020cell,jeon2020compressive,Sery2020,Sery2020a, sifaou2021robust,sifaou2022over,chang2020communication,zhu2020one,yang2020age,shi2020device, yang2019scheduling,sun2020energy,amiri2020update,ye2022decentralized}, referred to as federated edge earning (FEEL). Some of the studies focused on exploring over-the-air computation to reduce the communication overhead of the FL process \cite{Kai2020,amiri2021blind,amiri2020federated,sifaou2021robust,sifaou2022over,Sery2020,Sery2020a,harnessing,zhu2020one}. While most of the efforts considered first-order optimization methods, namely stochastic gradient descent (SGD), some works have studied the convergence of second-order methods such as the alternating direction method of multipliers (ADMM) \cite{harnessing,zhou2021communication}. Moreover, different scheduling techniques have been proposed for FEEL \cite{yang2020age,shi2020device, yang2019scheduling,sun2020energy,amiri2020update}, where a subset of clients is selected at each communication round considering the limited communication resources. Resource allocation has been also an active research direction in FEEL \cite{vu2020cell,ren2020accelerating,zeng2021energy,chen2020joint,dinh2020federated,shi2020joint}, with the aim of accelerating convergence and reducing the effect of the wireless channel. Furthermore, assuming multiple antennas at the CS and designing beamforming and receiving techniques have been considered in \cite{amiri2021blind,Kai2020,jeon2020compressive}. In \cite{amiri2021blind} a blind FEEL scheme was proposed assuming the CS is equipped with a large number of antennas and imperfect channel state information (CSI) available at the CS. The settings considered in \cite{amiri2021blind} are such that the CS is connected to a base station (BS) equipped with a large number of antennas, communicating with a finite number of clients participating in the FL process. Our work differs from \cite{amiri2021blind} in two aspects. Particularly, we consider in our work cell-free massive MIMO instead of cellular MIMO and the more practical scenario where the number of APs (and thus the number of antennas) and the number of clients are both large. To the best of our knowledge, our work in this paper is the first effort to study the performance of OTA-FL in the context of cell-free massive MIMO.

  OTA-FL has been studied extensively in both point-to-point and MIMO settings. However, these studies cannot be directly applied to cell-free massive MIMO scenarios. In Section III, we thoroughly demonstrate the vanishing interference terms and establish the necessary conditions specific to cell-free massive MIMO systems. These conditions differ significantly from those studied in the context of cellular massive MIMO, as outlined in \cite{amiri2021blind}. Particularly, while the conditions for cellular massive MIMO focus on a large number of antennas and a finite number of clients, our findings highlight the importance of both a large number of clients and a large number of antennas for cell-free massive MIMO. Moreover, proving the interference vanishing necessitates some practical assumptions specific to this setting. These assumptions, detailed in (\ref{assump}), are essential to justify the feasibility of OTA-FL in cell-free MIMO systems. 
Our proposed analysis also addresses a critical challenge often overlooked in existing literature, which is the scalability issue. Previous studies on OTA-FL in cellular MIMO predominantly focus on single-cell scenarios with a limited number of clients, neglecting inter-cell interference. In contrast, our analysis considers the entire cell-free network with a substantial number of clients, presenting a more scalable approach. By investigating OTA-FL within this comprehensive network framework, we contribute to the understanding of its applicability and performance in large-scale deployments. Finally, a strong motivation for this work is to study the advantages of implementing OTA-FL in terms of energy efficiency compared with cellular MIMO. This is an important research question and our work is a first attempt to answer it. Our experiments show clearly the advantage of cell-free massive MIMO systems in terms of energy efficiency, which is an essential factor in FL taking into account the limited power at the clients' devices.

The contributions of our work can be summarized as follows:
\begin{itemize}
\item We study, for the first time, the performance of OTA-FL over cell-free massive MIMO and investigate the benefits that the latter can bring for OTA-FL in terms of energy efficiency.
\item We propose a practical implementation of OTA-FL over cell-free massive MIMO under reasonable assumptions and a practical channel model. Specifically, we propose a receiver based on maximum ratio combining (MRC) that is shown to mitigate the effect of the fading channel as well as the interference.
\item We provide the convergence rate of the proposed implementation and test its performance via extensive experiments. Numerical results show that OTA-FL over cell-free massive MIMO requires lower power constraints at the clients to reach convergence compared to cellular massive MIMO. 
\end{itemize}

Throughout the paper, boldface lowercase is used for denoting column vectors, $\bx$, and upper case for matrices, $\bX$. The superscripts $^T$, $^H$, and $^*$ denote the transpose, the conjugate transpose and the conjugate, respectively. $\bI_p$ denotes the $p \times p$ identity matrix and $\|.\|$ is used to denote the $\ell_2$ norm. The trace of a matrix $\bA$ is denoted by $\tr \bA$ and $\diag(\bA_1, \cdots,\bA_n)$ stands for the block-diagonal matrix with the square matrices  $\bA_1, \cdots,\bA_n$ on the diagonal. $\mathcal{NC}(\bzero, \bR)$ denotes the multivariate circularly symmetric complex Gaussian distribution with correlation matrix $\bR$. $|\mathcal{S}|$ denotes the cardinality of a set $\mathcal{S}$. The expected value and the variance of the random quantity $x$ are denoted as $\mathbb{E}\{x\}$ and $\var\{x\}$, respectively. $Re \ x$ and $Im \ x$ are used to denote the real and imaginary parts of the complex number $x$. $[n]$ denotes the set $\{1,\cdots,n\}$. The notations $f(n) = \mathcal{O}(g(n))$ and $f(n) = o(g(n))$  mean that $\left|\frac{f(n)}{g(n)}\right|$ is bounded or approaches zero, respectively, as $n\to \infty$.

The remainder of the paper is structured as follows. In the next section, we present the system model and introduce  OTA-FL and cell-free massive MIMO. In Section \ref{Proposed_implementation}, the proposed FEEL scheme is detailed, while in Section \ref{conv}, its convergence is studied. We provide, in Section \ref{num}, extensive numerical experiments to test the performance of the proposed implementation. Finally, concluding remarks are drawn in Section \ref{conc}.

\section{System Model}
\label{system_model}
In this section, we give a brief overview of OTA-FL and cell-free massive MIMO. Particularly, we introduce the scalable user-centric architecture of cell-free massive MIMO that will be considered in our proposed implementation.

\subsection{Federated learning}
We consider a FL system, where $N$ clients, each with a local dataset $\mathcal{D}_n$ of size $D_n=|\mathcal{D}_n|$, communicate with a CS to collaboratively train a global model. The result of the FL process is the optimal parameter vector, $\btt^\star\in\mathbb{R}^p$, that minimizes a global loss function $F(\btt)$ given by
\begin{equation}
F(\btt)=\sum_{n=1}^N \frac{D_n}{D}F_n(\btt),
\end{equation}
where $D=\sum_{n=1}^N D_n $ and $F_n(\btt)$ is the local empirical loss function at client $k$ given by
\begin{equation}
F_n(\btt)=\frac{1}{D_n}\sum_{\bxi\in \mathcal{D}_n}\mathcal{L}(\btt; \bxi),
\end{equation}
where $\mathcal{L}(\btt; \bxi)$ denotes the loss function at data sample $\bxi$, which depends on the learning task.
One of the main challenges of FL is the communication overhead. In fact, the local updates are sent at every global training round to the CS usually over limited bandwidth wireless channels. To reduce the communication overhead, over-the-air communication has been proposed as a promising solution, where the model updates are sent simultaneously over the multiple access channel \cite{Kai2020,Sery2020}. This approach, referred to as OTA-FL, will be introduced hereafter.

OTA-FL consists of multiple global training rounds where the CS sends the model parameter vector, $\btt_t$, to the clients at each round $t$. Due to the high power available at the CS, the downlink communication is usually assumed to be perfect and the clients receive the global model without distortion \cite{Sery2020}. Then, client $n$  sets its local model as $\btt_{0}^n(t)=\btt(t)$ and runs its local SGD for $\tau$ iterations based on its local dataset 
\begin{equation}
\btt_{i+1}^n(t)=\btt_{i}^n(t)-\eta_t  F_{n,\xi_{n,i}^t}'(\btt_{i}^n(t)), \ \ {\rm for} \ \ i = 0,1,\cdots, \tau-1,
\label{SGD}
\end{equation}
where $\eta_t$ is the SGD step size at round $t$ and $F_{n,\xi_{n,i}^t}'(\btt_{i}^n(t))$ denotes the stochastic gradient computed using a local mini-batch sample $\xi_{n,i}^t$ chosen uniformly at random from the local dataset of client $n$. 
The clients then send their model updates,
\begin{equation}
\Delta\btt^n(t)= \btt_{\tau}^n(t) - \btt({t}),
\end{equation}
simultaneously to the CS via analog OTA. The CS recovers the new global model from the received signals. Further details will be given in Section \ref{Proposed_implementation}.

\subsection{Cell-free massive MIMO}
One of the promising technologies for 6G is cell-free massive MIMO, where distributed APs are deployed over the coverage area \cite{ngo2017cell,nayebi2017precoding}. Each AP can serve multiple UTs and a UT can communicate with several APs, which is different from classical cellular massive MIMO where each UT is served by a single BS. Cell-free massive MIMO is shown to provide better spectral efficiency than classical cellular massive MIMO \cite{Emilcellfree}.

We adopt in this work the scalable cell-free massive MIMO network proposed in \cite{Emilcellfree}. Such a network can support a large number of UTs while the signal processing algorithms and the fronthaul signaling remain practically implementable. The network is composed of $N$ single antenna UTs and $L$ APs, each equipped with $M$ antennas. Each subset of APs is connected to a CPU via high-capacity backhaul links. A user-centric network architecture is assumed where each UT is served by a subset of the APs providing the best channel conditions as illustrated in~\figref{cell_free}. We denote by $\mathcal{S}_\ell$ the subset of UTs served by AP $\ell$.
\begin{figure}[h]
\begin{center}
  \includegraphics[width=0.5\linewidth]{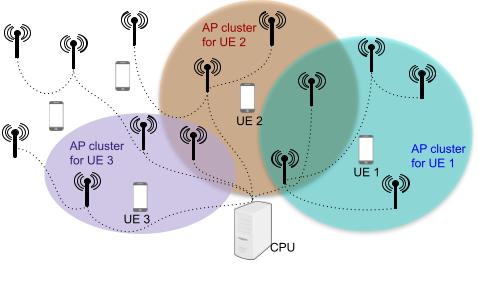}
  \caption{Scalable cell-free massive MIMO network.}
  \label{cell_free}
  \end{center}
\end{figure}

Since the model updates of a client will be received by multiple APs, we think that implementing OTA-FL over cell-free massive MIMO will be more robust to poor channel conditions and CSI errors. The objective of this work is to propose a practical implementation of OTA-FL over cell-free massive MIMO, investigate its advantages, and study its convergence.

\section{OTA-FL over cell-free massive MIMO}
\label{Proposed_implementation}
The UTs and the CPU\footnote{In practice, several CPUs are deployed in the network connected to a master CPU. In this case, the latter will be the CS coordinating the FL process.} are considered as the clients and the CS, respectively. We assume that the number of APs and the number of UTs in the network are large, while the number of antennas per AP is finite. These are typical assumptions in cell-free massive MIMO \cite{Emilcellfree}.

 At each global round, the local updates with dimension $d$ are sent over $S$ OFDM subchannels during $K$ transmissions, where $K=\lceil\frac{d}{2S}\rceil$. For transmission $k$ over OFDM subcarrier $s$, two entries of $\Delta\btt^n(t)$ are sent as complex symbol as follows
$$
x^{k,s}_n (t)= x^{k,s}_{n,re}(t) +j x^{k,s}_{n,im} (t), \ \  k\in[K],
$$
where 
\begin{align*}
&x^{k,s}_{n,re}(t)=\Delta\btt^n_{(2k-1)+2(s-1)K}(t) \\
&x^{k,s}_{n,im}(t)=\Delta\btt^n_{2k+2(s-1)K}(t).
\end{align*}
where $\Delta\btt^n_{i}(t)$ is the $i$-th entry of $\Delta\btt^n(t)$. Each client sends the following signal at the $k$-th OFDM symbol on subcarrier $s$:
\begin{align*}
\tilde x^{k,s}_n(t)={\alpha_t}x^{k,s}_n (t),
\end{align*}
where $\alpha_t$ is a scaling factor to satisfy a power budget constraint at the clients, defined as
\begin{align}
\frac{1}{TK}\sum_{t=1}^T\sum_{k=1}^K\ex \| \tilde \bx^{k}_n(t)\|^2\leq P,  \ \forall n
\label{pow_cons}
\end{align}
where $\tilde \bx^{k}_n(t) =[\tilde x^{k,1}_n(t), \cdots, \tilde x^{k, S}_n(t)]^T\in \mathbb{C}^{S\times 1} $ is the vector of signals transmitted on different subcarriers at transmission $k$ from client $n$. It is important to note that the total number of global rounds is not known in practice, so the scaling factor should not be determined using the power constraint. A more reasonable assumption is that $\alpha_t$ is determined based on the power budgets of the clients. Such information can be exchanged between the clients and the CS before or during the learning process and the CS can choose $\alpha_t$ accordingly. However, we adopt the constraint in \eqref{pow_cons} to have an idea about the amount of power needed to reach convergence and compare different implementation schemes of OTA-FL in terms of energy efficiency.

The received signal at AP $\ell$ during OFDM symbol transmission $k$ is given by
\begin{align}
\by^{k,s}_{\ell}(t)=\alpha_t\sum_{n=1}^N \bh^{k,s}_{n,\ell}(t)  x^{k,s}_n(t)+\bz^{k,s}_{\ell}(t), \ \  s\in[S],  \ \ \ell \in [L],\ \  k\in[K],
\label{eqy}
\end{align}
where $\bh^{k,s}_{n,\ell}(t) \in \mathbb{C}^M$ is the channel vector between UT $n$ and AP $\ell$ for transmission $k$ over subcarrier $s$, and $\bz^{k,s}_{\ell}(t)\sim \mathcal{CN}(\bzero,\sigma_z^2\bI_M)$ stands for the additive noise.

A Rayleigh fading channel model is assumed, where the channel vector between UT $n$ and AP $\ell$, at global iteration $t$, subcarrier $s$, and transmission $k$ is distributed as \cite{Emilcellfree}
\begin{align*}
\bh_{n,\ell}^{k,s}(t)\sim \mathcal{CN}(\bzero, \bR_{n,\ell}),
\end{align*}
where $ \bR_{n,\ell} \in \mathbb{C}^{M\times M}$ is a correlation matrix describing the large-scale fading, including geometric pathloss, shadowing, antenna gains, and spatial channel correlation~\cite{bjornson2017massive}. We assume that $\bR_{n,\ell}$ is available at the APs, which can be obtained by one of the existing correlation estimation methods \cite{neumann2018covariance,upadhya2018covariance,bjornson2016massive}.

Since the UTs far away from AP $\ell$ will experience a severe large-scale fading and a few UTs are close to a particular AP in practice, we will make the following assumption on the correlation matrix $\bR_{n,\ell}$ 
 
\begin{align}
&\frac{1}{M}\tr\bR_{n,\ell}=\go(1) \quad \text{ for} \quad n\in  \mathcal{S}_\ell\cup\tilde{\mathcal{S}}_\ell,\nonumber \\
&\frac{1}{M}\tr\bR_{n,\ell}=o(1) \quad \text{ for} \quad n\notin  \mathcal{S}_\ell\cup\tilde{\mathcal{S}}_\ell,
\label{assump}
\end{align}
where $\mathcal{S}_\ell$ is the set of UTs served by AP $\ell$ while $\tilde{\mathcal{S}}_\ell$ is a set of UTs close to AP $\ell$ and not served by this AP. In a practical scenario, one can assume that the cardinalities of ${\mathcal{S}}_\ell$ and $\tilde{\mathcal{S}}_\ell$ are finite while the total number of UTs is large, $N\to \infty$. These assumptions are needed for theoretical analysis purposes, in this section, to justify that the interference will vanish as will be detailed later. When implementing our approach in practice, such assumptions are not used. We note also that \eqref{assump} implies 
\begin{align}
&\frac{1}{M}\tr\bR_{n,\ell}\bR_{n',\ell}=o(1)\quad \text{ for} \quad n\in  \mathcal{S}_\ell\cup\tilde{\mathcal{S}}_\ell \ {\rm and} \ n'\notin  \mathcal{S}_\ell\cup\tilde{\mathcal{S}}_\ell, \nonumber \\
&\frac{1}{M}\tr\bR_{n,\ell}\bR_{n',\ell}=\mathcal{O}(1)\quad \text{ for} \quad n\in  \mathcal{S}_\ell\cup\tilde{\mathcal{S}}_\ell \ {\rm and} \ n'\in  \mathcal{S}_\ell\cup\tilde{\mathcal{S}}_\ell.
\label{assump2}
\end{align}

The channel vector from all APs to UT $n$ can be written as  $\bh_{n}^{k,s}(t)\!=\!\left[\bh^{k,s}_{n,1}(t)^T,\cdots,\bh^{k,s}_{n,L}(t)^T\right]^T$, which has the following correlation matrix $\bR_{n}  =\diag(\bR_{n,1},\dots,\bR_{n,L})\in \mathbb{R}^{LM\times LM}$. AP $\ell$ has access to an estimate of the channels of its corresponding UTs, that is, $\hat\bh^{k,s}_{n,\ell}(t) $ is available at AP $\ell$ for $n \in \mathcal{S}_{\ell}$. The channel estimates are such that
\begin{align}
\hat\bh^{k,s}_{n,\ell}(t)  = \bh^{k,s}_{n,\ell}(t) +\tilde\bh^{k,s}_{n,\ell}(t) 
\end{align}
where $\tilde\bh^{k,s}_{n,\ell}(t) \sim \mathcal{CN}(\bzero, \bC_{n,\ell})$ is the channel estimation error with correlation $\bC_{n,\ell}$ and independent of $\bh^{k,s}_{n,\ell}(t) $. The channel estimation can be done in a distributed manner where each AP estimates the channels of its corresponding UTs or in a centralized manner where the received pilot signals are conveyed to the CPU, which will be in charge of channel estimation. In our experiments, we have used the MMSE channel estimation technique as in \cite{Emilcellfree}. 

Each AP sends the received signal $\by_{\ell}^{k,s}(t)$ and the channel estimates, $\hat\bh^{k,s}_{n,\ell}(t)$, of its corresponding UTs to the CPU. The vector of all the received signals from all APs,  $\by^{k,s}(t)=[\by_1^{k,s}(t)^T,\cdots,\by_L^{k,s}(t)^T]^T$, is combined at the CPU using the following receive vector
$$
\bv^{k,s}(t)=\frac{1}{N}\sum_{n=1}^N\frac{\bD_n\hat\bh_{n}^{k,s}(t)}{\tr(\bD_n \bR_n)},
$$
where $\bD_n=\diag(\bD_{n,1},\cdots, \bD_{n,L})$ and $$
\bD_{n,\ell} = \begin{cases}\bI_M \ {\rm if} \ n\in\mathcal{S}_\ell\\
\bzero_M \ \ {\rm otherwise}.\end{cases}
$$
For notation convenience, we define $c_n\triangleq \tr(\bD_n \bR_n) = \sum_{\ell\in \mathcal{K}_n}\tr( \bR_{n,\ell})$, where $\mathcal{K}_n$ denotes the set of indexes of APs serving UT $n$. Combining the received signal vector $\by^{k,s}(t)$ using $\bv^{k,s}(t)$, the CS obtains,
\begin{align}
y^{k,s}(t)&=\bv^{k,s}(t)^H\by^{k,s}(t)\nonumber
\\&=\frac{\alpha_t}{N}\sum_{n=1}^N \sum_{\substack{n'=1}}^N \frac{\hat\bh^{k,s}_{n}(t)^H\bD_n \bh^{k,s}_{n'}(t)}{c_n}  x_{n'}^{k,s}(t)+\frac{1}{N}\sum_{n=1}^N \frac{\hat\bh^{k,s}_{n}(t)^H\bD_n\bz^{k,s}(t)}{c_n }\nonumber\\
&=\underbrace{\frac{\alpha_t}{N}\sum_{n=1}^N \frac{\bh^{k,s}_{n}(t)^H\bD_n \bh^{k,s}_{n}(t)}{c_n } x_n^{k,s}(t)}_{\text{signal term}}+\underbrace{\frac{\alpha_t}{N}\sum_{n=1}^N \sum_{\substack{n'=1\\n'\neq n}}^N \frac{\bh^{k,s}_{n}(t)^H\bD_n \bh^{k,s}_{n'}(t)}{c_n }  x_{n'}^{k,s}(t)}_{\text{interf 1}}\nonumber\\&+ \underbrace{\frac{\alpha_t}{N}\sum_{n=1}^N \sum_{\substack{n'=1}}^N \frac{ \tilde\bh^{k,s}_{n}(t)^H\bD_n \bh^{k,s}_{n'}(t)}{c_n} x_{n'}^{k,s}(t)}_{\text{interf 2}}  +\underbrace{\frac{1}{N}\sum_{n=1}^N \frac{\left(\bh^{k,s}_{n}(t)+\tilde\bh^{k,s}_{n}(t) \ \right)^H\bD_n\bz^{k,s}(t)}{c_n }}_{\text{noise term}}.
 \label{signal-y_imperfect1}
\end{align}
The combined signal is composed of a signal term, two interference terms, and a noise term. In what follows, we will show that the signal term converges to the average of the model updates and the two interference terms vanish as the number of clients grows large, $N\to\infty$.
The signal term in \eqref{signal-y_imperfect1} can be rewritten as
$$
y_{sig}^{k,s}(t)= \frac{\alpha_t}{N} \sum_{n=1}^N \frac{1}{c_n }x_n^{k,s}(t)\sum_{\ell\in \mathcal{K}_n}\bh_{n,l}^{k,s}(t)^H \bh_{n,l}^{k,s}(t).
$$
Using the law of large numbers, when $N\to \infty$, the signal term approaches 
$$
\bar y_{sig}^{k,s}(t)=\frac{\alpha_t}{N} \sum_{n=1}^N \frac{1}{c_n }x_n\sum_{\ell\in \mathcal{K}_n}\ex \bh_{n,l}^{k,s}(t)^H \bh_{n,l}^{k,s}(t) =\frac{\alpha_t}{N} \sum_{n=1}^Nx_n^{k,s}(t).
$$
Let us now analyze the interference terms. The expectation and the variance of $\bh^H_{n}\bD_n \bh_{n'}$ are given by
\begin{align}
\ex \bh_{n}^{k,s}(t)^H\bD_n \bh_{n'}^{k,s}(t) &=0,\\
\var\bh_{n}^{k,s}(t)^H\bD_n \bh_{n'}^{k,s}(t)  &= \tr (\bR_n \bD_n \bR_{n'}\bD_n)=\sum_{\ell\in \mathcal{K}_n}^L\tr (\bR_{n,\ell}  \bR_{n',\ell}).
\end{align}
Thus, the expectation of the first interference term is zero and its variance with respect to the randomness of the channel is given by
\begin{align}
&\var( \text{interf 1})=\var \frac{\alpha_t}{N}\sum_{n=1}^N \sum_{\substack{n'=1\\n'\neq n}}^N \frac{\bh_{n}^{k,s}(t)^H\bD_n \bh_{n'}^{k,s}(t)}{c_n} x_{n'}^{k,s}(t)\\&\quad =\frac{\alpha_t^2}{N^2}\sum_{n=1}^N \sum_{\substack{n'=1\\n'\neq n}}^N\sum_{\ell\in \mathcal{K}_n}\frac{\tr (\bR_{n,\ell}  \bR_{n',\ell})}{c_n^2}|x_{n'}^{k,s}(t)|^2\\
&\quad = \frac{\alpha_t^2}{N^2}\sum_{n=1}^N \sum_{\ell\in \mathcal{K}_n} \sum_{\substack{n'\neq n\\ n'\in  \mathcal{S}_\ell\cup\tilde{\mathcal{S}}_\ell}}^N\frac{\tr (\bR_{n,\ell}  \bR_{n',\ell})}{c_n^2}|x_{n'}^{k,s}(t)|^2+  \frac{\alpha_t^2}{N^2}\sum_{n=1}^N \sum_{\ell\in \mathcal{K}_n} \sum_{\substack{n'\neq n \\ n'\notin  \mathcal{S}_\ell\cup\tilde{\mathcal{S}}_\ell}}^N\frac{\tr (\bR_{n,\ell}  \bR_{n',\ell})}{c_n^2}|x_{n'}^{k,s}(t)|^2
 \label{interf1}
\end{align}
 Using the assumptions in \eqref{assump} and \eqref{assump2}, the first term of \eqref{interf1} is the sum of $N |\mathcal{K}_n| | \mathcal{S}_\ell \cup \tilde{\mathcal{S}}_\ell|$ terms of order $ \mathcal{O}(1)$ normalized by $N^2$, and since  $|\mathcal{K}_n| | \mathcal{S}_\ell \cup \tilde{\mathcal{S}}_\ell|$ is finite compared to $N$, this first term is $o(1)$ when $ N\to \infty$. For the second term, it is a sum of $N |\mathcal{K}_n| (N-| \mathcal{S}_\ell \cup \tilde{\mathcal{S}}_\ell|)$ terms of order $o(1)$ normalized by $N^2$, it is thus $o(1)$.  Based on the above analysis, the variance of the first interference term satisfies
 \begin{align}
\var( \text{interf 1})= o(1), \ \ \  \text{as} \ \  N\to \infty.
\end{align}
Similarly, the expectation of the second interference term is zero and its variance with respect to the randomness of the channel can be obtained as 
\begin{align}
&\var(\text{interf 2})= \var\frac{\alpha_t}{N}\sum_{n=1}^N \sum_{\substack{n'=1}}^N \frac{ \tilde\bh^{k,s}_{n}(t)^H\bD_n \bh^{k,s}_{n'}(t)}{c_n} x_{n'}^{k,s}(t)\\
&\quad =\frac{\alpha_t^2}{N^2}\sum_{n=1}^N \sum_{\substack{n'=1}}^N \frac{ \tr (\bR_{n'}\bD_n\bC_{n}\bD_n)}{c_n^2} |x_{n'}^{k,s}(t)|^2.    
\end{align}
Using similar arguments as in the analysis of the first interference term, one can easily show
 \begin{align}
\var( \text{interf 2})= o(1), \ \ \  \text{as} \ \  N\to \infty.
\end{align}
The above analysis justifies the estimation of the global model update entries as
\begin{align}
&\Delta\hat\btt_{(2k-1)+2(s-1)K}(t)=\frac{\text{Re}\left\{y^{k,s}(t)\right\} }{\alpha_t }\\
&\Delta\hat\btt_{2k+2(s-1)K}(t)=\frac{\text{Im}\left\{y^{k,s}(t)\right\} }{\alpha_t }.
\end{align}

\begin{remark}
We have shown in this section that the interference terms vanish and determined the required conditions for that. These conditions turn out to be completely different than the case of cellular massive MIMO studied in \cite{amiri2021blind}. Specifically, both the number of clients and the total number of antennas should be large while the conditions in the case of a cellular massive MIMO are such that the number of antennas is large while the number of clients is finite. 
\end{remark}

\section{Convergence analysis}
\label{conv}
In this section, we study the convergence of the proposed FEEL scheme. We assume without loss of generality that $d=2S$, which implies $K=1$ and thus we drop the dependency on $k$ of all variables for ease of the notations.
\subsection{Preliminaries}

Let $\btt^*$ be the optimal parameter vector that minimizes $F(\btt)$ and define 
$F^* = F(\btt^*)$. We define also $\Gamma = F^* -\sum_{n=1}^N \frac{D_n}{D}F_n^*$, which measures the amount of bias in the data distribution. Larger $\Gamma$ indicates higher data heterogeneity among clients. In the sequel, we consider the following assumptions, which are classical when it comes to studying the convergence of FL \cite{amiri2021blind,Sery2020a}.
\begin{assumption}
We assume the following regularities for the loss functions:
\begin{itemize}
\item[(i)] The loss functions $F_1,\cdots,F_N$ are $\mu$-strongly convex, that is, for all $\bx,\by \in \mathbb{R}^d$,
$$
F_n(\bx)-F_n(\by) \geq\langle F_n'(\by), \bx-\by\rangle+\frac{\mu}{2}\|\bx-\by\|^2.
$$
\item[(ii)] The loss functions $F_1,\cdots,F_N$ are all $\mathcal{L}$-smooth, that is, for all $\bx,\by \in \mathbb{R}^d$,
$$
F_n(\bx)-F_n(\by) \leq\langle F_n'(\by), \bx-\by\rangle+\frac{\mathcal{L}}{2}\|\bx-\by\|^2.
$$
\item[(iii)] The stochastic gradient $F'_{n,xi_{n,i}^t}(\bx)$ satisfies
$$
\mathbb{E}\|F_{n,\xi_{n,i}^t}'(\bx)\|^2\leq G^2, \  {\rm for   \ all} \  \bx \in \mathbb{R}^d.
$$
\end{itemize}
\label{assump1}
\end{assumption}

\subsection{Convergence results}
In this section, we establish the convergence rate of the proposed OTA-FL scheme. The proof is provided in Appendix A.
\begin{theorem}
Under Assumption \ref{assump1}, when the step size $\eta_t$ verifies $0<\eta_t \leq\min\left\{1,\frac{1}{\mu\tau}\right\}$, we have 
\begin{align}
\left\|\btt(T) -\btt^*\right\|^2 \leq \left( \prod_{t=0}^{T-1}A(t)\right)\left\|\btt(0) -\btt^*\right\|^2 + \sum_{t=0}^{T-1}B(t) \prod_{i=t+1}^{T-1}A(i),
\label{res11}
\end{align}
where
\begin{align}
A(t) &= (1-\mu\eta_t(\tau-\eta_t(\tau-1))),\\
B(t) &= 2\eta_t(\tau-1)\Gamma+(1+\mu(1-\eta_t))\eta_t^2G^2\frac{\tau(\tau-1)(2\tau-1)}{6}+\eta_t^2(\tau^2+\tau-1)G^2\nonumber\\& + \left(\frac{2}{N} +\frac{(\tilde\gamma+\gamma)}{2N} -\frac{\gamma}{2N^2} \right)\eta_t^2\tau^2G^2+ \frac{d}{2N\alpha_t^2}(\kappa +\tilde\kappa)\sigma_z^2,
\end{align}
with
\begin{align*}
\gamma &= \max_{n,n'}\frac{1}{c_n^2}\tr(\bR_n\bD_{n} \bR_{n'} \bD_{n})\\
\tilde\gamma &= \max_{n,n'}\frac{1}{c_n^2}\tr(\bC_n\bD_{n} \bR_{n'} \bD_{n})\\
\kappa &=  \frac{1}{N}\sum_{n=1}^N \frac{1}{c_n}\\
\tilde\kappa &= \frac{1}{N}\sum_{n=1}^N\frac{1}{c_n^2}  \tr(\bC_n\bD_n).
\end{align*}
\end{theorem}

The term $\prod_{t=0}^{T-1}A(t)$ in the bound in \eqref{res11} measures the decay of the distance to the optimal solution after $T$ rounds with respect to the initial distance. On the other hand, the term $\sum_{t=0}^{T-1}B(t) \prod_{i=t+1}^{T-1}A(i)$ measures the residual distance to the optimal solution at round $T$ with respect to that of the initial model to the optimal solution. Reducing $A(t)$ may lead to faster convergence but it may result in increasing the value of $B(t)$ and thus higher residual distance. For instance, increasing $\tau$ will result in a lower value of $A(t)$ and a higher value of $B(t)$, thus higher values of $\tau$ provide faster convergence but increased residual error. Moreover, the effect of the wireless channel appears in the last two terms of $B(t)$ via the noise variance $\sigma_z^2$ and the channel statistics captured by the parameters $\kappa$, $\tilde\kappa$, $\gamma$, and $\tilde\gamma$.                                  

We highlight that the effect of the channel estimation error appears in the parameters $\tilde\gamma$ and  $\tilde\kappa$. Specifically, these parameters vanish when perfect CSI is available at the APs since $\bC_n=\bzero$ in this case. Moreover, the parameters $\gamma$ and $\kappa$ are directly related to the quality of the wireless channels between the clients and their corresponding APs. For instance, $\kappa$  which can be rewritten as
\begin{align}
\kappa = \frac{1}{N}\sum_{n=1}^N  \frac{1}{\tr(\bD_n \bR_n) } =   \frac{1}{N}\sum_{n=1}^N  \frac{1}{\sum_{\ell \in \mathcal{S}_n}\tr(\bR_{n,\ell}) },
\label{kappa_rewritten}
\end{align}
increases when the quality of channel conditions of the clients decreases. We note also that the effect of the AP-UT association scheme is reflected in the parameters $\gamma$, $\tilde\gamma$, $\kappa$ and $\tilde{\kappa}$ through the association matrices $\bD_n$.

\begin{cor}
\begin{itemize}
\item When the step size $\eta_t$ verifies $0<\eta_t \leq\min\left\{1,\frac{1}{\mu\tau}\right\}$, we have
$$
\ex\left[F(\btt(T))\right] -F(\btt^*) \leq \frac{\mathcal{L}}{2}\left( \prod_{t=0}^{T-1}A(t)\right)\left\|\btt(0) -\btt^*\right\|^2 + \frac{\mathcal{L}}{2} \sum_{t=0}^{T-1}B(t) \prod_{i=t+1}^{T-1}A(i).
$$
\item When $\eta_t=\eta, \ \forall t$ and $\tau = 1$, we have for $0<\eta \leq\min\left\{1,\frac{1}{\mu\tau}\right\}$,
$$
\ex\left[F(\btt(T))\right] -F(\btt^*) \leq \frac{\mathcal{L}}{2}\left(1-\mu\eta\right)^T \left(\left\|\btt(0) -\btt^*\right\|^2-A_1 \right)+  \frac{\mathcal{L}}{2} A_1,
$$
where
\begin{align*}
A_1 &= \frac{1}{\mu\eta} \left(\eta^2G^2 + \left(\frac{2}{N} +\frac{(\tilde\gamma+\gamma)}{2N} -\frac{\gamma}{2N^2} \right)\eta^2 G^2+ \frac{d(\kappa +\tilde\kappa)}{2N\alpha_0^2}\sigma_z^2\right).
\label{asymp_error}
\end{align*}
\end{itemize}
\label{cor2}
\end{cor}

Corollary \ref{cor2} states that the proposed implementation converges at a linear rate towards a neighborhood of the optimal solution with an asymptotic learning error, $A_1$, that depends on the noise variance, the UTs channel conditions, and the AP-UT association scheme. The latter is reflected in the parameters $\gamma$, $\tilde\gamma$, $\kappa$ and $\tilde{\kappa}$. The above result can be leveraged to reduce the learning error of the FL process, also known as the optimality gap, by optimizing the AP-UT association scheme. Reducing the values of $\gamma+\tilde{\gamma}$ and $\kappa+\tilde{\kappa}$ will result in a lower learning error.
\begin{remark}
Various techniques, including user scheduling \cite{li2019convergence}, partial device participation \cite{hu2021device}, and gradient sparsification or compression \cite{han2020adaptive}, among others, can be employed in the context of cell-free massive MIMO for optimizing the FL process. The convergence analysis can be adapted depending on the specific technique but this is out of the scope of this work.
\end{remark}

\section{Numerical experiments}
\label{num}

In this section, we study the convergence of the proposed OTA-FL scheme experimentally and compare the performance of OTA-FL over cell-free and cellular massive MIMO. We consider two popular classification datasets, MNIST and CIFAR10. For both datasets, convolutional neural network models are deployed as detailed in Table \ref{tab1}. The performance is measured in terms of the test accuracy computed using the test dataset.

\begin{table}[h]
\caption{CNN models for classification on MNIST and CIFAR10.}
\begin{center}
\begin{tabular}{|c  | c |} 
 \hline
 MNIST  & CIFAR10  \\ 
\Xhline{3\arrayrulewidth}
 \multirow{3}{*}{\makecell{$5 \times 5$ conv layer, 32 channels,\\ ReLU activation, same padding}}
 & \makecell{$3 \times 3$ conv layer, 32 channels, ReLU activation, same padding } \\ \cline{2-2} 
  & \makecell{$3 \times 3$ conv layer, 32 channels, ReLU activation, same padding }\\ \cline{2-2} 
   & \makecell{$2 \times 2$ max pooling  }
 \\ 
 \hline
 \multirow{2}{*}{ \makecell{$2 \times 2$ max pooling}  } & \makecell{  dropout with probability 0.2}\\ \cline{2-2} 
 & \makecell{$3 \times 3$ conv layer, 64 channels, ReLU activation, same padding }\\
  \hline
 \multirow{3}{*}{  \makecell{$5\times 5$ conv layer, 64 channels,\\ ReLU activation, same padding  } } & \makecell{$3 \times 3$ conv layer, 64 channels, ReLU activation, same padding }\\ \cline{2-2} 
  & \makecell{ $2 \times 2$ max pooling  }\\ \cline{2-2} 
   & \makecell{  dropout with probability 0.3 } \\ \hline
  \multirow{2}{*}{ \makecell{  $2 \times 2$ max pooling}} & \makecell{$3 \times 3$ conv layer, 128 channels, ReLU activation, same padding }\\ \cline{2-2} 
& \makecell{$3 \times 3$ conv layer, 128 channels, ReLU activation, same padding }
 \\ \hline
   \multirow{2}{*}{  \makecell{  fully connected layer with 128 units,\\ ReLU activation } } &  \makecell{ $2 \times 2$ max pooling  }\\ \cline{2-2} 
     & \makecell{ dropout with probability 0.4} 
    \\ \hline
10 units softmax output layer & 10 units softmax output layer
    \\ \hline
\end{tabular}
\end{center}
\label{tab1}
\end{table}
We consider both i.i.d. and non-i.i.d. data distribution scenarios. For the i.i.d. data distribution, the training set is randomly splitted into $N$ disjoint datasets, and each of them is assigned to one client. As for the non-i.i.d. data distribution, the samples from the same class are splitted into $N/10$ disjoint groups (assuming $N \geq 10$ is divisible by 10) and each group is assigned to a client. Thus, each client has samples from one class only.

In all experiments, we assume that the number of OFDM subcarriers is $S=d/2$, and thus $K=1$ transmission time slot is needed at every global training round. We note that this assumption has no effect on the learning performance since changing the values of $S$ and $K$ will only change the value of the average power constraint $P$. We adopt the cell-free massive MIMO architecture proposed in \cite{Emilcellfree}, including the algorithms proposed for pilot assignment,  UT-AP association, and channel estimation. The values of the channel parameters are summarized in Table \ref{tab2}. 

In all experiments, the convergence is reported for $T=300$ global iterations. For the MNIST dataset, the learning rate is fixed to $\eta_t= 0.01$ and the batch size to $|\xi_{n,i}^t|=50$ while the learning rate is set to $\eta_t= 0.005$ and the batch size to $|\xi_{n,i}^t|=500$ for the CIFAR10 dataset. For the sake of comparison, we plot the performance of the FL process when perfect communication links between the clients and the CS are available. In such a scenario, the CS receives $\Delta\btt(t) = \frac{1}{N}\sum_{n=1}^N\Delta\btt_n(t) $ without distortions. This will serve as a benchmark to assess the proposed OTA-FL scheme. Moreover, in all experiments, we assume that the clients use single-antenna wireless devices.

\begin{table}[h]
\caption{Channel parameters.}
\begin{center}
}
\end{center}
\caption{Test accuracy vs. the number of global iterations using i.i.d. MNIST dataset with $L=100$, $M=4$, and $\tau=12$.}
\label{classMNIST_correlated}
\end{figure}

\begin{figure}[h]
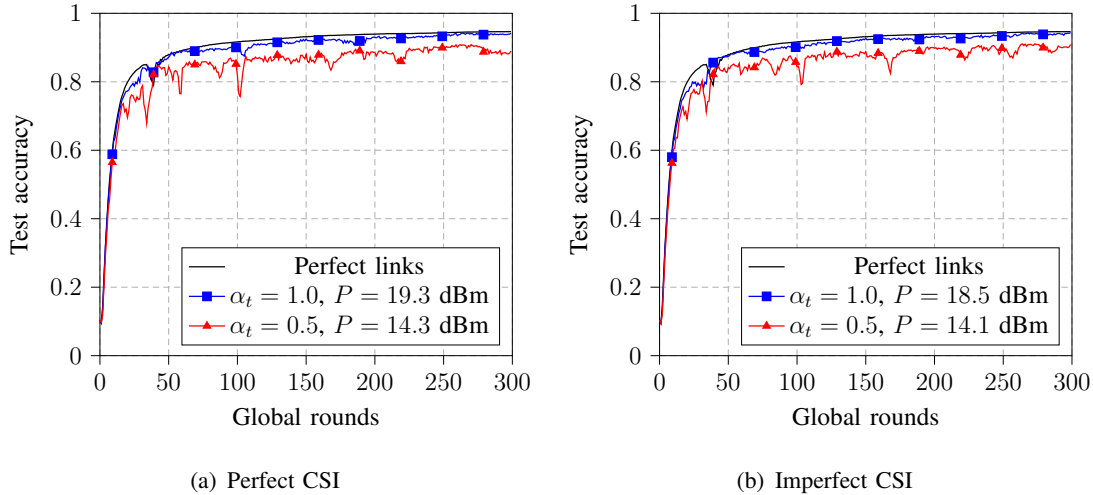

\begin{center}
\subfigure[Perfect CSI]{
%
}
\end{center}
\caption{Test accuracy vs. the number of global iterations using non-i.i.d. MNIST dataset with $\tau=12$.}
\label{noniidMNIST}
\end{figure}

\subsection{Performance on MNIST dataset}
In \figref{classMNIST_correlated}, we plot the performance of the proposed implementation based on the MNIST i.i.d. dataset for different values of $N$. In this experiment, the number of APs is $L=100$ each equipped with $M=4$ antennas and randomly distributed in a square area of size $2\times 2$ km. The performance of both perfect and imperfect CSI scenarios is reported. We note that the performance of the OTA-FL scheme approaches the performance of error-free links for small values of the power constraint  $P$, such as $P=10$ dBm. This illustrates the success of the proposed OTA-FL over cell-free massive MIMO and highlights the advantages that the latter can bring in terms of energy efficiency. We remark also that the choice of $\alpha_t$ is very important to ensure convergence as sufficiently small values of $\alpha_t$ lead to poor performance.

Considering the same settings as in \figref{classMNIST_correlated}, we report in \figref{noniidMNIST} the performance of the proposed OTA-FL scheme on the non-i.i.d. MNIST dataset. As expected, the convergence is guaranteed for sufficient values of the power constraint and approaches that of error-free links. The gap in performance compared to \figref{classMNIST_correlated} is due to the non-i.i.d. data distributions at the clients. This was predicted by the theoretical analysis where the learning error depends on the heterogeneity of the data.

\subsection{Comparison with cellular massive MIMO}

In the third experiment, we compare the performance of OTA-FL on cell-free massive MIMO and cellular massive MIMO. We consider $N=25$ clients uniformly distributed in a $500\times 500$m square coverage area. For cell-free massive MIMO, we consider $L=25$ APs, each equipped with $M=4$ antennas, and randomly distributed in the coverage area. As for the cellular massive MIMO, the clients communicate with a single BS placed in the center of the square coverage area and equipped with $\tilde M=100$ antennas. We note that for both systems, the number of antennas is the same to ensure a fair comparison. We plot in \figref{comp} the performance of both systems based on i.i.d. MNIST dataset for the same value of the power scaling parameter $\alpha_t$. From the figure, the performance of cell-free massive MIMO-based OTA-FL is better for a lower value of the power constraint. For example, with $P=4.47$ dBm, cell-free massive MIMO-based OTA-FL provides better performance than cellular OTA-FL with $P=13.6$ dBm. This highlights the benefits of cell-free massive MIMO in terms of energy efficiency. We note also from \figref{comp}.b, that for the same value of the power scaling factor, cellular massive MIMO fails to reach convergence while cell-free massive MIMO guarantees an acceptable performance.

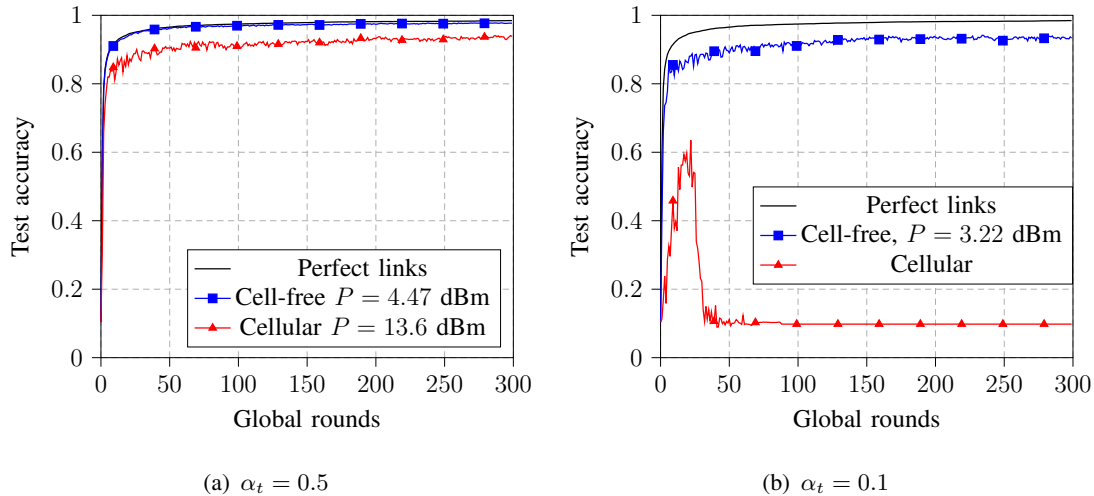
\begin{figure}[h]
\begin{center}
\subfigure[$\alpha_t  = 0.5$]{
\begin{tikzpicture}[scale=0.8]
\begin{axis}[
tick align=outside,
tick pos=left,
x grid style={white!69.0196078431373!black},
xlabel={Global rounds},
xmajorgrids,
xmin=0, xmax=300,
xtick style={color=black},
y grid style={white!69.0196078431373!black},
ylabel={Test accuracy},
ymajorgrids,
ymin=0, ymax=1,
ytick style={color=black},
grid=major,
scaled ticks=true,
legend pos=south east,	
grid style=densely dashed,
]

\addplot [semithick, color=black]
coordinates {
(0,0.104)(1,0.587)(2,0.7894)(3,0.8437)(4,0.8683)(5,0.8856)(6,0.8956)(7,0.9027)(8,0.9099)(9,0.9145)(10,0.9196)(11,0.9236)(12,0.9272)(13,0.9312)(14,0.9334)(15,0.9354)(16,0.9378)(17,0.9396)(18,0.9409)(19,0.9431)(20,0.9455)(21,0.9471)(22,0.9478)(23,0.9489)(24,0.9498)(25,0.9507)(26,0.9517)(27,0.9522)(28,0.9528)(29,0.9538)(30,0.9549)(31,0.9557)(32,0.9563)(33,0.9571)(34,0.9579)(35,0.9588)(36,0.9595)(37,0.9604)(38,0.9606)(39,0.9612)(40,0.9621)(41,0.9624)(42,0.9628)(43,0.9633)(44,0.9633)(45,0.9636)(46,0.9639)(47,0.964)(48,0.965)(49,0.965)(50,0.9654)(51,0.9658)(52,0.9661)(53,0.9668)(54,0.9672)(55,0.9676)(56,0.9681)(57,0.9684)(58,0.9684)(59,0.9684)(60,0.9686)(61,0.9692)(62,0.9694)(63,0.9696)(64,0.9698)(65,0.9704)(66,0.9706)(67,0.9705)(68,0.9707)(69,0.9709)(70,0.9709)(71,0.9709)(72,0.9709)(73,0.9709)(74,0.9709)(75,0.9713)(76,0.9715)(77,0.9718)(78,0.972)(79,0.9718)(80,0.972)(81,0.9721)(82,0.9725)(83,0.9727)(84,0.9728)(85,0.9731)(86,0.9731)(87,0.9732)(88,0.9733)(89,0.9734)(90,0.9736)(91,0.9737)(92,0.9737)(93,0.9737)(94,0.974)(95,0.9739)(96,0.974)(97,0.9743)(98,0.9744)(99,0.9744)(100,0.9743)(101,0.9747)(102,0.9747)(103,0.9747)(104,0.9749)(105,0.9751)(106,0.9752)(107,0.9752)(108,0.9754)(109,0.9753)(110,0.9754)(111,0.9752)(112,0.9754)(113,0.9757)(114,0.976)(115,0.976)(116,0.9762)(117,0.9764)(118,0.9765)(119,0.9765)(120,0.9765)(121,0.9764)(122,0.9765)(123,0.9767)(124,0.9768)(125,0.977)(126,0.9771)(127,0.9773)(128,0.9775)(129,0.9775)(130,0.9775)(131,0.9775)(132,0.9775)(133,0.9775)(134,0.9776)(135,0.9776)(136,0.9776)(137,0.9777)(138,0.9779)(139,0.978)(140,0.978)(141,0.9781)(142,0.9781)(143,0.9782)(144,0.9782)(145,0.9784)(146,0.9783)(147,0.9786)(148,0.9787)(149,0.9787)(150,0.9787)(151,0.9788)(152,0.9789)(153,0.9789)(154,0.979)(155,0.979)(156,0.9791)(157,0.9793)(158,0.9793)(159,0.9794)(160,0.9796)(161,0.9798)(162,0.9799)(163,0.9799)(164,0.98)(165,0.98)(166,0.9802)(167,0.9802)(168,0.9803)(169,0.9804)(170,0.9804)(171,0.9805)(172,0.9806)(173,0.9807)(174,0.9808)(175,0.9808)(176,0.9807)(177,0.9808)(178,0.9807)(179,0.9807)(180,0.9807)(181,0.9807)(182,0.9807)(183,0.981)(184,0.981)(185,0.981)(186,0.981)(187,0.981)(188,0.9812)(189,0.9812)(190,0.9812)(191,0.9813)(192,0.9813)(193,0.9814)(194,0.9815)(195,0.9815)(196,0.9815)(197,0.9815)(198,0.9816)(199,0.9815)(200,0.9815)(201,0.9816)(202,0.9816)(203,0.9818)(204,0.9818)(205,0.9819)(206,0.9819)(207,0.9819)(208,0.9819)(209,0.9819)(210,0.9819)(211,0.9819)(212,0.9818)(213,0.9818)(214,0.9819)(215,0.9819)(216,0.982)(217,0.982)(218,0.982)(219,0.982)(220,0.982)(221,0.982)(222,0.982)(223,0.9822)(224,0.9822)(225,0.9822)(226,0.9822)(227,0.9822)(228,0.9823)(229,0.9823)(230,0.9823)(231,0.9823)(232,0.9823)(233,0.9823)(234,0.9822)(235,0.9822)(236,0.9823)(237,0.9823)(238,0.9823)(239,0.9823)(240,0.9823)(241,0.9823)(242,0.9823)(243,0.9823)(244,0.9823)(245,0.9823)(246,0.9825)(247,0.9825)(248,0.9827)(249,0.9827)(250,0.9827)(251,0.9828)(252,0.9829)(253,0.9829)(254,0.9829)(255,0.9829)(256,0.9828)(257,0.9828)(258,0.9828)(259,0.9828)(260,0.9828)(261,0.9828)(262,0.9828)(263,0.9829)(264,0.9829)(265,0.983)(266,0.983)(267,0.983)(268,0.9831)(269,0.9831)(270,0.9831)(271,0.9831)(272,0.9831)(273,0.9833)(274,0.9834)(275,0.9834)(276,0.9838)(277,0.9838)(278,0.9838)(279,0.9838)(280,0.9838)(281,0.9839)(282,0.984)(283,0.9841)(284,0.9841)(285,0.9841)(286,0.9841)(287,0.984)(288,0.9841)(289,0.9841)(290,0.9842)(291,0.9842)(292,0.9843)(293,0.9843)(294,0.9843)(295,0.9844)(296,0.9846)(297,0.9846)(298,0.9846)(299,0.9846)}; \addlegendentry{Perfect links}

\addplot [semithick, color =blue, mark =  square*,  mark size = 2, mark repeat = 30, mark phase = 10]
coordinates {
(0,0.104)(1,0.568)(2,0.784)(3,0.8367)(4,0.8634)(5,0.88)(6,0.8934)(7,0.8993)(8,0.9059)(9,0.9104)(10,0.9164)(11,0.919)(12,0.9245)(13,0.9247)(14,0.9294)(15,0.9307)(16,0.9313)(17,0.9327)(18,0.936)(19,0.9371)(20,0.9404)(21,0.943)(22,0.9429)(23,0.9455)(24,0.9482)(25,0.949)(26,0.9495)(27,0.9468)(28,0.9509)(29,0.9512)(30,0.953)(31,0.9548)(32,0.9545)(33,0.9548)(34,0.9554)(35,0.9543)(36,0.956)(37,0.958)(38,0.9565)(39,0.9586)(40,0.9574)(41,0.9598)(42,0.9594)(43,0.9586)(44,0.9584)(45,0.9609)(46,0.9604)(47,0.9615)(48,0.9589)(49,0.9609)(50,0.9639)(51,0.9622)(52,0.964)(53,0.9643)(54,0.9654)(55,0.964)(56,0.9629)(57,0.9658)(58,0.9637)(59,0.9649)(60,0.9653)(61,0.9659)(62,0.9641)(63,0.9659)(64,0.9666)(65,0.9673)(66,0.9658)(67,0.9663)(68,0.9671)(69,0.9666)(70,0.9669)(71,0.9671)(72,0.9671)(73,0.9656)(74,0.9674)(75,0.9667)(76,0.9657)(77,0.967)(78,0.9685)(79,0.9689)(80,0.9684)(81,0.9671)(82,0.9675)(83,0.9675)(84,0.9677)(85,0.9693)(86,0.967)(87,0.9684)(88,0.9692)(89,0.9664)(90,0.9684)(91,0.9675)(92,0.9675)(93,0.9677)(94,0.9678)(95,0.9684)(96,0.9682)(97,0.969)(98,0.9694)(99,0.9696)(100,0.9692)(101,0.9696)(102,0.9691)(103,0.969)(104,0.9706)(105,0.97)(106,0.9708)(107,0.9708)(108,0.9685)(109,0.9698)(110,0.9701)(111,0.9716)(112,0.9695)(113,0.9682)(114,0.9708)(115,0.9716)(116,0.9701)(117,0.9705)(118,0.9711)(119,0.9692)(120,0.9716)(121,0.9705)(122,0.97)(123,0.971)(124,0.971)(125,0.9726)(126,0.9716)(127,0.9732)(128,0.9716)(129,0.972)(130,0.9725)(131,0.9723)(132,0.9727)(133,0.9718)(134,0.9728)(135,0.9731)(136,0.972)(137,0.9725)(138,0.9725)(139,0.9733)(140,0.9729)(141,0.9715)(142,0.9718)(143,0.9723)(144,0.9722)(145,0.971)(146,0.9724)(147,0.9723)(148,0.9729)(149,0.9727)(150,0.9727)(151,0.9737)(152,0.9735)(153,0.9739)(154,0.9737)(155,0.9726)(156,0.9728)(157,0.9739)(158,0.9726)(159,0.9718)(160,0.9734)(161,0.9729)(162,0.9733)(163,0.9715)(164,0.9726)(165,0.9724)(166,0.973)(167,0.9719)(168,0.9716)(169,0.9735)(170,0.9733)(171,0.9731)(172,0.9735)(173,0.9726)(174,0.9738)(175,0.9733)(176,0.9726)(177,0.9734)(178,0.973)(179,0.9727)(180,0.9742)(181,0.9745)(182,0.9747)(183,0.9744)(184,0.9742)(185,0.9747)(186,0.975)(187,0.9751)(188,0.9751)(189,0.9751)(190,0.9745)(191,0.974)(192,0.9745)(193,0.9745)(194,0.9753)(195,0.9746)(196,0.9751)(197,0.9733)(198,0.9751)(199,0.9747)(200,0.9737)(201,0.9743)(202,0.9753)(203,0.976)(204,0.9768)(205,0.9756)(206,0.9757)(207,0.9753)(208,0.9749)(209,0.9748)(210,0.9765)(211,0.9762)(212,0.9762)(213,0.9745)(214,0.9757)(215,0.9756)(216,0.9758)(217,0.9752)(218,0.9759)(219,0.9759)(220,0.9753)(221,0.9753)(222,0.9759)(223,0.9749)(224,0.9747)(225,0.9748)(226,0.9752)(227,0.9753)(228,0.9759)(229,0.9747)(230,0.9757)(231,0.9739)(232,0.9754)(233,0.9751)(234,0.9764)(235,0.976)(236,0.9755)(237,0.9766)(238,0.9774)(239,0.9766)(240,0.9763)(241,0.976)(242,0.9762)(243,0.9755)(244,0.9765)(245,0.9756)(246,0.9757)(247,0.976)(248,0.9764)(249,0.9756)(250,0.977)(251,0.977)(252,0.9759)(253,0.9753)(254,0.9759)(255,0.9752)(256,0.9762)(257,0.9769)(258,0.976)(259,0.9755)(260,0.9777)(261,0.9761)(262,0.977)(263,0.976)(264,0.9762)(265,0.9767)(266,0.9763)(267,0.9765)(268,0.9768)(269,0.9777)(270,0.9769)(271,0.9774)(272,0.9775)(273,0.9771)(274,0.9772)(275,0.977)(276,0.9769)(277,0.9769)(278,0.9773)(279,0.9767)(280,0.9774)(281,0.9771)(282,0.9766)(283,0.9776)(284,0.9772)(285,0.9771)(286,0.9775)(287,0.9776)(288,0.9773)(289,0.9769)(290,0.9769)(291,0.9779)(292,0.9776)(293,0.9767)(294,0.9763)(295,0.9764)(296,0.9769)(297,0.977)(298,0.9773)(299,0.9766)

}; \addlegendentry{Cell-free $P  = 4.47$ dBm}
\addplot  [semithick, color =red, mark = triangle*, mark size = 2, mark repeat = 30, mark phase = 10]
coordinates {
(0,0.104)(1,0.3142)(2,0.6579)(3,0.7408)(4,0.7834)(5,0.8199)(6,0.8197)(7,0.8397)(8,0.8368)(9,0.8479)(10,0.8116)(11,0.8347)(12,0.8728)(13,0.8546)(14,0.8672)(15,0.8759)(16,0.8344)(17,0.8657)(18,0.8513)(19,0.8823)(20,0.8556)(21,0.8728)(22,0.8657)(23,0.8791)(24,0.8881)(25,0.8882)(26,0.8999)(27,0.8836)(28,0.8722)(29,0.8678)(30,0.8951)(31,0.8958)(32,0.8902)(33,0.8903)(34,0.8882)(35,0.8913)(36,0.8757)(37,0.8961)(38,0.903)(39,0.9028)(40,0.8959)(41,0.8845)(42,0.8873)(43,0.8928)(44,0.8923)(45,0.9047)(46,0.9033)(47,0.9086)(48,0.9033)(49,0.9049)(50,0.9039)(51,0.9073)(52,0.8976)(53,0.9046)(54,0.8954)(55,0.9037)(56,0.9022)(57,0.9006)(58,0.9008)(59,0.9179)(60,0.9156)(61,0.9168)(62,0.9097)(63,0.9166)(64,0.9083)(65,0.9234)(66,0.9148)(67,0.907)(68,0.9202)(69,0.9049)(70,0.9181)(71,0.9229)(72,0.9132)(73,0.9097)(74,0.9219)(75,0.9139)(76,0.9157)(77,0.8984)(78,0.9132)(79,0.9229)(80,0.9202)(81,0.9224)(82,0.9212)(83,0.9059)(84,0.919)(85,0.9208)(86,0.9091)(87,0.9133)(88,0.9136)(89,0.9049)(90,0.9098)(91,0.8997)(92,0.9075)(93,0.9072)(94,0.91)(95,0.9164)(96,0.915)(97,0.9082)(98,0.9044)(99,0.9098)(100,0.9083)(101,0.9044)(102,0.9065)(103,0.9096)(104,0.9155)(105,0.9172)(106,0.9211)(107,0.9113)(108,0.9156)(109,0.9105)(110,0.9133)(111,0.9097)(112,0.9175)(113,0.9124)(114,0.9173)(115,0.922)(116,0.9219)(117,0.9193)(118,0.9235)(119,0.9154)(120,0.9118)(121,0.9208)(122,0.9078)(123,0.9153)(124,0.9175)(125,0.9199)(126,0.9163)(127,0.9144)(128,0.9107)(129,0.9158)(130,0.9154)(131,0.9123)(132,0.9141)(133,0.9177)(134,0.9231)(135,0.9231)(136,0.9198)(137,0.9222)(138,0.9191)(139,0.9241)(140,0.9236)(141,0.925)(142,0.9234)(143,0.914)(144,0.922)(145,0.9215)(146,0.9224)(147,0.9235)(148,0.9206)(149,0.9242)(150,0.9228)(151,0.9286)(152,0.9193)(153,0.9235)(154,0.9247)(155,0.9289)(156,0.9221)(157,0.9186)(158,0.9229)(159,0.9202)(160,0.9252)(161,0.9214)(162,0.9218)(163,0.9121)(164,0.9145)(165,0.9132)(166,0.9226)(167,0.9201)(168,0.9231)(169,0.9207)(170,0.9266)(171,0.9217)(172,0.9291)(173,0.9245)(174,0.9225)(175,0.9201)(176,0.9243)(177,0.9216)(178,0.9215)(179,0.9305)(180,0.9253)(181,0.9194)(182,0.9237)(183,0.9169)(184,0.9312)(185,0.9248)(186,0.9245)(187,0.9279)(188,0.9311)(189,0.933)(190,0.9308)(191,0.9258)(192,0.9254)(193,0.929)(194,0.9352)(195,0.9267)(196,0.9288)(197,0.9285)(198,0.9321)(199,0.9326)(200,0.9324)(201,0.9348)(202,0.9373)(203,0.9371)(204,0.9308)(205,0.9337)(206,0.931)(207,0.936)(208,0.9322)(209,0.9289)(210,0.9306)(211,0.9238)(212,0.9296)(213,0.932)(214,0.9316)(215,0.9329)(216,0.9287)(217,0.9193)(218,0.9283)(219,0.9275)(220,0.9337)(221,0.9311)(222,0.9307)(223,0.9304)(224,0.9292)(225,0.93)(226,0.9297)(227,0.93)(228,0.932)(229,0.9303)(230,0.9328)(231,0.9315)(232,0.9238)(233,0.9278)(234,0.9268)(235,0.9309)(236,0.9295)(237,0.9243)(238,0.9321)(239,0.9316)(240,0.9271)(241,0.9307)(242,0.9253)(243,0.9277)(244,0.9293)(245,0.9321)(246,0.9332)(247,0.9294)(248,0.9306)(249,0.9291)(250,0.9364)(251,0.9386)(252,0.9346)(253,0.9371)(254,0.9345)(255,0.933)(256,0.9341)(257,0.9346)(258,0.9347)(259,0.9308)(260,0.9289)(261,0.928)(262,0.9378)(263,0.9325)(264,0.9347)(265,0.9331)(266,0.9347)(267,0.9337)(268,0.9288)(269,0.9334)(270,0.9338)(271,0.9331)(272,0.9299)(273,0.9321)(274,0.9348)(275,0.9372)(276,0.9317)(277,0.9342)(278,0.9338)(279,0.9367)(280,0.94)(281,0.9381)(282,0.9339)(283,0.935)(284,0.9357)(285,0.9351)(286,0.94)(287,0.9409)(288,0.9375)(289,0.9401)(290,0.9363)(291,0.938)(292,0.9374)(293,0.9368)(294,0.9305)(295,0.9291)(296,0.9336)(297,0.9374)(298,0.9409)(299,0.9383)

}; \addlegendentry{Cellular $P  = 13.6$ dBm}

\end{axis}
\end{tikzpicture}}
\subfigure[$\alpha_t  = 0.1$]{
\begin{tikzpicture}[scale=0.8]
\begin{axis}[
tick align=outside,
tick pos=left,
x grid style={white!69.0196078431373!black},
xlabel={Global rounds},
xmajorgrids,
xmin=0, xmax=300,
xtick style={color=black},
y grid style={white!69.0196078431373!black},
ylabel={Test accuracy},
ymajorgrids,
ymin=0, ymax=1,
ytick style={color=black},
grid=major,
scaled ticks=true,
legend style={at={(1,0.5)},anchor=north east},	
grid style=densely dashed,
]

\addplot [semithick, color=black]
coordinates {
(0,0.104)(1,0.587)(2,0.7894)(3,0.8437)(4,0.8683)(5,0.8856)(6,0.8956)(7,0.9027)(8,0.9099)(9,0.9145)(10,0.9196)(11,0.9236)(12,0.9272)(13,0.9312)(14,0.9334)(15,0.9354)(16,0.9378)(17,0.9396)(18,0.9409)(19,0.9431)(20,0.9455)(21,0.9471)(22,0.9478)(23,0.9489)(24,0.9498)(25,0.9507)(26,0.9517)(27,0.9522)(28,0.9528)(29,0.9538)(30,0.9549)(31,0.9557)(32,0.9563)(33,0.9571)(34,0.9579)(35,0.9588)(36,0.9595)(37,0.9604)(38,0.9606)(39,0.9612)(40,0.9621)(41,0.9624)(42,0.9628)(43,0.9633)(44,0.9633)(45,0.9636)(46,0.9639)(47,0.964)(48,0.965)(49,0.965)(50,0.9654)(51,0.9658)(52,0.9661)(53,0.9668)(54,0.9672)(55,0.9676)(56,0.9681)(57,0.9684)(58,0.9684)(59,0.9684)(60,0.9686)(61,0.9692)(62,0.9694)(63,0.9696)(64,0.9698)(65,0.9704)(66,0.9706)(67,0.9705)(68,0.9707)(69,0.9709)(70,0.9709)(71,0.9709)(72,0.9709)(73,0.9709)(74,0.9709)(75,0.9713)(76,0.9715)(77,0.9718)(78,0.972)(79,0.9718)(80,0.972)(81,0.9721)(82,0.9725)(83,0.9727)(84,0.9728)(85,0.9731)(86,0.9731)(87,0.9732)(88,0.9733)(89,0.9734)(90,0.9736)(91,0.9737)(92,0.9737)(93,0.9737)(94,0.974)(95,0.9739)(96,0.974)(97,0.9743)(98,0.9744)(99,0.9744)(100,0.9743)(101,0.9747)(102,0.9747)(103,0.9747)(104,0.9749)(105,0.9751)(106,0.9752)(107,0.9752)(108,0.9754)(109,0.9753)(110,0.9754)(111,0.9752)(112,0.9754)(113,0.9757)(114,0.976)(115,0.976)(116,0.9762)(117,0.9764)(118,0.9765)(119,0.9765)(120,0.9765)(121,0.9764)(122,0.9765)(123,0.9767)(124,0.9768)(125,0.977)(126,0.9771)(127,0.9773)(128,0.9775)(129,0.9775)(130,0.9775)(131,0.9775)(132,0.9775)(133,0.9775)(134,0.9776)(135,0.9776)(136,0.9776)(137,0.9777)(138,0.9779)(139,0.978)(140,0.978)(141,0.9781)(142,0.9781)(143,0.9782)(144,0.9782)(145,0.9784)(146,0.9783)(147,0.9786)(148,0.9787)(149,0.9787)(150,0.9787)(151,0.9788)(152,0.9789)(153,0.9789)(154,0.979)(155,0.979)(156,0.9791)(157,0.9793)(158,0.9793)(159,0.9794)(160,0.9796)(161,0.9798)(162,0.9799)(163,0.9799)(164,0.98)(165,0.98)(166,0.9802)(167,0.9802)(168,0.9803)(169,0.9804)(170,0.9804)(171,0.9805)(172,0.9806)(173,0.9807)(174,0.9808)(175,0.9808)(176,0.9807)(177,0.9808)(178,0.9807)(179,0.9807)(180,0.9807)(181,0.9807)(182,0.9807)(183,0.981)(184,0.981)(185,0.981)(186,0.981)(187,0.981)(188,0.9812)(189,0.9812)(190,0.9812)(191,0.9813)(192,0.9813)(193,0.9814)(194,0.9815)(195,0.9815)(196,0.9815)(197,0.9815)(198,0.9816)(199,0.9815)(200,0.9815)(201,0.9816)(202,0.9816)(203,0.9818)(204,0.9818)(205,0.9819)(206,0.9819)(207,0.9819)(208,0.9819)(209,0.9819)(210,0.9819)(211,0.9819)(212,0.9818)(213,0.9818)(214,0.9819)(215,0.9819)(216,0.982)(217,0.982)(218,0.982)(219,0.982)(220,0.982)(221,0.982)(222,0.982)(223,0.9822)(224,0.9822)(225,0.9822)(226,0.9822)(227,0.9822)(228,0.9823)(229,0.9823)(230,0.9823)(231,0.9823)(232,0.9823)(233,0.9823)(234,0.9822)(235,0.9822)(236,0.9823)(237,0.9823)(238,0.9823)(239,0.9823)(240,0.9823)(241,0.9823)(242,0.9823)(243,0.9823)(244,0.9823)(245,0.9823)(246,0.9825)(247,0.9825)(248,0.9827)(249,0.9827)(250,0.9827)(251,0.9828)(252,0.9829)(253,0.9829)(254,0.9829)(255,0.9829)(256,0.9828)(257,0.9828)(258,0.9828)(259,0.9828)(260,0.9828)(261,0.9828)(262,0.9828)(263,0.9829)(264,0.9829)(265,0.983)(266,0.983)(267,0.983)(268,0.9831)(269,0.9831)(270,0.9831)(271,0.9831)(272,0.9831)(273,0.9833)(274,0.9834)(275,0.9834)(276,0.9838)(277,0.9838)(278,0.9838)(279,0.9838)(280,0.9838)(281,0.9839)(282,0.984)(283,0.9841)(284,0.9841)(285,0.9841)(286,0.9841)(287,0.984)(288,0.9841)(289,0.9841)(290,0.9842)(291,0.9842)(292,0.9843)(293,0.9843)(294,0.9843)(295,0.9844)(296,0.9846)(297,0.9846)(298,0.9846)(299,0.9846)}; \addlegendentry{Perfect links}

\addplot [semithick, color =blue, mark =  square*,  mark size = 2, mark repeat = 30, mark phase = 10]
coordinates {
(0,0.104)(1,0.4071)(2,0.6612)(3,0.7367)(4,0.7451)(5,0.7787)(6,0.8561)(7,0.8323)(8,0.8346)(9,0.8551)(10,0.8242)(11,0.8586)(12,0.8556)(13,0.8377)(14,0.8662)(15,0.8585)(16,0.8418)(17,0.8461)(18,0.8685)(19,0.8638)(20,0.8613)(21,0.8804)(22,0.8807)(23,0.8874)(24,0.851)(25,0.884)(26,0.886)(27,0.8591)(28,0.8765)(29,0.8855)(30,0.8864)(31,0.8859)(32,0.8766)(33,0.8897)(34,0.8883)(35,0.8695)(36,0.8699)(37,0.8817)(38,0.8865)(39,0.8951)(40,0.8934)(41,0.9044)(42,0.9022)(43,0.8974)(44,0.8776)(45,0.8789)(46,0.8987)(47,0.9078)(48,0.8977)(49,0.897)(50,0.8902)(51,0.9005)(52,0.9072)(53,0.8947)(54,0.896)(55,0.9109)(56,0.9112)(57,0.9067)(58,0.9132)(59,0.8935)(60,0.9075)(61,0.9093)(62,0.9033)(63,0.8949)(64,0.9055)(65,0.8979)(66,0.9066)(67,0.8971)(68,0.9077)(69,0.8953)(70,0.898)(71,0.9009)(72,0.9044)(73,0.9123)(74,0.9164)(75,0.9078)(76,0.9139)(77,0.9118)(78,0.913)(79,0.9184)(80,0.9078)(81,0.9131)(82,0.9082)(83,0.9131)(84,0.9144)(85,0.9129)(86,0.908)(87,0.9093)(88,0.9167)(89,0.8953)(90,0.9216)(91,0.9142)(92,0.92)(93,0.9203)(94,0.9179)(95,0.9118)(96,0.9159)(97,0.9134)(98,0.9035)(99,0.9106)(100,0.9135)(101,0.9164)(102,0.9121)(103,0.9204)(104,0.9188)(105,0.9172)(106,0.9182)(107,0.9168)(108,0.9113)(109,0.917)(110,0.918)(111,0.9162)(112,0.9263)(113,0.923)(114,0.9245)(115,0.9247)(116,0.9221)(117,0.9133)(118,0.919)(119,0.9239)(120,0.9268)(121,0.9201)(122,0.9207)(123,0.9235)(124,0.9194)(125,0.923)(126,0.9148)(127,0.9218)(128,0.9211)(129,0.9282)(130,0.9248)(131,0.9256)(132,0.9354)(133,0.9289)(134,0.9311)(135,0.9299)(136,0.9228)(137,0.9263)(138,0.924)(139,0.9316)(140,0.9273)(141,0.9233)(142,0.9281)(143,0.9321)(144,0.93)(145,0.9308)(146,0.9326)(147,0.9308)(148,0.9348)(149,0.928)(150,0.9317)(151,0.9294)(152,0.9361)(153,0.9326)(154,0.9341)(155,0.933)(156,0.9356)(157,0.9332)(158,0.9292)(159,0.9292)(160,0.9313)(161,0.9298)(162,0.9333)(163,0.9344)(164,0.9348)(165,0.9277)(166,0.9383)(167,0.9354)(168,0.9347)(169,0.9381)(170,0.9342)(171,0.9369)(172,0.9357)(173,0.931)(174,0.9343)(175,0.9224)(176,0.9309)(177,0.9301)(178,0.9292)(179,0.931)(180,0.927)(181,0.9287)(182,0.9292)(183,0.9355)(184,0.9342)(185,0.9354)(186,0.9308)(187,0.9338)(188,0.9298)(189,0.9304)(190,0.9346)(191,0.9346)(192,0.935)(193,0.9312)(194,0.9352)(195,0.9308)(196,0.9365)(197,0.9363)(198,0.9342)(199,0.9383)(200,0.9363)(201,0.9338)(202,0.9326)(203,0.9372)(204,0.9358)(205,0.9327)(206,0.9303)(207,0.9296)(208,0.9292)(209,0.9326)(210,0.9305)(211,0.93)(212,0.9326)(213,0.9367)(214,0.9354)(215,0.9322)(216,0.9394)(217,0.9378)(218,0.9321)(219,0.9317)(220,0.9299)(221,0.934)(222,0.9338)(223,0.9339)(224,0.9347)(225,0.934)(226,0.9349)(227,0.9394)(228,0.9286)(229,0.9307)(230,0.9297)(231,0.936)(232,0.9344)(233,0.9364)(234,0.9418)(235,0.941)(236,0.9367)(237,0.9368)(238,0.9381)(239,0.9339)(240,0.9343)(241,0.9328)(242,0.9284)(243,0.9261)(244,0.9312)(245,0.9336)(246,0.9351)(247,0.9273)(248,0.93)(249,0.926)(250,0.9304)(251,0.9341)(252,0.9308)(253,0.9298)(254,0.9295)(255,0.9305)(256,0.9308)(257,0.9348)(258,0.9371)(259,0.9356)(260,0.931)(261,0.9361)(262,0.9348)(263,0.9367)(264,0.9369)(265,0.9315)(266,0.933)(267,0.9252)(268,0.9361)(269,0.9328)(270,0.9368)(271,0.9324)(272,0.9357)(273,0.9265)(274,0.936)(275,0.9331)(276,0.9291)(277,0.9299)(278,0.937)(279,0.9329)(280,0.934)(281,0.9359)(282,0.9383)(283,0.9384)(284,0.9373)(285,0.9397)(286,0.934)(287,0.9342)(288,0.9326)(289,0.9356)(290,0.9344)(291,0.9304)(292,0.936)(293,0.9363)(294,0.9345)(295,0.9315)(296,0.9322)(297,0.9312)(298,0.9369)(299,0.9334)

}; \addlegendentry{Cell-free, $P = 3.22$ dBm}
\addplot  [semithick, color =red, mark = triangle*, mark size = 2, mark repeat = 30, mark phase = 10]
coordinates {
(0,0.104)(1,0.1131)(2,0.1699)(3,0.238)(4,0.1585)(5,0.2871)(6,0.3256)(7,0.2961)(8,0.3775)(9,0.4582)(10,0.3772)(11,0.3968)(12,0.3709)(13,0.5564)(14,0.491)(15,0.5625)(16,0.5627)(17,0.5933)(18,0.5742)(19,0.5961)(20,0.5394)(21,0.5309)(22,0.6349)(23,0.4981)(24,0.5411)(25,0.5383)(26,0.3389)(27,0.3196)(28,0.2806)(29,0.2309)(30,0.2234)(31,0.1396)(32,0.1056)(33,0.1497)(34,0.113)(35,0.1404)(36,0.108)(37,0.1339)(38,0.1046)(39,0.1063)(40,0.1522)(41,0.0891)(42,0.0896)(43,0.1156)(44,0.1128)(45,0.0992)(46,0.105)(47,0.0957)(48,0.1061)(49,0.1012)(50,0.1118)(51,0.0933)(52,0.1006)(53,0.0998)(54,0.1031)(55,0.1113)(56,0.1077)(57,0.1014)(58,0.1009)(59,0.097)(60,0.123)(61,0.0971)(62,0.0989)(63,0.0979)(64,0.0975)(65,0.0969)(66,0.0941)(67,0.0974)(68,0.1038)(69,0.1027)(70,0.1033)(71,0.102)(72,0.1031)(73,0.1031)(74,0.1035)(75,0.104)(76,0.1032)(77,0.104)(78,0.103)(79,0.1034)(80,0.1026)(81,0.1027)(82,0.1027)(83,0.1027)(84,0.1027)(85,0.1032)(86,0.104)(87,0.1032)(88,0.098)(89,0.098)(90,0.098)(91,0.098)(92,0.098)(93,0.098)(94,0.098)(95,0.098)(96,0.098)(97,0.098)(98,0.098)(99,0.098)(100,0.098)(101,0.098)(102,0.098)(103,0.098)(104,0.098)(105,0.098)(106,0.098)(107,0.098)(108,0.098)(109,0.098)(110,0.098)(111,0.098)(112,0.098)(113,0.098)(114,0.098)(115,0.098)(116,0.098)(117,0.098)(118,0.098)(119,0.098)(120,0.098)(121,0.098)(122,0.098)(123,0.098)(124,0.098)(125,0.098)(126,0.098)(127,0.098)(128,0.098)(129,0.098)(130,0.098)(131,0.098)(132,0.098)(133,0.098)(134,0.098)(135,0.098)(136,0.098)(137,0.098)(138,0.098)(139,0.098)(140,0.098)(141,0.098)(142,0.098)(143,0.098)(144,0.098)(145,0.098)(146,0.098)(147,0.098)(148,0.098)(149,0.098)(150,0.098)(151,0.098)(152,0.098)(153,0.098)(154,0.098)(155,0.098)(156,0.098)(157,0.098)(158,0.098)(159,0.098)(160,0.098)(161,0.098)(162,0.098)(163,0.098)(164,0.098)(165,0.098)(166,0.098)(167,0.098)(168,0.098)(169,0.098)(170,0.098)(171,0.098)(172,0.098)(173,0.098)(174,0.098)(175,0.098)(176,0.098)(177,0.098)(178,0.098)(179,0.098)(180,0.098)(181,0.098)(182,0.098)(183,0.098)(184,0.098)(185,0.098)(186,0.098)(187,0.098)(188,0.098)(189,0.098)(190,0.098)(191,0.098)(192,0.098)(193,0.098)(194,0.098)(195,0.098)(196,0.098)(197,0.098)(198,0.098)(199,0.098)(200,0.098)(201,0.098)(202,0.098)(203,0.098)(204,0.098)(205,0.098)(206,0.098)(207,0.098)(208,0.098)(209,0.098)(210,0.098)(211,0.098)(212,0.098)(213,0.098)(214,0.098)(215,0.098)(216,0.098)(217,0.098)(218,0.098)(219,0.098)(220,0.098)(221,0.098)(222,0.098)(223,0.098)(224,0.098)(225,0.098)(226,0.098)(227,0.098)(228,0.098)(229,0.098)(230,0.098)(231,0.098)(232,0.098)(233,0.098)(234,0.098)(235,0.098)(236,0.098)(237,0.098)(238,0.098)(239,0.098)(240,0.098)(241,0.098)(242,0.098)(243,0.098)(244,0.098)(245,0.098)(246,0.098)(247,0.098)(248,0.098)(249,0.098)(250,0.098)(251,0.098)(252,0.098)(253,0.098)(254,0.098)(255,0.098)(256,0.098)(257,0.098)(258,0.098)(259,0.098)(260,0.098)(261,0.098)(262,0.098)(263,0.098)(264,0.098)(265,0.098)(266,0.098)(267,0.098)(268,0.098)(269,0.098)(270,0.098)(271,0.098)(272,0.098)(273,0.098)(274,0.098)(275,0.098)(276,0.098)(277,0.098)(278,0.098)(279,0.098)(280,0.098)(281,0.098)(282,0.098)(283,0.098)(284,0.098)(285,0.098)(286,0.098)(287,0.098)(288,0.098)(289,0.098)(290,0.098)(291,0.098)(292,0.098)(293,0.098)(294,0.098)(295,0.098)(296,0.098)(297,0.098)(298,0.098)(299,0.098)

}; \addlegendentry{Cellular}

\end{axis}
\end{tikzpicture}}
\end{center}
\caption{Test accuracy vs. the number of global iterations using i.i.d. MNIST dataset with $\tau=12$. Comparison between cellular and cell-free massive MIMO.}
\label{comp}
\end{figure}

Since the number of global training rounds does not necessarily reflect the time needed till convergence, we added a simulation where we compare the performance of OTA-FL over cellular and cell-free architectures as a function of the training time. In particular, we report in \figref{clock_time} the test accuracy versus the training time using the same setting as in \figref{comp}. Similar convergence behavior as in Fig. 4  is observed and this is due to the fact that a global round takes approximately the same time over both cellular and cell-free massive MIMO as we have remarked in our experiments.

\begin{figure}[h]
\begin{center}
\begin{tikzpicture}[scale=0.8]
\begin{axis}[
tick align=outside,
tick pos=left,
x grid style={white!69.0196078431373!black},
xlabel={Training time (sec.)},
xmajorgrids,
xmin=0, xmax=4800,
xtick style={color=black},
y grid style={white!69.0196078431373!black},
ylabel={Test accuracy},
ymajorgrids,
ymin=0, ymax=1,
ytick style={color=black},
grid=major,
scaled ticks=true,
legend pos=south east,	
grid style=densely dashed,
]

\addplot  [very thick, color =blue, mark = square, mark size = 2, mark repeat = 30, mark phase = 10]
coordinates {
(21.0,0.104)(38.0,0.5623)(54.0,0.7717)(71.0,0.8295)(89.0,0.8625)(106.0,0.8786)(123.0,0.8894)(139.0,0.9004)(156.0,0.9069)(173.0,0.9118)(189.0,0.9151)(205.0,0.9165)(222.0,0.9237)(239.0,0.9251)(256.0,0.9279)(272.0,0.9282)(289.0,0.9306)(306.0,0.9328)(323.0,0.9362)(340.0,0.9399)(357.0,0.9408)(373.0,0.9421)(390.0,0.9432)(406.0,0.9423)(423.0,0.9435)(439.0,0.9423)(456.0,0.9425)(472.0,0.9454)(489.0,0.9442)(506.0,0.9475)(522.0,0.9485)(539.0,0.9498)(555.0,0.952)(571.0,0.9523)(588.0,0.9516)(605.0,0.9517)(622.0,0.9524)(638.0,0.9524)(655.0,0.9536)(672.0,0.9542)(688.0,0.956)(705.0,0.9579)(721.0,0.958)(737.0,0.9562)(754.0,0.9577)(771.0,0.9592)(787.0,0.9586)(804.0,0.9593)(821.0,0.9602)(837.0,0.9603)(854.0,0.9574)(870.0,0.9592)(887.0,0.9564)(903.0,0.9587)(920.0,0.9612)(936.0,0.9592)(952.0,0.9593)(969.0,0.9611)(985.0,0.9613)(1002.0,0.9602)(1019.0,0.9614)(1036.0,0.9615)(1052.0,0.9615)(1068.0,0.9607)(1085.0,0.9631)(1101.0,0.9627)(1118.0,0.9637)(1135.0,0.9637)(1151.0,0.9641)(1167.0,0.9655)(1184.0,0.966)(1201.0,0.9657)(1217.0,0.9643)(1234.0,0.9652)(1251.0,0.9667)(1267.0,0.9661)(1284.0,0.9655)(1300.0,0.9659)(1317.0,0.9664)(1333.0,0.9676)(1350.0,0.9682)(1366.0,0.967)(1383.0,0.9674)(1399.0,0.9675)(1416.0,0.9674)(1432.0,0.968)(1449.0,0.9667)(1465.0,0.9663)(1482.0,0.9659)(1499.0,0.9668)(1515.0,0.9661)(1531.0,0.9666)(1548.0,0.9687)(1565.0,0.9689)(1581.0,0.9673)(1597.0,0.9689)(1614.0,0.9687)(1630.0,0.9674)(1647.0,0.9675)(1663.0,0.968)(1680.0,0.9669)(1697.0,0.9684)(1713.0,0.9671)(1730.0,0.9678)(1746.0,0.9675)(1763.0,0.9674)(1779.0,0.9673)(1795.0,0.9671)(1812.0,0.968)(1828.0,0.9687)(1845.0,0.9682)(1862.0,0.9686)(1878.0,0.9685)(1894.0,0.9685)(1911.0,0.9682)(1927.0,0.9684)(1944.0,0.9691)(1961.0,0.9685)(1978.0,0.9679)(1996.0,0.9672)(2013.0,0.969)(2029.0,0.9672)(2048.0,0.9674)(2065.0,0.9654)(2083.0,0.9652)(2100.0,0.9663)(2118.0,0.9661)(2138.0,0.9668)(2157.0,0.9671)(2174.0,0.9679)(2190.0,0.9678)(2207.0,0.9661)(2224.0,0.9684)(2241.0,0.9676)(2257.0,0.9693)(2274.0,0.9683)(2290.0,0.9672)(2307.0,0.968)(2324.0,0.9692)(2341.0,0.9693)(2357.0,0.9692)(2374.0,0.969)(2391.0,0.9692)(2407.0,0.9695)(2424.0,0.9686)(2441.0,0.9688)(2457.0,0.9698)(2474.0,0.9677)(2491.0,0.9693)(2507.0,0.969)(2524.0,0.9699)(2541.0,0.9697)(2557.0,0.9694)(2574.0,0.9687)(2591.0,0.9678)(2607.0,0.9682)(2624.0,0.9679)(2641.0,0.9693)(2658.0,0.9697)(2675.0,0.97)(2692.0,0.9686)(2708.0,0.9703)(2724.0,0.9709)(2741.0,0.9672)(2757.0,0.9694)(2774.0,0.9707)(2790.0,0.9706)(2807.0,0.9693)(2824.0,0.9705)(2840.0,0.9695)(2857.0,0.9685)(2873.0,0.9696)(2890.0,0.9703)(2906.0,0.9711)(2923.0,0.9712)(2939.0,0.9712)(2956.0,0.969)(2972.0,0.9718)(2989.0,0.9681)(3006.0,0.9687)(3022.0,0.9705)(3039.0,0.9709)(3055.0,0.9707)(3072.0,0.9716)(3088.0,0.9719)(3104.0,0.9719)(3121.0,0.9713)(3137.0,0.9706)(3154.0,0.9712)(3170.0,0.969)(3187.0,0.9704)(3203.0,0.9715)(3220.0,0.9716)(3236.0,0.9698)(3253.0,0.9702)(3270.0,0.9684)(3287.0,0.9702)(3304.0,0.9697)(3320.0,0.9703)(3337.0,0.9688)(3353.0,0.9691)(3370.0,0.968)(3386.0,0.9692)(3403.0,0.9704)(3420.0,0.9695)(3436.0,0.97)(3453.0,0.9699)(3469.0,0.9716)(3485.0,0.9678)(3502.0,0.9702)(3518.0,0.97)(3535.0,0.9717)(3551.0,0.9713)(3568.0,0.9711)(3584.0,0.9705)(3601.0,0.9698)(3617.0,0.9706)(3633.0,0.9711)(3650.0,0.9694)(3666.0,0.9705)(3683.0,0.9703)(3700.0,0.9686)(3716.0,0.9706)(3733.0,0.9711)(3749.0,0.9691)(3765.0,0.9691)(3782.0,0.9707)(3799.0,0.9705)(3815.0,0.97)(3832.0,0.9708)(3848.0,0.9699)(3864.0,0.9703)(3881.0,0.9711)(3897.0,0.9703)(3914.0,0.9716)(3931.0,0.9694)(3947.0,0.97)(3964.0,0.9682)(3980.0,0.9687)(3996.0,0.9702)(4013.0,0.9687)(4029.0,0.9699)(4046.0,0.9707)(4062.0,0.9708)(4079.0,0.9709)(4095.0,0.9705)(4111.0,0.9697)(4128.0,0.9691)(4144.0,0.9684)(4161.0,0.969)(4177.0,0.9686)(4194.0,0.9692)(4211.0,0.9692)(4227.0,0.9682)(4243.0,0.9691)(4260.0,0.9694)(4276.0,0.9695)(4292.0,0.968)(4309.0,0.9686)(4326.0,0.9709)(4343.0,0.9713)(4360.0,0.9691)(4377.0,0.9711)(4393.0,0.9709)(4410.0,0.9696)(4426.0,0.97)(4443.0,0.9697)(4460.0,0.9705)(4477.0,0.9687)(4494.0,0.9684)(4510.0,0.9689)(4528.0,0.9689)(4545.0,0.9693)(4562.0,0.9684)(4579.0,0.967)(4597.0,0.9674)(4614.0,0.968)(4632.0,0.9679)(4649.0,0.9687)(4666.0,0.9665)(4683.0,0.9675)(4701.0,0.9671)(4718.0,0.9676)(4735.0,0.9671)(4753.0,0.9683)(4771.0,0.9677)(4788.0,0.9677)(4806.0,0.9692)(4823.0,0.9683)(4840.0,0.9665)(4858.0,0.9678)(4875.0,0.9674)(4893.0,0.9681)(4910.0,0.9661)(4927.0,0.9657)(4945.0,0.9676)(4962.0,0.9671)(4979.0,0.9683)(4997.0,0.9675)(5014.0,0.966)
}; \addlegendentry{$\alpha_t = 0.5$,   cell-free}

\addplot  [very thick, color =red, mark = triangle, mark size = 2, mark repeat = 30, mark phase = 10]
coordinates {
(16.0,0.104)(33.0,0.3066)(49.0,0.5361)(65.0,0.6527)(81.0,0.7174)(98.0,0.7179)(114.0,0.7704)(130.0,0.7305)(147.0,0.7776)(163.0,0.6932)(180.0,0.8256)(196.0,0.7564)(212.0,0.8329)(229.0,0.8133)(245.0,0.8047)(262.0,0.7877)(278.0,0.7788)(294.0,0.8412)(311.0,0.7768)(327.0,0.836)(343.0,0.8319)(359.0,0.8146)(376.0,0.8244)(392.0,0.8308)(408.0,0.8536)(424.0,0.8274)(440.0,0.8226)(457.0,0.8071)(473.0,0.8448)(489.0,0.8583)(505.0,0.7832)(521.0,0.8641)(538.0,0.8712)(554.0,0.8386)(570.0,0.8852)(586.0,0.8836)(603.0,0.8421)(619.0,0.8899)(635.0,0.8705)(651.0,0.8815)(667.0,0.8689)(683.0,0.8792)(699.0,0.8544)(716.0,0.8745)(732.0,0.8634)(748.0,0.8637)(764.0,0.8698)(780.0,0.8657)(797.0,0.8585)(813.0,0.8775)(829.0,0.8872)(845.0,0.8953)(861.0,0.8943)(877.0,0.8863)(894.0,0.8913)(909.0,0.8742)(926.0,0.883)(942.0,0.8826)(958.0,0.88)(974.0,0.8881)(991.0,0.875)(1007.0,0.8641)(1023.0,0.8888)(1039.0,0.8788)(1056.0,0.8723)(1072.0,0.8939)(1088.0,0.8981)(1104.0,0.8934)(1120.0,0.8857)(1136.0,0.9005)(1153.0,0.889)(1169.0,0.8842)(1185.0,0.8959)(1201.0,0.8908)(1217.0,0.9003)(1233.0,0.9037)(1249.0,0.8931)(1266.0,0.8984)(1282.0,0.8836)(1298.0,0.8971)(1314.0,0.8947)(1331.0,0.8938)(1347.0,0.8915)(1363.0,0.8878)(1379.0,0.8869)(1395.0,0.8984)(1411.0,0.8914)(1428.0,0.8918)(1444.0,0.8905)(1460.0,0.8964)(1476.0,0.8974)(1492.0,0.8825)(1508.0,0.8967)(1524.0,0.9081)(1540.0,0.9004)(1556.0,0.9023)(1572.0,0.8857)(1588.0,0.9077)(1605.0,0.8987)(1621.0,0.9028)(1637.0,0.9039)(1653.0,0.9048)(1669.0,0.904)(1686.0,0.8952)(1702.0,0.9001)(1718.0,0.8981)(1734.0,0.9027)(1751.0,0.8972)(1767.0,0.8956)(1783.0,0.9098)(1799.0,0.899)(1815.0,0.9087)(1831.0,0.9045)(1847.0,0.9073)(1864.0,0.9104)(1880.0,0.9023)(1896.0,0.9094)(1912.0,0.9096)(1928.0,0.906)(1944.0,0.9115)(1960.0,0.9071)(1976.0,0.9095)(1993.0,0.9167)(2009.0,0.9105)(2025.0,0.8995)(2041.0,0.9033)(2058.0,0.917)(2074.0,0.9091)(2090.0,0.9)(2106.0,0.8962)(2122.0,0.9031)(2138.0,0.9059)(2154.0,0.9128)(2171.0,0.9144)(2187.0,0.9132)(2203.0,0.9035)(2219.0,0.913)(2235.0,0.9158)(2251.0,0.9163)(2267.0,0.9149)(2284.0,0.915)(2300.0,0.9137)(2316.0,0.9192)(2332.0,0.9165)(2348.0,0.9139)(2365.0,0.9231)(2381.0,0.9212)(2397.0,0.9155)(2413.0,0.9214)(2429.0,0.917)(2445.0,0.9178)(2462.0,0.9144)(2478.0,0.914)(2494.0,0.9123)(2510.0,0.9135)(2526.0,0.914)(2542.0,0.9164)(2559.0,0.8977)(2575.0,0.9159)(2591.0,0.9142)(2607.0,0.912)(2623.0,0.9172)(2639.0,0.9121)(2655.0,0.9093)(2671.0,0.9154)(2687.0,0.9184)(2704.0,0.9076)(2720.0,0.9158)(2736.0,0.9233)(2752.0,0.9181)(2768.0,0.9186)(2784.0,0.9209)(2800.0,0.9145)(2817.0,0.9055)(2833.0,0.9224)(2849.0,0.9184)(2865.0,0.9151)(2881.0,0.9136)(2898.0,0.9121)(2915.0,0.9075)(2931.0,0.9121)(2947.0,0.9071)(2964.0,0.8865)(2980.0,0.9085)(2996.0,0.9094)(3012.0,0.8929)(3029.0,0.9084)(3045.0,0.8947)(3061.0,0.8987)(3077.0,0.9013)(3093.0,0.9078)(3110.0,0.9137)(3126.0,0.9019)(3142.0,0.9076)(3158.0,0.9138)(3174.0,0.9103)(3191.0,0.9128)(3208.0,0.9032)(3224.0,0.8978)(3241.0,0.9113)(3257.0,0.9061)(3273.0,0.911)(3289.0,0.8945)(3306.0,0.9119)(3322.0,0.9058)(3338.0,0.9063)(3355.0,0.9083)(3371.0,0.908)(3388.0,0.9023)(3404.0,0.9004)(3420.0,0.9007)(3437.0,0.9116)(3454.0,0.9177)(3470.0,0.9082)(3486.0,0.9113)(3503.0,0.9115)(3519.0,0.911)(3535.0,0.9112)(3552.0,0.9049)(3568.0,0.9161)(3584.0,0.9151)(3601.0,0.9046)(3617.0,0.9084)(3634.0,0.9099)(3650.0,0.8988)(3667.0,0.8965)(3683.0,0.9049)(3700.0,0.9056)(3716.0,0.9062)(3733.0,0.9038)(3749.0,0.9073)(3766.0,0.9047)(3782.0,0.9086)(3799.0,0.9077)(3815.0,0.9087)(3833.0,0.9132)(3850.0,0.9134)(3866.0,0.8997)(3882.0,0.9121)(3898.0,0.9005)(3915.0,0.9068)(3931.0,0.8893)(3948.0,0.8955)(3964.0,0.8811)(3980.0,0.8877)(3997.0,0.8865)(4013.0,0.9094)(4030.0,0.8936)(4046.0,0.9074)(4062.0,0.9101)(4079.0,0.8955)(4095.0,0.9038)(4111.0,0.8984)(4127.0,0.8931)(4143.0,0.8945)(4160.0,0.8993)(4176.0,0.9001)(4192.0,0.896)(4208.0,0.8992)(4225.0,0.8837)(4241.0,0.898)(4257.0,0.8977)(4274.0,0.8972)(4290.0,0.8842)(4306.0,0.8805)(4323.0,0.8905)(4339.0,0.8889)(4355.0,0.8967)(4371.0,0.8964)(4387.0,0.8968)(4404.0,0.894)(4420.0,0.8978)(4437.0,0.8893)(4453.0,0.8857)(4469.0,0.8708)(4485.0,0.8874)(4502.0,0.885)(4518.0,0.8774)(4535.0,0.8871)(4551.0,0.877)(4568.0,0.8718)(4584.0,0.87)(4600.0,0.8786)(4616.0,0.8836)(4633.0,0.8688)(4649.0,0.8757)(4666.0,0.8778)(4682.0,0.8763)(4698.0,0.8891)(4715.0,0.8698)(4731.0,0.8702)(4748.0,0.8919)(4764.0,0.8916)(4780.0,0.8805)(4797.0,0.8889)(4813.0,0.8942)(4830.0,0.89)(4847.0,0.8986)(4865.0,0.8838)(4881.0,0.8821)
}; \addlegendentry{$\alpha_t = 0.5$,   cellular}

\end{axis}
\end{tikzpicture}
\end{center}
\caption{Test accuracy vs. training time using i.i.d. MNIST dataset for $N=25$ and $\tau=12$. Comparison between cellular and cell-free massive MIMO.}
\label{clock_time}
\end{figure}
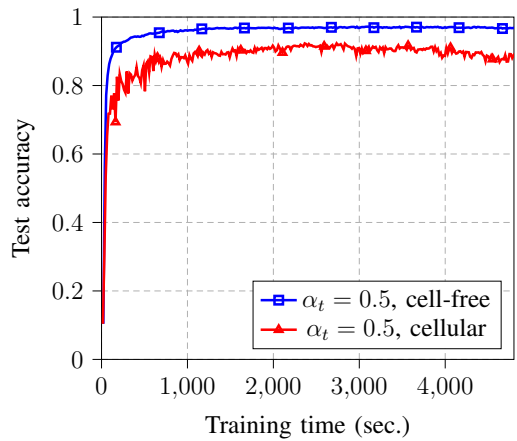

As we mentioned in the previous section, various techniques can be employed to enhance the performance of the FL process in both cell-free and cellular massive MIMO such as client scheduling and partial device participation. In order to provide illustrative examples, we have conducted additional simulations incorporating different user scheduling techniques. In \figref{scheduling}, we present a performance comparison of two distinct scheduling approaches with partial client participation, considering both cellular and cell-free scenarios. 

In the case of partial client participation, only a subset of clients is selected to upload their local updates during each global round. For our simulations, we assume a total of $N = 25$ clients, out of which $r = 10$ clients are chosen to send their local updates at each global round. We investigate two widely adopted scheduling schemes: random selection (RS) and update significance (US).
In the RS scheme, $r$ clients are chosen randomly from the available pool. On the other hand, the US scheme selects the clients with the highest norm of the model updates.

We conclude from the figure that the performance of the two scheduling techniques is almost the same for i.i.d. data distribution for both cell-free and cellular architectures. On the other hand, US scheduling outperforms RS scheduling in the case of non-i.i.d. data over cellular massive MIMO. We remark also that the cell-free architecture provides better performance than cellular MIMO under different scheduling techniques for the same power scaling factor $\alpha_t$.

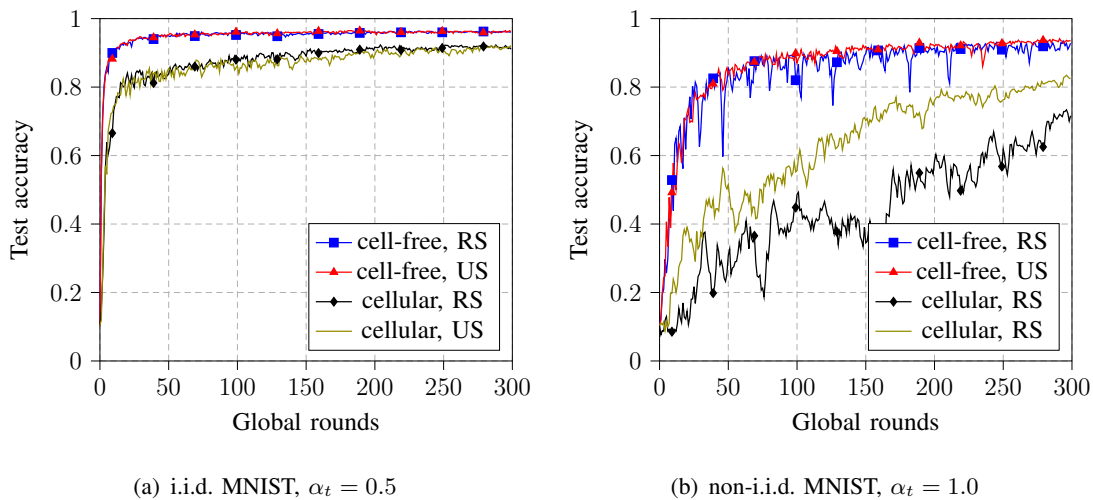
\begin{figure}[h]
\begin{center}
\subfigure[i.i.d. MNIST, $\alpha_t = 0.5$]{
\begin{tikzpicture}[scale=0.8]
\begin{axis}[
tick align=outside,
tick pos=left,
x grid style={white!69.0196078431373!black},
xlabel={Global rounds},
xmajorgrids,
xmin=0, xmax=300,
xtick style={color=black},
y grid style={white!69.0196078431373!black},
ylabel={Test accuracy},
ymajorgrids,
ymin=0, ymax=1,
ytick style={color=black},
grid=major,
scaled ticks=true,
legend pos=south east,	
grid style=densely dashed,
]

\addplot  [semithick, color = blue, mark = square*, mark size = 2, mark repeat = 30, mark phase = 10]
coordinates {

(0,0.104)(1,0.5544)(2,0.7073)(3,0.8039)(4,0.8399)(5,0.8568)(6,0.8434)(7,0.893)(8,0.8844)(9,0.8997)(10,0.8996)(11,0.9083)(12,0.9114)(13,0.9164)(14,0.911)(15,0.9212)(16,0.9253)(17,0.9225)(18,0.9295)(19,0.9226)(20,0.9284)(21,0.9294)(22,0.9308)(23,0.9181)(24,0.9359)(25,0.936)(26,0.9393)(27,0.9419)(28,0.9438)(29,0.9424)(30,0.9437)(31,0.9396)(32,0.945)(33,0.9449)(34,0.9421)(35,0.9376)(36,0.9421)(37,0.9464)(38,0.9464)(39,0.9402)(40,0.9478)(41,0.9507)(42,0.9494)(43,0.9495)(44,0.9468)(45,0.9516)(46,0.9502)(47,0.9518)(48,0.9498)(49,0.9491)(50,0.9495)(51,0.9474)(52,0.9511)(53,0.9489)(54,0.9507)(55,0.9478)(56,0.9569)(57,0.9495)(58,0.9563)(59,0.9511)(60,0.9548)(61,0.9546)(62,0.9558)(63,0.9562)(64,0.9535)(65,0.9536)(66,0.9502)(67,0.9582)(68,0.9549)(69,0.949)(70,0.9576)(71,0.9577)(72,0.9531)(73,0.9524)(74,0.9605)(75,0.954)(76,0.9562)(77,0.9511)(78,0.9524)(79,0.949)(80,0.9507)(81,0.955)(82,0.9554)(83,0.9572)(84,0.9586)(85,0.9592)(86,0.9575)(87,0.9571)(88,0.9548)(89,0.9548)(90,0.9531)(91,0.9581)(92,0.9571)(93,0.9509)(94,0.9575)(95,0.9572)(96,0.9563)(97,0.9598)(98,0.9534)(99,0.9522)(100,0.9559)(101,0.9512)(102,0.9537)(103,0.957)(104,0.9575)(105,0.9555)(106,0.9603)(107,0.955)(108,0.9545)(109,0.9532)(110,0.9533)(111,0.9532)(112,0.9568)(113,0.9589)(114,0.9565)(115,0.95)(116,0.9545)(117,0.957)(118,0.9577)(119,0.9536)(120,0.9586)(121,0.9583)(122,0.9599)(123,0.9569)(124,0.9584)(125,0.9529)(126,0.954)(127,0.9553)(128,0.9527)(129,0.9485)(130,0.9559)(131,0.9544)(132,0.9577)(133,0.953)(134,0.952)(135,0.9462)(136,0.953)(137,0.9491)(138,0.9538)(139,0.9567)(140,0.9516)(141,0.9546)(142,0.9537)(143,0.9552)(144,0.9569)(145,0.9566)(146,0.9557)(147,0.9535)(148,0.9527)(149,0.9542)(150,0.9548)(151,0.9569)(152,0.9545)(153,0.9574)(154,0.9557)(155,0.9571)(156,0.9568)(157,0.9558)(158,0.9551)(159,0.9551)(160,0.9561)(161,0.9563)(162,0.9586)(163,0.9576)(164,0.9569)(165,0.9588)(166,0.9596)(167,0.9585)(168,0.9534)(169,0.9577)(170,0.9558)(171,0.9547)(172,0.9549)(173,0.9582)(174,0.9563)(175,0.953)(176,0.9593)(177,0.957)(178,0.9583)(179,0.9607)(180,0.9586)(181,0.9575)(182,0.9567)(183,0.9562)(184,0.9584)(185,0.9597)(186,0.9569)(187,0.9587)(188,0.9587)(189,0.9582)(190,0.9603)(191,0.9599)(192,0.9591)(193,0.9544)(194,0.9593)(195,0.9595)(196,0.9609)(197,0.9568)(198,0.9605)(199,0.9598)(200,0.9588)(201,0.9585)(202,0.9601)(203,0.9608)(204,0.9584)(205,0.9568)(206,0.9569)(207,0.9608)(208,0.9593)(209,0.9593)(210,0.9611)(211,0.9621)(212,0.9625)(213,0.962)(214,0.961)(215,0.9572)(216,0.959)(217,0.959)(218,0.9541)(219,0.9599)(220,0.9579)(221,0.9584)(222,0.9605)(223,0.9601)(224,0.9597)(225,0.9623)(226,0.9618)(227,0.9608)(228,0.9576)(229,0.9615)(230,0.9587)(231,0.9629)(232,0.9636)(233,0.962)(234,0.9624)(235,0.9607)(236,0.9599)(237,0.9603)(238,0.9567)(239,0.9571)(240,0.9619)(241,0.9626)(242,0.9616)(243,0.9593)(244,0.9594)(245,0.9591)(246,0.9601)(247,0.9586)(248,0.9604)(249,0.9599)(250,0.9596)(251,0.9598)(252,0.9614)(253,0.9623)(254,0.9594)(255,0.9621)(256,0.9613)(257,0.961)(258,0.9617)(259,0.9613)(260,0.9621)(261,0.9598)(262,0.9621)(263,0.9603)(264,0.9598)(265,0.9592)(266,0.9605)(267,0.9608)(268,0.9626)(269,0.9626)(270,0.9625)(271,0.9618)(272,0.9618)(273,0.9615)(274,0.9625)(275,0.9607)(276,0.9611)(277,0.9574)(278,0.961)(279,0.9624)(280,0.9602)(281,0.9646)(282,0.9612)(283,0.9619)(284,0.9595)(285,0.9618)(286,0.9604)(287,0.9603)(288,0.9592)(289,0.9614)(290,0.9592)(291,0.9624)(292,0.9617)(293,0.9623)(294,0.963)(295,0.961)(296,0.9625)(297,0.9617)(298,0.9606)(299,0.961)
}; \addlegendentry{ cell-free, RS  }

\addplot [semithick, color =red, mark =  triangle*,  mark size = 2, mark repeat = 30, mark phase = 10]
coordinates {
(0,0.104)(1,0.5356)(2,0.7291)(3,0.7889)(4,0.8356)(5,0.8461)(6,0.8767)(7,0.8868)(8,0.8956)(9,0.8823)(10,0.9023)(11,0.9091)(12,0.9147)(13,0.9128)(14,0.9178)(15,0.9254)(16,0.9262)(17,0.9299)(18,0.9261)(19,0.9268)(20,0.9219)(21,0.9361)(22,0.9388)(23,0.9358)(24,0.9381)(25,0.9362)(26,0.9374)(27,0.9426)(28,0.9412)(29,0.9474)(30,0.946)(31,0.9432)(32,0.9424)(33,0.9428)(34,0.9427)(35,0.9422)(36,0.9467)(37,0.9456)(38,0.9497)(39,0.9448)(40,0.9522)(41,0.9448)(42,0.952)(43,0.9506)(44,0.9502)(45,0.9521)(46,0.9473)(47,0.9524)(48,0.9454)(49,0.9524)(50,0.9526)(51,0.9569)(52,0.9547)(53,0.9515)(54,0.9494)(55,0.9538)(56,0.9469)(57,0.9548)(58,0.9488)(59,0.9539)(60,0.9553)(61,0.9574)(62,0.9474)(63,0.9519)(64,0.9526)(65,0.9528)(66,0.9539)(67,0.9561)(68,0.9494)(69,0.9532)(70,0.9533)(71,0.9507)(72,0.9537)(73,0.9572)(74,0.9551)(75,0.9592)(76,0.9528)(77,0.9562)(78,0.9562)(79,0.9554)(80,0.9518)(81,0.95)(82,0.9554)(83,0.9551)(84,0.9563)(85,0.9535)(86,0.9525)(87,0.955)(88,0.958)(89,0.9565)(90,0.9563)(91,0.9543)(92,0.9574)(93,0.9533)(94,0.9577)(95,0.9564)(96,0.9598)(97,0.9584)(98,0.9575)(99,0.9612)(100,0.9591)(101,0.9626)(102,0.9571)(103,0.9596)(104,0.9592)(105,0.9588)(106,0.9609)(107,0.9617)(108,0.9538)(109,0.957)(110,0.9535)(111,0.958)(112,0.9596)(113,0.9596)(114,0.9605)(115,0.9591)(116,0.9532)(117,0.9567)(118,0.9607)(119,0.9589)(120,0.956)(121,0.958)(122,0.9591)(123,0.9502)(124,0.9541)(125,0.9585)(126,0.9547)(127,0.9547)(128,0.9537)(129,0.9566)(130,0.9586)(131,0.9578)(132,0.9579)(133,0.9585)(134,0.9576)(135,0.9609)(136,0.9594)(137,0.9606)(138,0.9584)(139,0.9599)(140,0.9581)(141,0.96)(142,0.959)(143,0.9603)(144,0.9598)(145,0.9608)(146,0.9584)(147,0.9573)(148,0.9609)(149,0.961)(150,0.959)(151,0.9601)(152,0.9609)(153,0.9623)(154,0.9601)(155,0.9528)(156,0.9624)(157,0.9603)(158,0.9631)(159,0.9636)(160,0.9633)(161,0.9564)(162,0.9626)(163,0.9627)(164,0.9618)(165,0.9614)(166,0.9627)(167,0.9576)(168,0.9592)(169,0.9609)(170,0.9622)(171,0.9546)(172,0.9611)(173,0.9567)(174,0.9633)(175,0.9637)(176,0.9627)(177,0.9589)(178,0.9636)(179,0.9582)(180,0.9643)(181,0.965)(182,0.9648)(183,0.965)(184,0.9632)(185,0.9631)(186,0.9605)(187,0.9606)(188,0.9591)(189,0.9645)(190,0.9616)(191,0.962)(192,0.9628)(193,0.9647)(194,0.9617)(195,0.9633)(196,0.964)(197,0.9585)(198,0.9621)(199,0.9591)(200,0.9619)(201,0.9527)(202,0.9598)(203,0.9602)(204,0.958)(205,0.9588)(206,0.9609)(207,0.9578)(208,0.9575)(209,0.9585)(210,0.9623)(211,0.9631)(212,0.9624)(213,0.9597)(214,0.9607)(215,0.9595)(216,0.9575)(217,0.9562)(218,0.9612)(219,0.9604)(220,0.9585)(221,0.9566)(222,0.9608)(223,0.963)(224,0.9581)(225,0.9627)(226,0.9586)(227,0.9586)(228,0.9584)(229,0.96)(230,0.9609)(231,0.9618)(232,0.9621)(233,0.9632)(234,0.9601)(235,0.9627)(236,0.9606)(237,0.9592)(238,0.9635)(239,0.9628)(240,0.9634)(241,0.9572)(242,0.9615)(243,0.9577)(244,0.9604)(245,0.9588)(246,0.9591)(247,0.9607)(248,0.9574)(249,0.9633)(250,0.9606)(251,0.9611)(252,0.9604)(253,0.9635)(254,0.9604)(255,0.9623)(256,0.9619)(257,0.9568)(258,0.9633)(259,0.9628)(260,0.963)(261,0.9615)(262,0.9606)(263,0.9585)(264,0.9608)(265,0.9629)(266,0.9631)(267,0.9621)(268,0.9626)(269,0.9619)(270,0.9633)(271,0.9602)(272,0.9622)(273,0.9586)(274,0.9608)(275,0.9578)(276,0.962)(277,0.9633)(278,0.9612)(279,0.9593)(280,0.962)(281,0.9624)(282,0.9608)(283,0.9622)(284,0.9593)(285,0.9561)(286,0.9636)(287,0.9622)(288,0.9615)(289,0.9606)(290,0.962)(291,0.9599)(292,0.9597)(293,0.9613)(294,0.9616)(295,0.9649)(296,0.9602)(297,0.9617)(298,0.9646)(299,0.9622)
}; \addlegendentry{cell-free, US }

\addplot  [semithick, color =black, mark = diamond*, mark size = 2, mark repeat = 30, mark phase = 10]
coordinates {
(0,0.104)(1,0.152)(2,0.2493)(3,0.3855)(4,0.5646)(5,0.6394)(6,0.6043)(7,0.6236)(8,0.6632)(9,0.6654)(10,0.7305)(11,0.7488)(12,0.7725)(13,0.7719)(14,0.7691)(15,0.8128)(16,0.8085)(17,0.8009)(18,0.8027)(19,0.8264)(20,0.8408)(21,0.8207)(22,0.8119)(23,0.8355)(24,0.8101)(25,0.7987)(26,0.8409)(27,0.8471)(28,0.8423)(29,0.849)(30,0.8363)(31,0.8437)(32,0.8083)(33,0.8042)(34,0.8264)(35,0.8317)(36,0.8274)(37,0.8358)(38,0.8434)(39,0.812)(40,0.8388)(41,0.8442)(42,0.8301)(43,0.832)(44,0.8392)(45,0.8359)(46,0.827)(47,0.8247)(48,0.8264)(49,0.8349)(50,0.838)(51,0.8567)(52,0.8413)(53,0.8558)(54,0.862)(55,0.8668)(56,0.8531)(57,0.8577)(58,0.8623)(59,0.8593)(60,0.8517)(61,0.8504)(62,0.8599)(63,0.8574)(64,0.8617)(65,0.8695)(66,0.864)(67,0.8541)(68,0.8564)(69,0.8601)(70,0.8621)(71,0.8489)(72,0.8594)(73,0.8595)(74,0.8672)(75,0.8534)(76,0.8521)(77,0.8627)(78,0.8699)(79,0.8729)(80,0.88)(81,0.8677)(82,0.8711)(83,0.8691)(84,0.8704)(85,0.8584)(86,0.8658)(87,0.8693)(88,0.8742)(89,0.8773)(90,0.8603)(91,0.878)(92,0.8772)(93,0.8711)(94,0.8911)(95,0.875)(96,0.8835)(97,0.8818)(98,0.891)(99,0.8802)(100,0.8857)(101,0.8845)(102,0.8762)(103,0.8867)(104,0.8825)(105,0.892)(106,0.8831)(107,0.8907)(108,0.865)(109,0.8595)(110,0.8793)(111,0.884)(112,0.8756)(113,0.8772)(114,0.8926)(115,0.889)(116,0.8971)(117,0.8987)(118,0.8957)(119,0.8949)(120,0.8851)(121,0.887)(122,0.8896)(123,0.8921)(124,0.8818)(125,0.9004)(126,0.9001)(127,0.8772)(128,0.8905)(129,0.8813)(130,0.8894)(131,0.8981)(132,0.8884)(133,0.8924)(134,0.8906)(135,0.901)(136,0.9003)(137,0.897)(138,0.8962)(139,0.89)(140,0.8898)(141,0.885)(142,0.8778)(143,0.8827)(144,0.8831)(145,0.8932)(146,0.8947)(147,0.9045)(148,0.902)(149,0.8971)(150,0.8903)(151,0.9068)(152,0.9055)(153,0.9111)(154,0.8995)(155,0.9053)(156,0.9074)(157,0.8992)(158,0.8998)(159,0.8993)(160,0.9024)(161,0.905)(162,0.9028)(163,0.9039)(164,0.904)(165,0.9044)(166,0.9018)(167,0.9047)(168,0.9048)(169,0.9031)(170,0.9044)(171,0.9075)(172,0.9079)(173,0.9011)(174,0.9063)(175,0.9025)(176,0.9024)(177,0.9046)(178,0.9122)(179,0.9072)(180,0.9087)(181,0.9089)(182,0.906)(183,0.9079)(184,0.9021)(185,0.899)(186,0.9002)(187,0.9013)(188,0.905)(189,0.9094)(190,0.9043)(191,0.9045)(192,0.9092)(193,0.9055)(194,0.9163)(195,0.9116)(196,0.9146)(197,0.9119)(198,0.9135)(199,0.9083)(200,0.9167)(201,0.9156)(202,0.9154)(203,0.9164)(204,0.9231)(205,0.9176)(206,0.9059)(207,0.915)(208,0.9174)(209,0.9113)(210,0.916)(211,0.9156)(212,0.9143)(213,0.9158)(214,0.9198)(215,0.915)(216,0.9091)(217,0.9115)(218,0.9203)(219,0.9081)(220,0.9141)(221,0.9156)(222,0.9116)(223,0.9137)(224,0.9088)(225,0.9129)(226,0.9073)(227,0.902)(228,0.909)(229,0.9061)(230,0.9126)(231,0.9143)(232,0.9119)(233,0.9122)(234,0.9162)(235,0.9111)(236,0.91)(237,0.915)(238,0.9062)(239,0.9115)(240,0.9132)(241,0.9133)(242,0.9075)(243,0.9078)(244,0.9127)(245,0.9145)(246,0.9062)(247,0.9133)(248,0.9129)(249,0.9142)(250,0.913)(251,0.9226)(252,0.9215)(253,0.9218)(254,0.9173)(255,0.9153)(256,0.9209)(257,0.9188)(258,0.9207)(259,0.9197)(260,0.9193)(261,0.9174)(262,0.9208)(263,0.922)(264,0.9164)(265,0.9177)(266,0.9176)(267,0.9183)(268,0.9151)(269,0.9161)(270,0.9115)(271,0.9123)(272,0.9194)(273,0.9171)(274,0.9192)(275,0.9182)(276,0.9204)(277,0.9204)(278,0.921)(279,0.9186)(280,0.9165)(281,0.9217)(282,0.9189)(283,0.9168)(284,0.9179)(285,0.9197)(286,0.9152)(287,0.9171)(288,0.9231)(289,0.9183)(290,0.9155)(291,0.9156)(292,0.9148)(293,0.9143)(294,0.9165)(295,0.921)(296,0.9172)(297,0.9125)(298,0.9158)(299,0.9107)
}; \addlegendentry{cellular, RS} 

\addplot [semithick, color = olive]
coordinates{
(0,0.104)(1,0.1191)(2,0.2698)(3,0.4243)(4,0.5702)(5,0.5453)(6,0.6722)(7,0.696)(8,0.722)(9,0.7257)(10,0.7449)(11,0.7482)(12,0.7546)(13,0.7832)(14,0.7806)(15,0.7766)(16,0.804)(17,0.8244)(18,0.8042)(19,0.7756)(20,0.7945)(21,0.8154)(22,0.7769)(23,0.8104)(24,0.8207)(25,0.829)(26,0.8122)(27,0.8136)(28,0.8258)(29,0.7996)(30,0.8257)(31,0.8223)(32,0.8219)(33,0.8249)(34,0.8458)(35,0.8406)(36,0.8545)(37,0.809)(38,0.8427)(39,0.8375)(40,0.8429)(41,0.8247)(42,0.8448)(43,0.8168)(44,0.8162)(45,0.846)(46,0.8239)(47,0.844)(48,0.8233)(49,0.8463)(50,0.8545)(51,0.8518)(52,0.8324)(53,0.8456)(54,0.8261)(55,0.8565)(56,0.8581)(57,0.8498)(58,0.8325)(59,0.8447)(60,0.8359)(61,0.8416)(62,0.8331)(63,0.8426)(64,0.8504)(65,0.8592)(66,0.8509)(67,0.8501)(68,0.8574)(69,0.8652)(70,0.8789)(71,0.8772)(72,0.8726)(73,0.8499)(74,0.8295)(75,0.8396)(76,0.846)(77,0.8632)(78,0.8586)(79,0.8645)(80,0.8678)(81,0.8738)(82,0.8615)(83,0.8573)(84,0.8635)(85,0.8551)(86,0.8721)(87,0.844)(88,0.8797)(89,0.8657)(90,0.8665)(91,0.8762)(92,0.8745)(93,0.8763)(94,0.8696)(95,0.8663)(96,0.8634)(97,0.8676)(98,0.8498)(99,0.8566)(100,0.8621)(101,0.8398)(102,0.8563)(103,0.8557)(104,0.8689)(105,0.8664)(106,0.86)(107,0.8484)(108,0.8722)(109,0.8757)(110,0.8569)(111,0.8764)(112,0.8752)(113,0.8788)(114,0.8723)(115,0.8576)(116,0.8656)(117,0.8531)(118,0.8612)(119,0.8642)(120,0.8632)(121,0.8706)(122,0.8804)(123,0.8711)(124,0.8809)(125,0.8712)(126,0.8718)(127,0.8582)(128,0.8749)(129,0.8631)(130,0.8623)(131,0.8879)(132,0.8842)(133,0.884)(134,0.8818)(135,0.8821)(136,0.8846)(137,0.8742)(138,0.8565)(139,0.8584)(140,0.8516)(141,0.8639)(142,0.8714)(143,0.8695)(144,0.8734)(145,0.8772)(146,0.8919)(147,0.8852)(148,0.8819)(149,0.874)(150,0.8732)(151,0.8848)(152,0.8763)(153,0.8743)(154,0.8843)(155,0.8827)(156,0.8866)(157,0.8851)(158,0.8842)(159,0.8795)(160,0.8841)(161,0.8793)(162,0.8856)(163,0.8851)(164,0.8872)(165,0.8822)(166,0.8779)(167,0.8637)(168,0.8735)(169,0.8924)(170,0.8817)(171,0.8916)(172,0.8879)(173,0.8887)(174,0.8882)(175,0.88)(176,0.8824)(177,0.8804)(178,0.885)(179,0.8935)(180,0.8924)(181,0.8985)(182,0.8945)(183,0.8928)(184,0.8952)(185,0.9009)(186,0.8967)(187,0.8979)(188,0.8979)(189,0.8991)(190,0.9062)(191,0.9029)(192,0.8985)(193,0.9045)(194,0.902)(195,0.9018)(196,0.8963)(197,0.8852)(198,0.8948)(199,0.8848)(200,0.895)(201,0.9037)(202,0.9091)(203,0.9006)(204,0.902)(205,0.8972)(206,0.9073)(207,0.9093)(208,0.9046)(209,0.9101)(210,0.898)(211,0.895)(212,0.9006)(213,0.8968)(214,0.8971)(215,0.9006)(216,0.8962)(217,0.8935)(218,0.8884)(219,0.8923)(220,0.9057)(221,0.9078)(222,0.9056)(223,0.8979)(224,0.895)(225,0.8919)(226,0.9041)(227,0.9002)(228,0.8966)(229,0.9056)(230,0.9088)(231,0.9127)(232,0.9145)(233,0.9094)(234,0.9109)(235,0.9114)(236,0.9089)(237,0.9038)(238,0.9046)(239,0.8984)(240,0.9048)(241,0.9041)(242,0.9124)(243,0.9093)(244,0.9154)(245,0.9065)(246,0.8999)(247,0.9052)(248,0.9094)(249,0.903)(250,0.9047)(251,0.9035)(252,0.8976)(253,0.8974)(254,0.8916)(255,0.9012)(256,0.9099)(257,0.9094)(258,0.9091)(259,0.907)(260,0.8972)(261,0.903)(262,0.9072)(263,0.9012)(264,0.8973)(265,0.9084)(266,0.9054)(267,0.9133)(268,0.9149)(269,0.9109)(270,0.9102)(271,0.9015)(272,0.9148)(273,0.9106)(274,0.911)(275,0.9175)(276,0.9059)(277,0.9145)(278,0.9168)(279,0.9165)(280,0.9138)(281,0.9149)(282,0.9163)(283,0.9155)(284,0.9198)(285,0.9197)(286,0.9156)(287,0.9147)(288,0.9202)(289,0.9156)(290,0.9155)(291,0.9097)(292,0.9098)(293,0.9165)(294,0.9139)(295,0.9169)(296,0.9144)(297,0.917)(298,0.9167)(299,0.9216)
};\addlegendentry{cellular, US}
\end{axis}
\end{tikzpicture}}
\subfigure[non-i.i.d. MNIST, $\alpha_t=1.0$]{
\begin{tikzpicture}[scale=0.8]
\begin{axis}[
tick align=outside,
tick pos=left,
x grid style={white!69.0196078431373!black},
xlabel={Global rounds},
xmajorgrids,
xmin=0, xmax=300,
xtick style={color=black},
y grid style={white!69.0196078431373!black},
ylabel={Test accuracy},
ymajorgrids,
ymin=0, ymax=1,
ytick style={color=black},
grid=major,
scaled ticks=true,
legend pos=south east,	
grid style=densely dashed,
]

\addplot [semithick, color =blue, mark =  square*,  mark size = 2, mark repeat = 30, mark phase = 10]
coordinates {
(0,0.104)(1,0.1287)(2,0.1956)(3,0.2086)(4,0.2971)(5,0.2566)(6,0.2984)(7,0.4077)(8,0.427)(9,0.5284)(10,0.439)(11,0.6336)(12,0.6502)(13,0.5802)(14,0.67)(15,0.6849)(16,0.5999)(17,0.5618)(18,0.7103)(19,0.6181)(20,0.713)(21,0.7205)(22,0.7716)(23,0.7063)(24,0.8021)(25,0.773)(26,0.8141)(27,0.8026)(28,0.7446)(29,0.6259)(30,0.6732)(31,0.7527)(32,0.7741)(33,0.7641)(34,0.8101)(35,0.7749)(36,0.8116)(37,0.8081)(38,0.8268)(39,0.8255)(40,0.7999)(41,0.832)(42,0.8342)(43,0.8326)(44,0.8402)(45,0.7441)(46,0.597)(47,0.7456)(48,0.8009)(49,0.8175)(50,0.8554)(51,0.8453)(52,0.807)(53,0.7899)(54,0.8254)(55,0.8206)(56,0.8417)(57,0.8271)(58,0.8198)(59,0.8268)(60,0.8499)(61,0.8629)(62,0.8495)(63,0.8565)(64,0.804)(65,0.854)(66,0.7689)(67,0.8183)(68,0.8123)(69,0.8758)(70,0.851)(71,0.8808)(72,0.8661)(73,0.8753)(74,0.8899)(75,0.8661)(76,0.8617)(77,0.8783)(78,0.8907)(79,0.8647)(80,0.7892)(81,0.8296)(82,0.8635)(83,0.8859)(84,0.8729)(85,0.8793)(86,0.8773)(87,0.8851)(88,0.8508)(89,0.8508)(90,0.8597)(91,0.8549)(92,0.8805)(93,0.8134)(94,0.891)(95,0.8917)(96,0.8931)(97,0.8889)(98,0.8903)(99,0.82)(100,0.8378)(101,0.8669)(102,0.7888)(103,0.7695)(104,0.8672)(105,0.8957)(106,0.8802)(107,0.8659)(108,0.8936)(109,0.8908)(110,0.8926)(111,0.881)(112,0.8824)(113,0.887)(114,0.8733)(115,0.846)(116,0.8498)(117,0.8848)(118,0.8903)(119,0.8776)(120,0.8906)(121,0.8843)(122,0.8987)(123,0.8807)(124,0.8349)(125,0.8501)(126,0.7463)(127,0.8244)(128,0.8495)(129,0.8725)(130,0.8977)(131,0.8981)(132,0.9026)(133,0.887)(134,0.8735)(135,0.8678)(136,0.9016)(137,0.8923)(138,0.8708)(139,0.8829)(140,0.8865)(141,0.9011)(142,0.8912)(143,0.9)(144,0.8725)(145,0.8539)(146,0.8504)(147,0.8639)(148,0.875)(149,0.8913)(150,0.8978)(151,0.9146)(152,0.9155)(153,0.8866)(154,0.9167)(155,0.891)(156,0.9091)(157,0.9105)(158,0.8976)(159,0.9071)(160,0.887)(161,0.8745)(162,0.904)(163,0.8626)(164,0.8479)(165,0.8981)(166,0.8756)(167,0.8848)(168,0.9025)(169,0.9118)(170,0.9088)(171,0.8928)(172,0.9062)(173,0.8947)(174,0.9016)(175,0.9151)(176,0.9131)(177,0.9018)(178,0.9205)(179,0.9053)(180,0.8806)(181,0.8699)(182,0.7849)(183,0.8787)(184,0.8885)(185,0.906)(186,0.8891)(187,0.8931)(188,0.9124)(189,0.9139)(190,0.8983)(191,0.9102)(192,0.9112)(193,0.9053)(194,0.9052)(195,0.9004)(196,0.9084)(197,0.9148)(198,0.9086)(199,0.9116)(200,0.9137)(201,0.9007)(202,0.9166)(203,0.8918)(204,0.9051)(205,0.9042)(206,0.9034)(207,0.9014)(208,0.9163)(209,0.8921)(210,0.8283)(211,0.811)(212,0.8849)(213,0.916)(214,0.905)(215,0.9075)(216,0.9264)(217,0.9283)(218,0.9137)(219,0.9111)(220,0.9262)(221,0.9211)(222,0.9223)(223,0.9252)(224,0.9216)(225,0.9241)(226,0.921)(227,0.9245)(228,0.9178)(229,0.9011)(230,0.9139)(231,0.9097)(232,0.9263)(233,0.9236)(234,0.9066)(235,0.915)(236,0.9187)(237,0.9262)(238,0.919)(239,0.9046)(240,0.9242)(241,0.9199)(242,0.9162)(243,0.9105)(244,0.9225)(245,0.9294)(246,0.9095)(247,0.9184)(248,0.9171)(249,0.9093)(250,0.9129)(251,0.9088)(252,0.8882)(253,0.9097)(254,0.9214)(255,0.9217)(256,0.9195)(257,0.9158)(258,0.916)(259,0.909)(260,0.9192)(261,0.8999)(262,0.9148)(263,0.9288)(264,0.9031)(265,0.8989)(266,0.9201)(267,0.9259)(268,0.9023)(269,0.8707)(270,0.9007)(271,0.9228)(272,0.9302)(273,0.929)(274,0.9192)(275,0.9173)(276,0.9162)(277,0.9298)(278,0.9259)(279,0.919)(280,0.9173)(281,0.9164)(282,0.9093)(283,0.9253)(284,0.913)(285,0.9201)(286,0.929)(287,0.9204)(288,0.9062)(289,0.908)(290,0.9125)(291,0.9085)(292,0.9297)(293,0.9277)(294,0.9342)(295,0.9212)(296,0.9277)(297,0.9105)(298,0.9184)(299,0.928)
}; \addlegendentry{ cell-free,   RS}
\addplot  [semithick, color = red, mark = triangle*, mark size = 2, mark repeat = 30, mark phase = 10]
coordinates {
(0,0.104)(1,0.1185)(2,0.1994)(3,0.2433)(4,0.2478)(5,0.4064)(6,0.3342)(7,0.4779)(8,0.389)(9,0.4921)(10,0.5781)(11,0.517)(12,0.4997)(13,0.5834)(14,0.6351)(15,0.6171)(16,0.5788)(17,0.6305)(18,0.7058)(19,0.7045)(20,0.7238)(21,0.7388)(22,0.6969)(23,0.7009)(24,0.6987)(25,0.7648)(26,0.7834)(27,0.7633)(28,0.7665)(29,0.7707)(30,0.7762)(31,0.772)(32,0.7809)(33,0.7702)(34,0.7887)(35,0.7941)(36,0.8067)(37,0.7958)(38,0.7956)(39,0.8096)(40,0.8156)(41,0.8263)(42,0.8343)(43,0.8523)(44,0.8292)(45,0.8215)(46,0.7913)(47,0.8011)(48,0.8433)(49,0.8557)(50,0.8352)(51,0.8472)(52,0.8516)(53,0.8566)(54,0.8308)(55,0.8658)(56,0.8615)(57,0.8504)(58,0.854)(59,0.847)(60,0.8416)(61,0.8577)(62,0.849)(63,0.8692)(64,0.8734)(65,0.8724)(66,0.8648)(67,0.8728)(68,0.8684)(69,0.8733)(70,0.8601)(71,0.8696)(72,0.8732)(73,0.8901)(74,0.8883)(75,0.8891)(76,0.866)(77,0.8948)(78,0.8801)(79,0.8814)(80,0.8905)(81,0.8827)(82,0.8667)(83,0.8659)(84,0.8664)(85,0.8966)(86,0.86)(87,0.8847)(88,0.8854)(89,0.8991)(90,0.8839)(91,0.8835)(92,0.8831)(93,0.899)(94,0.8848)(95,0.8857)(96,0.8847)(97,0.8852)(98,0.8766)(99,0.8984)(100,0.8706)(101,0.8835)(102,0.8901)(103,0.9093)(104,0.8928)(105,0.8931)(106,0.8903)(107,0.8812)(108,0.8807)(109,0.8997)(110,0.904)(111,0.8787)(112,0.8967)(113,0.8698)(114,0.9077)(115,0.895)(116,0.9043)(117,0.8928)(118,0.8962)(119,0.8898)(120,0.8935)(121,0.8954)(122,0.902)(123,0.9028)(124,0.9088)(125,0.904)(126,0.9074)(127,0.8999)(128,0.9184)(129,0.9051)(130,0.9083)(131,0.8865)(132,0.8932)(133,0.885)(134,0.8843)(135,0.8975)(136,0.9147)(137,0.9122)(138,0.9152)(139,0.9025)(140,0.9187)(141,0.9086)(142,0.9145)(143,0.9132)(144,0.9135)(145,0.9075)(146,0.9128)(147,0.909)(148,0.9216)(149,0.9142)(150,0.9188)(151,0.909)(152,0.9132)(153,0.9108)(154,0.914)(155,0.9068)(156,0.9144)(157,0.9077)(158,0.9169)(159,0.9094)(160,0.9106)(161,0.9077)(162,0.9151)(163,0.9145)(164,0.9068)(165,0.9216)(166,0.9176)(167,0.9253)(168,0.9213)(169,0.9199)(170,0.9102)(171,0.9148)(172,0.907)(173,0.914)(174,0.9189)(175,0.9169)(176,0.9153)(177,0.9091)(178,0.9108)(179,0.9111)(180,0.9164)(181,0.9229)(182,0.9203)(183,0.92)(184,0.9096)(185,0.8977)(186,0.9056)(187,0.9062)(188,0.9286)(189,0.9291)(190,0.9219)(191,0.9156)(192,0.9246)(193,0.9064)(194,0.9179)(195,0.9161)(196,0.9121)(197,0.9266)(198,0.9236)(199,0.9235)(200,0.9192)(201,0.9246)(202,0.9222)(203,0.9246)(204,0.9197)(205,0.9197)(206,0.9153)(207,0.9138)(208,0.9141)(209,0.9153)(210,0.9192)(211,0.9041)(212,0.9242)(213,0.916)(214,0.914)(215,0.9205)(216,0.912)(217,0.9124)(218,0.9139)(219,0.9213)(220,0.9168)(221,0.9324)(222,0.9287)(223,0.9236)(224,0.9154)(225,0.9134)(226,0.9109)(227,0.8908)(228,0.9123)(229,0.9211)(230,0.9228)(231,0.9268)(232,0.9238)(233,0.9123)(234,0.9074)(235,0.861)(236,0.8796)(237,0.9023)(238,0.9253)(239,0.9276)(240,0.9306)(241,0.9322)(242,0.9292)(243,0.9307)(244,0.9226)(245,0.9242)(246,0.931)(247,0.9303)(248,0.93)(249,0.9273)(250,0.9349)(251,0.9297)(252,0.9254)(253,0.9281)(254,0.9296)(255,0.9336)(256,0.924)(257,0.9328)(258,0.9277)(259,0.9338)(260,0.9315)(261,0.9373)(262,0.9324)(263,0.9309)(264,0.9328)(265,0.9315)(266,0.9294)(267,0.931)(268,0.9339)(269,0.932)(270,0.9316)(271,0.9266)(272,0.9233)(273,0.9321)(274,0.9329)(275,0.9377)(276,0.9353)(277,0.9348)(278,0.9315)(279,0.937)(280,0.9343)(281,0.9374)(282,0.9322)(283,0.9349)(284,0.9387)(285,0.9298)(286,0.9392)(287,0.937)(288,0.9358)(289,0.9348)(290,0.9324)(291,0.9414)(292,0.9351)(293,0.9343)(294,0.9277)(295,0.9322)(296,0.9365)(297,0.9337)(298,0.9354)(299,0.9352)
}; \addlegendentry{ cell-free, US}
\addplot  [semithick, color = black, mark = diamond*, mark size = 2, mark repeat = 30, mark phase = 10]
coordinates {
(0,0.104)(1,0.0768)(2,0.0886)(3,0.0871)(4,0.1189)(5,0.1293)(6,0.0871)(7,0.0921)(8,0.0823)(9,0.0854)(10,0.0801)(11,0.0917)(12,0.0871)(13,0.1097)(14,0.1531)(15,0.1275)(16,0.1124)(17,0.1696)(18,0.1281)(19,0.1223)(20,0.1447)(21,0.1064)(22,0.1274)(23,0.1801)(24,0.1812)(25,0.2169)(26,0.1682)(27,0.2064)(28,0.2463)(29,0.2545)(30,0.2722)(31,0.3411)(32,0.3728)(33,0.3762)(34,0.3319)(35,0.291)(36,0.2617)(37,0.2548)(38,0.193)(39,0.1982)(40,0.2182)(41,0.2672)(42,0.3384)(43,0.3194)(44,0.2802)(45,0.2917)(46,0.2721)(47,0.2421)(48,0.2325)(49,0.2551)(50,0.3085)(51,0.3273)(52,0.2535)(53,0.2715)(54,0.2751)(55,0.2899)(56,0.2907)(57,0.3139)(58,0.3482)(59,0.3243)(60,0.3703)(61,0.371)(62,0.3785)(63,0.409)(64,0.371)(65,0.328)(66,0.3922)(67,0.3562)(68,0.3479)(69,0.3656)(70,0.2883)(71,0.2482)(72,0.2446)(73,0.2505)(74,0.2069)(75,0.2126)(76,0.1876)(77,0.2377)(78,0.3132)(79,0.298)(80,0.3426)(81,0.4235)(82,0.4509)(83,0.4676)(84,0.4475)(85,0.4299)(86,0.4053)(87,0.3719)(88,0.3871)(89,0.3857)(90,0.3736)(91,0.3866)(92,0.3645)(93,0.3952)(94,0.3766)(95,0.3799)(96,0.4037)(97,0.4064)(98,0.4448)(99,0.4485)(100,0.4865)(101,0.4927)(102,0.4399)(103,0.4436)(104,0.4164)(105,0.424)(106,0.4147)(107,0.4229)(108,0.411)(109,0.4043)(110,0.3841)(111,0.4095)(112,0.3834)(113,0.4036)(114,0.4152)(115,0.4094)(116,0.4139)(117,0.4447)(118,0.4399)(119,0.4523)(120,0.4656)(121,0.4467)(122,0.4263)(123,0.4207)(124,0.4016)(125,0.4157)(126,0.3905)(127,0.371)(128,0.3986)(129,0.3774)(130,0.3393)(131,0.379)(132,0.3652)(133,0.4124)(134,0.3944)(135,0.3943)(136,0.3704)(137,0.415)(138,0.3801)(139,0.4252)(140,0.4113)(141,0.4175)(142,0.4221)(143,0.3913)(144,0.4147)(145,0.4088)(146,0.3628)(147,0.3522)(148,0.3347)(149,0.3907)(150,0.3455)(151,0.3491)(152,0.3798)(153,0.4093)(154,0.3833)(155,0.3778)(156,0.3685)(157,0.3488)(158,0.3485)(159,0.3475)(160,0.3353)(161,0.3118)(162,0.2877)(163,0.285)(164,0.3349)(165,0.4299)(166,0.4268)(167,0.4225)(168,0.4701)(169,0.5286)(170,0.5585)(171,0.5226)(172,0.5116)(173,0.4776)(174,0.5383)(175,0.5227)(176,0.4844)(177,0.4598)(178,0.4665)(179,0.533)(180,0.5536)(181,0.5358)(182,0.5105)(183,0.5707)(184,0.5447)(185,0.5241)(186,0.5009)(187,0.5228)(188,0.5389)(189,0.5493)(190,0.5252)(191,0.523)(192,0.5174)(193,0.5068)(194,0.5635)(195,0.5558)(196,0.5508)(197,0.5497)(198,0.559)(199,0.546)(200,0.5841)(201,0.6028)(202,0.5833)(203,0.5822)(204,0.5634)(205,0.5658)(206,0.5835)(207,0.5978)(208,0.5703)(209,0.5234)(210,0.4951)(211,0.4793)(212,0.5603)(213,0.5398)(214,0.5749)(215,0.571)(216,0.5656)(217,0.5417)(218,0.5421)(219,0.4974)(220,0.5106)(221,0.4923)(222,0.5188)(223,0.5302)(224,0.5515)(225,0.558)(226,0.5717)(227,0.5277)(228,0.5236)(229,0.5219)(230,0.5091)(231,0.5234)(232,0.5669)(233,0.5651)(234,0.602)(235,0.6083)(236,0.5792)(237,0.6058)(238,0.5941)(239,0.6281)(240,0.6503)(241,0.656)(242,0.6435)(243,0.6688)(244,0.6492)(245,0.6229)(246,0.6054)(247,0.5965)(248,0.5756)(249,0.5679)(250,0.6186)(251,0.5989)(252,0.563)(253,0.5443)(254,0.5767)(255,0.6316)(256,0.6326)(257,0.6011)(258,0.6372)(259,0.6517)(260,0.6303)(261,0.6286)(262,0.6303)(263,0.6269)(264,0.6277)(265,0.6201)(266,0.6159)(267,0.6127)(268,0.631)(269,0.6274)(270,0.6544)(271,0.6699)(272,0.6515)(273,0.6267)(274,0.6035)(275,0.633)(276,0.6448)(277,0.6429)(278,0.6311)(279,0.6251)(280,0.6322)(281,0.6812)(282,0.6996)(283,0.6996)(284,0.6773)(285,0.6963)(286,0.6958)(287,0.7203)(288,0.712)(289,0.7064)(290,0.713)(291,0.7205)(292,0.7173)(293,0.7137)(294,0.7285)(295,0.7288)(296,0.7337)(297,0.7234)(298,0.7007)(299,0.7162)
}; \addlegendentry{cellular, RS} 

\addplot [semithick, color = olive]
coordinates{
(0,0.104)(1,0.1104)(2,0.1058)(3,0.1141)(4,0.0892)(5,0.1163)(6,0.1006)(7,0.1101)(8,0.1979)(9,0.198)(10,0.2388)(11,0.2241)(12,0.2007)(13,0.2218)(14,0.274)(15,0.287)(16,0.2943)(17,0.3646)(18,0.384)(19,0.3697)(20,0.3901)(21,0.3762)(22,0.3489)(23,0.3558)(24,0.3399)(25,0.3182)(26,0.2907)(27,0.3782)(28,0.328)(29,0.3355)(30,0.376)(31,0.3626)(32,0.4156)(33,0.4378)(34,0.4683)(35,0.4552)(36,0.4734)(37,0.4506)(38,0.4473)(39,0.4711)(40,0.422)(41,0.4446)(42,0.4449)(43,0.4381)(44,0.4801)(45,0.5242)(46,0.5626)(47,0.5513)(48,0.5114)(49,0.5028)(50,0.5152)(51,0.4773)(52,0.443)(53,0.4451)(54,0.423)(55,0.3795)(56,0.3939)(57,0.3993)(58,0.4328)(59,0.4403)(60,0.4582)(61,0.4677)(62,0.4956)(63,0.4838)(64,0.4883)(65,0.4641)(66,0.4858)(67,0.4311)(68,0.4505)(69,0.4486)(70,0.4153)(71,0.4178)(72,0.455)(73,0.5017)(74,0.4996)(75,0.4608)(76,0.5167)(77,0.511)(78,0.5327)(79,0.5402)(80,0.5172)(81,0.5445)(82,0.5308)(83,0.5439)(84,0.4974)(85,0.5205)(86,0.5356)(87,0.5339)(88,0.5327)(89,0.517)(90,0.5416)(91,0.5163)(92,0.5636)(93,0.5672)(94,0.5695)(95,0.564)(96,0.5417)(97,0.5498)(98,0.5815)(99,0.5524)(100,0.5559)(101,0.5911)(102,0.6163)(103,0.5932)(104,0.5638)(105,0.5474)(106,0.5361)(107,0.5114)(108,0.5338)(109,0.5748)(110,0.5606)(111,0.5589)(112,0.5604)(113,0.5783)(114,0.6011)(115,0.6321)(116,0.6496)(117,0.6425)(118,0.6512)(119,0.6338)(120,0.5915)(121,0.6281)(122,0.655)(123,0.6333)(124,0.6608)(125,0.6635)(126,0.6489)(127,0.6416)(128,0.6199)(129,0.6229)(130,0.6287)(131,0.65)(132,0.6501)(133,0.617)(134,0.6372)(135,0.6635)(136,0.6587)(137,0.6366)(138,0.6379)(139,0.6674)(140,0.6724)(141,0.6672)(142,0.6898)(143,0.6685)(144,0.6688)(145,0.6511)(146,0.6729)(147,0.6851)(148,0.7076)(149,0.7111)(150,0.7024)(151,0.708)(152,0.6892)(153,0.7045)(154,0.6978)(155,0.7076)(156,0.7367)(157,0.7429)(158,0.7084)(159,0.7028)(160,0.6995)(161,0.7204)(162,0.7087)(163,0.729)(164,0.7337)(165,0.726)(166,0.736)(167,0.7503)(168,0.7524)(169,0.7524)(170,0.7549)(171,0.7542)(172,0.7696)(173,0.7362)(174,0.7516)(175,0.7513)(176,0.7347)(177,0.7275)(178,0.7286)(179,0.7237)(180,0.744)(181,0.7436)(182,0.7488)(183,0.7423)(184,0.7462)(185,0.7112)(186,0.7134)(187,0.7232)(188,0.717)(189,0.7248)(190,0.7053)(191,0.6784)(192,0.6815)(193,0.7153)(194,0.7297)(195,0.7328)(196,0.7542)(197,0.7676)(198,0.7558)(199,0.7483)(200,0.7609)(201,0.768)(202,0.7717)(203,0.7696)(204,0.7749)(205,0.7928)(206,0.7802)(207,0.7751)(208,0.7752)(209,0.7862)(210,0.771)(211,0.7691)(212,0.7678)(213,0.7659)(214,0.7596)(215,0.7737)(216,0.7736)(217,0.7751)(218,0.7722)(219,0.7666)(220,0.7747)(221,0.7701)(222,0.7742)(223,0.7864)(224,0.7798)(225,0.7788)(226,0.764)(227,0.7784)(228,0.7709)(229,0.7549)(230,0.7475)(231,0.7559)(232,0.7804)(233,0.7739)(234,0.7572)(235,0.741)(236,0.7523)(237,0.7556)(238,0.7631)(239,0.7746)(240,0.7705)(241,0.7845)(242,0.7929)(243,0.7941)(244,0.8034)(245,0.7937)(246,0.7649)(247,0.7598)(248,0.7559)(249,0.759)(250,0.7632)(251,0.7643)(252,0.7823)(253,0.7701)(254,0.768)(255,0.7852)(256,0.7907)(257,0.7937)(258,0.7939)(259,0.7877)(260,0.7839)(261,0.79)(262,0.7876)(263,0.7869)(264,0.7924)(265,0.7935)(266,0.8031)(267,0.8002)(268,0.8013)(269,0.7966)(270,0.7975)(271,0.8022)(272,0.7905)(273,0.7913)(274,0.7994)(275,0.8119)(276,0.8035)(277,0.8016)(278,0.8163)(279,0.8094)(280,0.8096)(281,0.8124)(282,0.8137)(283,0.8166)(284,0.8168)(285,0.8156)(286,0.8183)(287,0.8196)(288,0.8144)(289,0.8106)(290,0.799)(291,0.8004)(292,0.8125)(293,0.807)(294,0.8283)(295,0.8304)(296,0.8346)(297,0.8304)(298,0.826)(299,0.8243)
};\addlegendentry{cellular, RS}
\end{axis}
\end{tikzpicture}}
\end{center}
\caption{Test accuracy vs. the number of global iterations using MNIST dataset for $N=25$, $r=10$, $L = 25$ and $\tau=12$.}
\label{scheduling}
\end{figure}

\subsection{Effect of the AP-UT association scheme}

It is important to note that the AP-UT association scheme influences the performance of cell-free MIMO in general and the FL process in particular. Our theoretical convergence analysis remains valid for any AP-UT association scheme, as long as the user-centric cell-free massive MIMO architecture is considered. The influence of the AP-UT association scheme is reflected in the constants $\gamma$, $\tilde\gamma$, $\kappa$, and $\tilde\kappa$ through the association matrices $\bD_n$. For the sake of illustration, we conducted additional simulations to explore the effect of different AP-UT association techniques on the FL process. Specifically, we compare the performance of our proposed implementation using two distinct AP-UT association schemes in \figref{AP_UT}. The first scheme, used in all our previous experiments and proposed in \cite{Emilcellfree}, assigns APs and pilots jointly to UTs based on large-scale fading coefficients and pilot contamination levels. The second scheme, developed in \cite{chen2022improving}, is an interference-aware massive access scheme that achieves joint AP-UT association and pilot assignment by leveraging large-scale interference characteristics.

From the figure, it is clear that the AP-UT-2 scheme (\cite{chen2022improving}) outperforms the AP-UT-1 scheme (\cite{Emilcellfree}) for smaller values of the power scaling factor $\alpha_t$. However, for relatively high values of $\alpha_t$, the performance of both schemes becomes indistinguishable.

\begin{figure}[h]
\begin{center}
\begin{tikzpicture}[scale=0.8]
\begin{axis}[
tick align=outside,
tick pos=left,
x grid style={white!69.0196078431373!black},
xlabel={Global rounds},
xmajorgrids,
xmin=0, xmax=300,
xtick style={color=black},
y grid style={white!69.0196078431373!black},
ylabel={Test accuracy},
ymajorgrids,
ymin=0, ymax=1,
ytick style={color=black},
grid=major,
scaled ticks=true,
legend pos=south east,	
grid style=densely dashed,
]

\addplot  [very thick, color =olive, mark = none, mark size = 2, mark repeat = 30, mark phase = 10]
coordinates {
(0,0.104)(1,0.547)(2,0.777)(3,0.8323)(4,0.8508)(5,0.873)(6,0.8842)(7,0.8937)(8,0.9039)(9,0.9056)(10,0.9083)(11,0.9154)(12,0.917)(13,0.9207)(14,0.917)(15,0.92)(16,0.9126)(17,0.9238)(18,0.9293)(19,0.9243)(20,0.9349)(21,0.9318)(22,0.9338)(23,0.9356)(24,0.9342)(25,0.9353)(26,0.9359)(27,0.9377)(28,0.9374)(29,0.936)(30,0.9358)(31,0.941)(32,0.941)(33,0.9429)(34,0.9435)(35,0.9426)(36,0.9404)(37,0.9475)(38,0.9446)(39,0.9462)(40,0.9474)(41,0.9462)(42,0.9445)(43,0.9506)(44,0.9484)(45,0.9513)(46,0.951)(47,0.949)(48,0.9532)(49,0.9529)(50,0.953)(51,0.9519)(52,0.9531)(53,0.9531)(54,0.9527)(55,0.9528)(56,0.9528)(57,0.9522)(58,0.9524)(59,0.9531)(60,0.952)(61,0.953)(62,0.9518)(63,0.9558)(64,0.9558)(65,0.9536)(66,0.9531)(67,0.951)(68,0.953)(69,0.9516)(70,0.9509)(71,0.955)(72,0.9547)(73,0.9516)(74,0.9552)(75,0.9546)(76,0.9558)(77,0.9557)(78,0.9542)(79,0.9555)(80,0.9555)(81,0.9513)(82,0.954)(83,0.9548)(84,0.9555)(85,0.9522)(86,0.9557)(87,0.9548)(88,0.952)(89,0.9561)(90,0.9547)(91,0.956)(92,0.9554)(93,0.9531)(94,0.9568)(95,0.9533)(96,0.9576)(97,0.9543)(98,0.9583)(99,0.956)(100,0.9587)(101,0.9596)(102,0.9553)(103,0.9569)(104,0.9567)(105,0.9588)(106,0.9575)(107,0.9599)(108,0.96)(109,0.9601)(110,0.9585)(111,0.9601)(112,0.9597)(113,0.961)(114,0.961)(115,0.9582)(116,0.9591)(117,0.9598)(118,0.9595)(119,0.9609)(120,0.9563)(121,0.9607)(122,0.9587)(123,0.9579)(124,0.9606)(125,0.9612)(126,0.9599)(127,0.957)(128,0.9586)(129,0.9572)(130,0.9596)(131,0.9577)(132,0.958)(133,0.959)(134,0.958)(135,0.9581)(136,0.9608)(137,0.9589)(138,0.9599)(139,0.9592)(140,0.9604)(141,0.9595)(142,0.9605)(143,0.9592)(144,0.9567)(145,0.958)(146,0.9574)(147,0.961)(148,0.9605)(149,0.9579)(150,0.9588)(151,0.9584)(152,0.9569)(153,0.9593)(154,0.9583)(155,0.959)(156,0.9572)(157,0.9589)(158,0.9567)(159,0.9555)(160,0.9582)(161,0.957)(162,0.9546)(163,0.958)(164,0.9571)(165,0.9564)(166,0.9592)(167,0.9576)(168,0.9587)(169,0.9569)(170,0.9552)(171,0.9561)(172,0.9588)(173,0.9574)(174,0.959)(175,0.9577)(176,0.9575)(177,0.9549)(178,0.96)(179,0.9589)(180,0.9575)(181,0.9573)(182,0.959)(183,0.9586)(184,0.9593)(185,0.9607)(186,0.9611)(187,0.9606)(188,0.9611)(189,0.9622)(190,0.9606)(191,0.9619)(192,0.9595)(193,0.9598)(194,0.9594)(195,0.9611)(196,0.9613)(197,0.9612)(198,0.9596)(199,0.9608)(200,0.9602)(201,0.9603)(202,0.9616)(203,0.9626)(204,0.9631)(205,0.9634)(206,0.9596)(207,0.9611)(208,0.9639)(209,0.9627)(210,0.9624)(211,0.9639)(212,0.9634)(213,0.9634)(214,0.9618)(215,0.9612)(216,0.9631)(217,0.9619)(218,0.9628)(219,0.9622)(220,0.9627)(221,0.9599)(222,0.9626)(223,0.9641)(224,0.9639)(225,0.965)(226,0.9622)(227,0.9616)(228,0.964)(229,0.9639)(230,0.9632)(231,0.9648)(232,0.9649)(233,0.9651)(234,0.9641)(235,0.9625)(236,0.964)(237,0.9654)(238,0.9638)(239,0.9649)(240,0.9635)(241,0.9624)(242,0.9649)(243,0.9643)(244,0.9656)(245,0.9666)(246,0.9664)(247,0.9673)(248,0.9652)(249,0.9641)(250,0.9648)(251,0.9653)(252,0.965)(253,0.9644)(254,0.9648)(255,0.9653)(256,0.9654)(257,0.9651)(258,0.9653)(259,0.9647)(260,0.9649)(261,0.9635)(262,0.9644)(263,0.9657)(264,0.966)(265,0.9628)(266,0.9656)(267,0.9635)(268,0.9631)(269,0.9639)(270,0.9654)(271,0.9639)(272,0.9649)(273,0.9654)(274,0.9635)(275,0.9653)(276,0.9648)(277,0.9634)(278,0.966)(279,0.9634)(280,0.9669)(281,0.966)(282,0.9664)(283,0.9674)(284,0.967)(285,0.9678)(286,0.9662)(287,0.9668)(288,0.9659)(289,0.9632)(290,0.9656)(291,0.9645)(292,0.9649)(293,0.9645)(294,0.9626)(295,0.9638)(296,0.9614)(297,0.9636)(298,0.965)(299,0.9661)
}; \addlegendentry{$\alpha_t = 0.5$,   AP-UT-1}

\addplot  [semithick, color = black, mark = none, loosely dashed, mark size = 2, mark repeat = 30, mark phase = 10]
coordinates {
(0,0.104)(1,0.547)(2,0.777)(3,0.8323)(4,0.8508)(5,0.873)(6,0.8842)(7,0.8937)(8,0.9039)(9,0.9056)(10,0.9083)(11,0.9154)(12,0.917)(13,0.9207)(14,0.917)(15,0.92)(16,0.9126)(17,0.9238)(18,0.9293)(19,0.9243)(20,0.9349)(21,0.9318)(22,0.9338)(23,0.9356)(24,0.9342)(25,0.9353)(26,0.9359)(27,0.9377)(28,0.9374)(29,0.936)(30,0.9358)(31,0.941)(32,0.941)(33,0.9429)(34,0.9435)(35,0.9426)(36,0.9404)(37,0.9475)(38,0.9446)(39,0.9462)(40,0.9474)(41,0.9462)(42,0.9445)(43,0.9506)(44,0.9484)(45,0.9513)(46,0.951)(47,0.949)(48,0.9532)(49,0.9529)(50,0.953)(51,0.9519)(52,0.9531)(53,0.9531)(54,0.9527)(55,0.9528)(56,0.9528)(57,0.9522)(58,0.9524)(59,0.9531)(60,0.952)(61,0.953)(62,0.9518)(63,0.9558)(64,0.9558)(65,0.9536)(66,0.9531)(67,0.951)(68,0.953)(69,0.9516)(70,0.9509)(71,0.955)(72,0.9547)(73,0.9516)(74,0.9552)(75,0.9546)(76,0.9558)(77,0.9557)(78,0.9542)(79,0.9555)(80,0.9555)(81,0.9513)(82,0.954)(83,0.9548)(84,0.9555)(85,0.9522)(86,0.9557)(87,0.9548)(88,0.952)(89,0.9561)(90,0.9547)(91,0.956)(92,0.9554)(93,0.9531)(94,0.9568)(95,0.9533)(96,0.9576)(97,0.9543)(98,0.9583)(99,0.956)(100,0.9587)(101,0.9596)(102,0.9553)(103,0.9569)(104,0.9567)(105,0.9588)(106,0.9575)(107,0.9599)(108,0.96)(109,0.9601)(110,0.9585)(111,0.9601)(112,0.9597)(113,0.961)(114,0.961)(115,0.9582)(116,0.9591)(117,0.9598)(118,0.9595)(119,0.9609)(120,0.9563)(121,0.9607)(122,0.9587)(123,0.9579)(124,0.9606)(125,0.9612)(126,0.9599)(127,0.957)(128,0.9586)(129,0.9572)(130,0.9596)(131,0.9577)(132,0.958)(133,0.959)(134,0.958)(135,0.9581)(136,0.9608)(137,0.9589)(138,0.9599)(139,0.9592)(140,0.9604)(141,0.9595)(142,0.9605)(143,0.9592)(144,0.9567)(145,0.958)(146,0.9574)(147,0.961)(148,0.9605)(149,0.9579)(150,0.9588)(151,0.9584)(152,0.9569)(153,0.9593)(154,0.9583)(155,0.959)(156,0.9572)(157,0.9589)(158,0.9567)(159,0.9555)(160,0.9582)(161,0.957)(162,0.9546)(163,0.958)(164,0.9571)(165,0.9564)(166,0.9592)(167,0.9576)(168,0.9587)(169,0.9569)(170,0.9552)(171,0.9561)(172,0.9588)(173,0.9574)(174,0.959)(175,0.9577)(176,0.9575)(177,0.9549)(178,0.96)(179,0.9589)(180,0.9575)(181,0.9573)(182,0.959)(183,0.9586)(184,0.9593)(185,0.9607)(186,0.9611)(187,0.9606)(188,0.9611)(189,0.9622)(190,0.9606)(191,0.9619)(192,0.9595)(193,0.9598)(194,0.9594)(195,0.9611)(196,0.9613)(197,0.9612)(198,0.9596)(199,0.9608)(200,0.9602)(201,0.9603)(202,0.9616)(203,0.9626)(204,0.9631)(205,0.9634)(206,0.9596)(207,0.9611)(208,0.9639)(209,0.9627)(210,0.9624)(211,0.9639)(212,0.9634)(213,0.9634)(214,0.9618)(215,0.9612)(216,0.9631)(217,0.9619)(218,0.9628)(219,0.9622)(220,0.9627)(221,0.9599)(222,0.9626)(223,0.9641)(224,0.9639)(225,0.965)(226,0.9622)(227,0.9616)(228,0.964)(229,0.9639)(230,0.9632)(231,0.9648)(232,0.9649)(233,0.9651)(234,0.9641)(235,0.9625)(236,0.964)(237,0.9654)(238,0.9638)(239,0.9649)(240,0.9635)(241,0.9624)(242,0.9649)(243,0.9643)(244,0.9656)(245,0.9666)(246,0.9664)(247,0.9673)(248,0.9652)(249,0.9641)(250,0.9648)(251,0.9653)(252,0.965)(253,0.9644)(254,0.9648)(255,0.9653)(256,0.9654)(257,0.9651)(258,0.9653)(259,0.9647)(260,0.9649)(261,0.9635)(262,0.9644)(263,0.9657)(264,0.966)(265,0.9628)(266,0.9656)(267,0.9635)(268,0.9631)(269,0.9639)(270,0.9654)(271,0.9639)(272,0.9649)(273,0.9654)(274,0.9635)(275,0.9653)(276,0.9648)(277,0.9634)(278,0.966)(279,0.9634)(280,0.9669)(281,0.966)(282,0.9664)(283,0.9674)(284,0.967)(285,0.9678)(286,0.9662)(287,0.9668)(288,0.9659)(289,0.9632)(290,0.9656)(291,0.9645)(292,0.9649)(293,0.9645)(294,0.9626)(295,0.9638)(296,0.9614)(297,0.9636)(298,0.965)(299,0.9661)
}; \addlegendentry{$\alpha_t = 0.5$,   AP-UT-2}

\addplot [semithick, color =blue, mark =  square*,  mark size = 2, mark repeat = 30, mark phase = 10]
coordinates {
(0,0.104)(1,0.178)(2,0.4269)(3,0.4099)(4,0.4178)(5,0.5381)(6,0.5986)(7,0.6553)(8,0.6473)(9,0.6066)(10,0.7016)(11,0.7148)(12,0.7416)(13,0.6792)(14,0.7369)(15,0.6802)(16,0.6916)(17,0.7455)(18,0.751)(19,0.7253)(20,0.7711)(21,0.7269)(22,0.7697)(23,0.7211)(24,0.7481)(25,0.7481)(26,0.8031)(27,0.8014)(28,0.7934)(29,0.7593)(30,0.7893)(31,0.7983)(32,0.7923)(33,0.7987)(34,0.7778)(35,0.8109)(36,0.7805)(37,0.8052)(38,0.7797)(39,0.7914)(40,0.8113)(41,0.8042)(42,0.8275)(43,0.7778)(44,0.7913)(45,0.81)(46,0.8233)(47,0.8468)(48,0.8391)(49,0.8444)(50,0.845)(51,0.8551)(52,0.8263)(53,0.8314)(54,0.8086)(55,0.8325)(56,0.829)(57,0.832)(58,0.8336)(59,0.8582)(60,0.8549)(61,0.8627)(62,0.8623)(63,0.8369)(64,0.8211)(65,0.8457)(66,0.8641)(67,0.8574)(68,0.8468)(69,0.8704)(70,0.8552)(71,0.871)(72,0.8574)(73,0.8709)(74,0.8725)(75,0.8716)(76,0.8687)(77,0.8648)(78,0.8663)(79,0.8757)(80,0.8628)(81,0.8485)(82,0.8608)(83,0.8618)(84,0.8585)(85,0.8633)(86,0.8509)(87,0.8652)(88,0.863)(89,0.8594)(90,0.8639)(91,0.8674)(92,0.8737)(93,0.8722)(94,0.8667)(95,0.8545)(96,0.8517)(97,0.8587)(98,0.8604)(99,0.8531)(100,0.8579)(101,0.8588)(102,0.8638)(103,0.8367)(104,0.8271)(105,0.8047)(106,0.8379)(107,0.8309)(108,0.8221)(109,0.8458)(110,0.8393)(111,0.8454)(112,0.8216)(113,0.8041)(114,0.8316)(115,0.8358)(116,0.8391)(117,0.8279)(118,0.8443)(119,0.8492)(120,0.8354)(121,0.8384)(122,0.8258)(123,0.8286)(124,0.8263)(125,0.794)(126,0.795)(127,0.7745)(128,0.7722)(129,0.7442)(130,0.7524)(131,0.7394)(132,0.7188)(133,0.7041)(134,0.7202)(135,0.7032)(136,0.7263)(137,0.7229)(138,0.751)(139,0.7056)(140,0.6915)(141,0.7032)(142,0.6952)(143,0.6466)(144,0.7088)(145,0.7173)(146,0.6657)(147,0.6542)(148,0.682)(149,0.6315)(150,0.6505)(151,0.6928)(152,0.7251)(153,0.6836)(154,0.664)(155,0.6708)(156,0.7211)(157,0.7149)(158,0.6583)(159,0.6331)(160,0.6755)(161,0.6846)(162,0.6861)(163,0.6789)(164,0.6472)(165,0.6712)(166,0.649)(167,0.6565)(168,0.6099)(169,0.6443)(170,0.6035)(171,0.6593)(172,0.6571)(173,0.6819)(174,0.6377)(175,0.6806)(176,0.6681)(177,0.6376)(178,0.6537)(179,0.65)(180,0.7045)(181,0.6604)(182,0.6884)(183,0.654)(184,0.6539)(185,0.6317)(186,0.6402)(187,0.6854)(188,0.6871)(189,0.6214)(190,0.6137)(191,0.5247)(192,0.6489)(193,0.4942)(194,0.5532)(195,0.5988)(196,0.5241)(197,0.5984)(198,0.5868)(199,0.5729)(200,0.609)(201,0.5062)(202,0.6074)(203,0.6296)(204,0.5355)(205,0.4727)(206,0.6214)(207,0.5366)(208,0.6111)(209,0.5924)(210,0.5886)(211,0.5818)(212,0.6019)(213,0.5396)(214,0.5591)(215,0.5486)(216,0.5836)(217,0.6011)(218,0.5303)(219,0.5704)(220,0.5599)(221,0.5072)(222,0.4889)(223,0.5525)(224,0.5629)(225,0.5801)(226,0.5196)(227,0.4671)(228,0.5078)(229,0.5551)(230,0.5757)(231,0.5256)(232,0.5294)(233,0.5437)(234,0.571)(235,0.5699)(236,0.4797)(237,0.58)(238,0.5423)(239,0.5823)(240,0.5497)(241,0.5796)(242,0.5753)(243,0.4959)(244,0.4682)(245,0.5727)(246,0.4752)(247,0.5849)(248,0.5887)(249,0.575)(250,0.6022)(251,0.5674)(252,0.5161)(253,0.5752)(254,0.4792)(255,0.5373)(256,0.609)(257,0.4494)(258,0.5783)(259,0.4746)(260,0.547)(261,0.5367)(262,0.5646)(263,0.5792)(264,0.5756)(265,0.5171)(266,0.4798)(267,0.489)(268,0.5452)(269,0.5522)(270,0.5438)(271,0.4898)(272,0.4729)(273,0.4956)(274,0.4309)(275,0.4372)(276,0.5464)(277,0.5056)(278,0.556)(279,0.5418)(280,0.4744)(281,0.577)(282,0.523)(283,0.4189)(284,0.5066)(285,0.5452)(286,0.5384)(287,0.5369)(288,0.5097)(289,0.5094)(290,0.533)(291,0.5215)(292,0.5243)(293,0.4927)(294,0.4921)(295,0.535)(296,0.4901)(297,0.4981)(298,0.3774)(299,0.446)
}; \addlegendentry{$\alpha_t = 0.1$,   AP-UT-1}
\addplot  [semithick, color =red, mark = triangle*, mark size = 2, mark repeat = 30, mark phase = 10]
coordinates {

(0,0.104)(1,0.2297)(2,0.3945)(3,0.4913)(4,0.5096)(5,0.6089)(6,0.5845)(7,0.6789)(8,0.6689)(9,0.6038)(10,0.6616)(11,0.6913)(12,0.6649)(13,0.6751)(14,0.7209)(15,0.683)(16,0.7465)(17,0.7143)(18,0.6709)(19,0.7656)(20,0.7443)(21,0.785)(22,0.7875)(23,0.769)(24,0.7991)(25,0.7987)(26,0.8019)(27,0.8131)(28,0.8194)(29,0.7947)(30,0.8053)(31,0.8415)(32,0.8401)(33,0.7903)(34,0.8204)(35,0.8242)(36,0.8096)(37,0.8132)(38,0.8156)(39,0.8285)(40,0.8484)(41,0.842)(42,0.8461)(43,0.805)(44,0.8383)(45,0.8464)(46,0.8299)(47,0.81)(48,0.8369)(49,0.8285)(50,0.8629)(51,0.8544)(52,0.8703)(53,0.8555)(54,0.8745)(55,0.8559)(56,0.8351)(57,0.8408)(58,0.8632)(59,0.8293)(60,0.8578)(61,0.8506)(62,0.8509)(63,0.8369)(64,0.8661)(65,0.8643)(66,0.8677)(67,0.8513)(68,0.8741)(69,0.8689)(70,0.8789)(71,0.8675)(72,0.8561)(73,0.8592)(74,0.8488)(75,0.8747)(76,0.8504)(77,0.8592)(78,0.8534)(79,0.8807)(80,0.8796)(81,0.8532)(82,0.8611)(83,0.8669)(84,0.8702)(85,0.8693)(86,0.8676)(87,0.8692)(88,0.8654)(89,0.8646)(90,0.8539)(91,0.8573)(92,0.8658)(93,0.8586)(94,0.8502)(95,0.8636)(96,0.8612)(97,0.859)(98,0.8549)(99,0.8556)(100,0.8523)(101,0.8676)(102,0.8634)(103,0.8617)(104,0.8677)(105,0.8572)(106,0.861)(107,0.8696)(108,0.8598)(109,0.8508)(110,0.8751)(111,0.8573)(112,0.869)(113,0.8725)(114,0.8633)(115,0.8578)(116,0.864)(117,0.8734)(118,0.8555)(119,0.8587)(120,0.8561)(121,0.8451)(122,0.8494)(123,0.8536)(124,0.8394)(125,0.8224)(126,0.8002)(127,0.7953)(128,0.839)(129,0.8319)(130,0.8113)(131,0.8188)(132,0.7862)(133,0.7906)(134,0.8203)(135,0.8095)(136,0.798)(137,0.7887)(138,0.8121)(139,0.8052)(140,0.7949)(141,0.7684)(142,0.7854)(143,0.7395)(144,0.8007)(145,0.7796)(146,0.7974)(147,0.7995)(148,0.7593)(149,0.7547)(150,0.7782)(151,0.7745)(152,0.7908)(153,0.7183)(154,0.672)(155,0.7072)(156,0.7599)(157,0.7697)(158,0.767)(159,0.7482)(160,0.7768)(161,0.7805)(162,0.7649)(163,0.7265)(164,0.7351)(165,0.7373)(166,0.6966)(167,0.709)(168,0.722)(169,0.7313)(170,0.7298)(171,0.7144)(172,0.7093)(173,0.7294)(174,0.7497)(175,0.7451)(176,0.7411)(177,0.6884)(178,0.7467)(179,0.7351)(180,0.7378)(181,0.719)(182,0.7587)(183,0.7491)(184,0.6933)(185,0.7576)(186,0.7564)(187,0.6955)(188,0.6798)(189,0.7378)(190,0.7356)(191,0.7451)(192,0.6994)(193,0.7271)(194,0.749)(195,0.7522)(196,0.7583)(197,0.7145)(198,0.732)(199,0.7731)(200,0.7659)(201,0.7351)(202,0.7451)(203,0.7722)(204,0.7196)(205,0.7219)(206,0.7188)(207,0.7511)(208,0.695)(209,0.6681)(210,0.6918)(211,0.7466)(212,0.6126)(213,0.7194)(214,0.7136)(215,0.7075)(216,0.7113)(217,0.6707)(218,0.6924)(219,0.6628)(220,0.6463)(221,0.745)(222,0.7318)(223,0.6908)(224,0.6933)(225,0.6609)(226,0.7462)(227,0.7734)(228,0.7351)(229,0.7211)(230,0.6757)(231,0.748)(232,0.7359)(233,0.7316)(234,0.6899)(235,0.7565)(236,0.6844)(237,0.7795)(238,0.6933)(239,0.6898)(240,0.6042)(241,0.7409)(242,0.6594)(243,0.6707)(244,0.676)(245,0.7442)(246,0.7386)(247,0.7502)(248,0.7098)(249,0.6676)(250,0.5602)(251,0.7031)(252,0.7131)(253,0.6456)(254,0.7824)(255,0.6531)(256,0.6815)(257,0.7062)(258,0.7171)(259,0.7004)(260,0.7648)(261,0.736)(262,0.7322)(263,0.7006)(264,0.7202)(265,0.64)(266,0.7156)(267,0.6752)(268,0.7509)(269,0.6772)(270,0.6981)(271,0.6111)(272,0.6739)(273,0.6666)(274,0.7179)(275,0.6319)(276,0.6836)(277,0.6285)(278,0.6983)(279,0.6704)(280,0.6941)(281,0.5849)(282,0.6889)(283,0.7113)(284,0.6639)(285,0.6053)(286,0.6563)(287,0.6469)(288,0.6514)(289,0.5947)(290,0.6657)(291,0.6909)(292,0.6719)(293,0.6756)(294,0.6497)(295,0.5924)(296,0.7013)(297,0.7086)(298,0.6769)(299,0.6813)
}; \addlegendentry{$\alpha_t = 0.1$,  AP-UT-2} 

\end{axis}
\end{tikzpicture}
\end{center}
\caption{Test accuracy vs. the number of global iterations using i.i.d. MNIST dataset for $N=50$, $L = 100$ and $\tau=12$. "AP-UT-1" stands for the AP-UT association scheme proposed in \cite{Emilcellfree} while "AP-UT-2" refers to the scheme developed in  \cite{chen2022improving}.}
\label{AP_UT}
\end{figure}
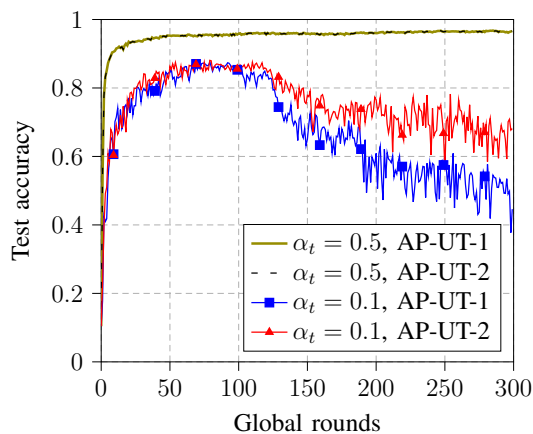

\subsection{Performance on CIFAR10 dataset}
In \figref{cifar_sim}, we consider the more challenging CIFAR10 dataset. We report the performance of the proposed OTA-FL scheme for different values of the power scaling factor $\alpha_t$. Similarly to the MNIST experiment, the proposed implementation guarantees near-optimal performance while requiring low transmit power at the clients.

\begin{figure}[h]
\begin{center}
\subfigure[Perfect CSI]{
\begin{tikzpicture}[scale=0.8]
\begin{axis}[
tick align=outside,
tick pos=left,
x grid style={white!69.0196078431373!black},
xlabel={Global rounds},
xmajorgrids,
xmin=0, xmax=300,
xtick style={color=black},
y grid style={white!69.0196078431373!black},
ylabel={Test accuracy},
ymajorgrids,
ymin=0, ymax=1,
ytick style={color=black},
grid=major,
scaled ticks=true,
legend pos=south east,	
grid style=densely dashed,
]

\addplot [semithick, color=black]
coordinates {
(0,0.0938)(1,0.1458)(2,0.1315)(3,0.2049)(4,0.205)(5,0.2376)(6,0.289)(7,0.2927)(8,0.3394)(9,0.3214)(10,0.3795)(11,0.3834)(12,0.4124)(13,0.4196)(14,0.4459)(15,0.4604)(16,0.4757)(17,0.4795)(18,0.4879)(19,0.4984)(20,0.5203)(21,0.5191)(22,0.5413)(23,0.5454)(24,0.5623)(25,0.56)(26,0.5714)(27,0.5744)(28,0.5847)(29,0.5873)(30,0.5973)(31,0.6067)(32,0.6078)(33,0.6152)(34,0.6166)(35,0.6165)(36,0.6257)(37,0.631)(38,0.6426)(39,0.6431)(40,0.6462)(41,0.6493)(42,0.653)(43,0.6607)(44,0.6616)(45,0.6657)(46,0.666)(47,0.6718)(48,0.6744)(49,0.6806)(50,0.6844)(51,0.6843)(52,0.6872)(53,0.6857)(54,0.695)(55,0.6949)(56,0.6995)(57,0.6979)(58,0.7007)(59,0.705)(60,0.7074)(61,0.7068)(62,0.7079)(63,0.7117)(64,0.7126)(65,0.714)(66,0.7134)(67,0.7194)(68,0.7164)(69,0.7189)(70,0.7194)(71,0.722)(72,0.7265)(73,0.7263)(74,0.7247)(75,0.7299)(76,0.7308)(77,0.7329)(78,0.7328)(79,0.737)(80,0.74)(81,0.7433)(82,0.7446)(83,0.7432)(84,0.7429)(85,0.7473)(86,0.7432)(87,0.7476)(88,0.7485)(89,0.7525)(90,0.7516)(91,0.7513)(92,0.7509)(93,0.7504)(94,0.7527)(95,0.7485)(96,0.7582)(97,0.7582)(98,0.7559)(99,0.7568)(100,0.758)(101,0.7584)(102,0.7621)(103,0.7601)(104,0.7643)(105,0.7632)(106,0.7663)(107,0.7665)(108,0.7662)(109,0.7685)(110,0.7662)(111,0.767)(112,0.7687)(113,0.7714)(114,0.7688)(115,0.7735)(116,0.7725)(117,0.771)(118,0.7716)(119,0.775)(120,0.7695)(121,0.7757)(122,0.7776)(123,0.7766)(124,0.7739)(125,0.7735)(126,0.7775)(127,0.7769)(128,0.7773)(129,0.775)(130,0.7758)(131,0.7784)(132,0.784)(133,0.7824)(134,0.7814)(135,0.7749)(136,0.7785)(137,0.7813)(138,0.7815)(139,0.7811)(140,0.7875)(141,0.7845)(142,0.7847)(143,0.781)(144,0.7821)(145,0.7846)(146,0.7841)(147,0.7879)(148,0.7853)(149,0.786)(150,0.7857)(151,0.7894)(152,0.7878)(153,0.7889)(154,0.7903)(155,0.786)(156,0.7889)(157,0.7865)(158,0.7884)(159,0.7919)(160,0.7925)(161,0.7916)(162,0.7834)(163,0.7885)(164,0.7871)(165,0.7859)(166,0.789)(167,0.7908)(168,0.7881)(169,0.7901)(170,0.7944)(171,0.7921)(172,0.792)(173,0.7894)(174,0.7918)(175,0.7886)(176,0.795)(177,0.7939)(178,0.7943)(179,0.7931)(180,0.7893)(181,0.795)(182,0.7941)(183,0.7955)(184,0.792)(185,0.7899)(186,0.7885)(187,0.791)(188,0.7919)(189,0.7921)(190,0.7941)(191,0.7945)(192,0.7935)(193,0.7956)(194,0.7978)(195,0.7985)(196,0.7902)(197,0.8001)(198,0.7954)(199,0.798)(200,0.7995)(201,0.7941)(202,0.7963)(203,0.7959)(204,0.7964)(205,0.7957)(206,0.7969)(207,0.7987)(208,0.7974)(209,0.7972)(210,0.7979)(211,0.7953)(212,0.8002)(213,0.8025)(214,0.7996)(215,0.7968)(216,0.7988)(217,0.8018)(218,0.7995)(219,0.798)(220,0.798)(221,0.7983)(222,0.7997)(223,0.7991)(224,0.7954)(225,0.7962)(226,0.7985)(227,0.7992)(228,0.7957)(229,0.7969)(230,0.7973)(231,0.7932)(232,0.7958)(233,0.7955)(234,0.8002)(235,0.8019)(236,0.8019)(237,0.7983)(238,0.7985)(239,0.7998)(240,0.7992)(241,0.804)(242,0.7975)(243,0.7985)(244,0.8012)(245,0.8054)(246,0.7976)(247,0.7996)(248,0.8036)(249,0.8032)(250,0.8007)(251,0.8032)(252,0.7992)(253,0.8008)(254,0.8025)(255,0.7964)(256,0.7948)(257,0.7996)(258,0.8014)(259,0.8008)(260,0.8049)(261,0.7992)(262,0.801)(263,0.8044)(264,0.7973)(265,0.8032)(266,0.8025)(267,0.8015)(268,0.7991)(269,0.8047)(270,0.8002)(271,0.7998)(272,0.7996)(273,0.7997)(274,0.8014)(275,0.797)(276,0.8036)(277,0.8057)(278,0.803)(279,0.8071)(280,0.799)(281,0.8044)(282,0.8017)(283,0.8063)(284,0.7985)(285,0.7983)(286,0.8054)(287,0.7968)(288,0.8006)(289,0.7971)(290,0.8044)(291,0.8042)(292,0.8043)(293,0.8025)(294,0.8005)(295,0.7989)(296,0.798)(297,0.8068)(298,0.8013)(299,0.8047)}; \addlegendentry{Perfect links}

\addplot [semithick, color =blue, mark =  square*,  mark size = 2, mark repeat = 30, mark phase = 10]
coordinates {
(0,0.0938)(1,0.1346)(2,0.1183)(3,0.1478)(4,0.1747)(5,0.2202)(6,0.2324)(7,0.2351)(8,0.2585)(9,0.2957)(10,0.2876)(11,0.3157)(12,0.3127)(13,0.3569)(14,0.3588)(15,0.3841)(16,0.3912)(17,0.3998)(18,0.4141)(19,0.4356)(20,0.4304)(21,0.4447)(22,0.4484)(23,0.4552)(24,0.4585)(25,0.4692)(26,0.4749)(27,0.4757)(28,0.4859)(29,0.4784)(30,0.4898)(31,0.4937)(32,0.5053)(33,0.5141)(34,0.5081)(35,0.5149)(36,0.5285)(37,0.5334)(38,0.5339)(39,0.5403)(40,0.5386)(41,0.5511)(42,0.5459)(43,0.5555)(44,0.5645)(45,0.5571)(46,0.5683)(47,0.5715)(48,0.5721)(49,0.5781)(50,0.5754)(51,0.5835)(52,0.5975)(53,0.5891)(54,0.5965)(55,0.6009)(56,0.6057)(57,0.6004)(58,0.6098)(59,0.613)(60,0.6137)(61,0.6225)(62,0.623)(63,0.6205)(64,0.6209)(65,0.6312)(66,0.6212)(67,0.6365)(68,0.6403)(69,0.6464)(70,0.6396)(71,0.6474)(72,0.6457)(73,0.6532)(74,0.6585)(75,0.6613)(76,0.6437)(77,0.6599)(78,0.6623)(79,0.6614)(80,0.674)(81,0.6681)(82,0.6769)(83,0.6723)(84,0.6767)(85,0.6659)(86,0.6697)(87,0.6771)(88,0.6752)(89,0.6875)(90,0.6848)(91,0.6889)(92,0.6892)(93,0.6904)(94,0.6905)(95,0.6942)(96,0.6962)(97,0.6983)(98,0.7004)(99,0.6967)(100,0.6967)(101,0.7066)(102,0.7026)(103,0.708)(104,0.7051)(105,0.7121)(106,0.7087)(107,0.7065)(108,0.7152)(109,0.7141)(110,0.7128)(111,0.7098)(112,0.7121)(113,0.7158)(114,0.718)(115,0.7116)(116,0.7199)(117,0.7143)(118,0.7191)(119,0.7206)(120,0.7289)(121,0.7236)(122,0.7253)(123,0.7208)(124,0.7302)(125,0.7318)(126,0.7316)(127,0.7279)(128,0.7295)(129,0.7286)(130,0.7366)(131,0.7366)(132,0.7314)(133,0.7363)(134,0.7406)(135,0.737)(136,0.734)(137,0.7386)(138,0.7403)(139,0.7397)(140,0.7435)(141,0.7431)(142,0.7401)(143,0.7408)(144,0.7361)(145,0.7429)(146,0.7397)(147,0.7421)(148,0.7419)(149,0.7431)(150,0.7456)(151,0.7515)(152,0.7512)(153,0.7524)(154,0.7504)(155,0.7539)(156,0.748)(157,0.7549)(158,0.7513)(159,0.7559)(160,0.7439)(161,0.7535)(162,0.7461)(163,0.7506)(164,0.7545)(165,0.754)(166,0.7545)(167,0.753)(168,0.7577)(169,0.763)(170,0.7523)(171,0.7527)(172,0.7529)(173,0.7518)(174,0.7519)(175,0.7557)(176,0.7643)(177,0.7594)(178,0.7598)(179,0.7612)(180,0.7619)(181,0.7689)(182,0.7638)(183,0.7608)(184,0.7596)(185,0.7633)(186,0.7654)(187,0.7649)(188,0.7656)(189,0.7612)(190,0.7651)(191,0.7645)(192,0.7656)(193,0.7625)(194,0.7704)(195,0.7656)(196,0.7629)(197,0.7644)(198,0.7657)(199,0.7658)(200,0.7662)(201,0.7652)(202,0.7657)(203,0.7624)(204,0.7637)(205,0.7648)(206,0.7686)(207,0.7732)(208,0.7653)(209,0.7735)(210,0.7706)(211,0.7681)(212,0.7727)(213,0.7704)(214,0.7705)(215,0.7702)(216,0.7707)(217,0.7732)(218,0.7693)(219,0.7704)(220,0.7717)(221,0.7781)(222,0.7722)(223,0.7731)(224,0.7728)(225,0.7771)(226,0.77)(227,0.7661)(228,0.7712)(229,0.7728)(230,0.7754)(231,0.7692)(232,0.7796)(233,0.7738)(234,0.7746)(235,0.7761)(236,0.78)(237,0.7735)(238,0.7728)(239,0.783)(240,0.7779)(241,0.7814)(242,0.7782)(243,0.7822)(244,0.7758)(245,0.7804)(246,0.785)(247,0.7742)(248,0.7814)(249,0.7759)(250,0.7786)(251,0.7789)(252,0.7794)(253,0.7795)(254,0.7814)(255,0.7789)(256,0.7829)(257,0.7764)(258,0.779)(259,0.7749)(260,0.7821)(261,0.7834)(262,0.7816)(263,0.7805)(264,0.7823)(265,0.7788)(266,0.782)(267,0.7805)(268,0.7815)(269,0.7868)(270,0.7823)(271,0.7804)(272,0.7787)(273,0.7837)(274,0.7881)(275,0.783)(276,0.7785)(277,0.7827)(278,0.7794)(279,0.7787)(280,0.7834)(281,0.7797)(282,0.7835)(283,0.7848)(284,0.7769)(285,0.7767)(286,0.785)(287,0.7822)(288,0.7877)(289,0.781)(290,0.7844)(291,0.7835)(292,0.7832)(293,0.7872)(294,0.7815)(295,0.7852)(296,0.7826)(297,0.7919)(298,0.7836)(299,0.7892)

}; \addlegendentry{$\alpha_t = 0.2, P =49.8$ dBm}
\addplot  [semithick, color =red, mark = triangle*, mark size = 2, mark repeat = 30, mark phase = 10]
coordinates {
(0,0.0938)(1,0.1249)(2,0.1122)(3,0.1326)(4,0.1729)(5,0.1955)(6,0.2118)(7,0.2035)(8,0.224)(9,0.2531)(10,0.2445)(11,0.2508)(12,0.2496)(13,0.2749)(14,0.2734)(15,0.3023)(16,0.2978)(17,0.3151)(18,0.3319)(19,0.3388)(20,0.3363)(21,0.324)(22,0.3337)(23,0.3394)(24,0.3453)(25,0.3451)(26,0.3739)(27,0.3603)(28,0.3733)(29,0.3702)(30,0.3674)(31,0.3823)(32,0.3935)(33,0.3943)(34,0.3863)(35,0.3992)(36,0.4159)(37,0.4039)(38,0.403)(39,0.4066)(40,0.4223)(41,0.4205)(42,0.4332)(43,0.4351)(44,0.4338)(45,0.4198)(46,0.4366)(47,0.4384)(48,0.4372)(49,0.4543)(50,0.4463)(51,0.4477)(52,0.4366)(53,0.4444)(54,0.4575)(55,0.4577)(56,0.4579)(57,0.4426)(58,0.4709)(59,0.4779)(60,0.4795)(61,0.4751)(62,0.4789)(63,0.4861)(64,0.49)(65,0.4935)(66,0.4718)(67,0.4941)(68,0.5052)(69,0.5013)(70,0.502)(71,0.493)(72,0.4973)(73,0.4995)(74,0.5098)(75,0.5104)(76,0.5038)(77,0.5084)(78,0.5131)(79,0.512)(80,0.5288)(81,0.5334)(82,0.5405)(83,0.5269)(84,0.5244)(85,0.5253)(86,0.5393)(87,0.5309)(88,0.5413)(89,0.5417)(90,0.5464)(91,0.55)(92,0.551)(93,0.5496)(94,0.5533)(95,0.5591)(96,0.5529)(97,0.5633)(98,0.5624)(99,0.553)(100,0.5587)(101,0.5627)(102,0.5623)(103,0.5746)(104,0.5766)(105,0.5723)(106,0.5771)(107,0.5721)(108,0.5847)(109,0.5868)(110,0.5911)(111,0.5782)(112,0.5922)(113,0.5858)(114,0.5944)(115,0.5934)(116,0.5877)(117,0.5887)(118,0.5943)(119,0.6007)(120,0.601)(121,0.6038)(122,0.6105)(123,0.6085)(124,0.6127)(125,0.6155)(126,0.6091)(127,0.6109)(128,0.6228)(129,0.6056)(130,0.6197)(131,0.6243)(132,0.6215)(133,0.6271)(134,0.6249)(135,0.6298)(136,0.6334)(137,0.6273)(138,0.628)(139,0.633)(140,0.6328)(141,0.6293)(142,0.6387)(143,0.6336)(144,0.6286)(145,0.6343)(146,0.6379)(147,0.6393)(148,0.6409)(149,0.6431)(150,0.6416)(151,0.6522)(152,0.6466)(153,0.6493)(154,0.6503)(155,0.6518)(156,0.6503)(157,0.6536)(158,0.6535)(159,0.6542)(160,0.6594)(161,0.6525)(162,0.6507)(163,0.6605)(164,0.6639)(165,0.6644)(166,0.6673)(167,0.6593)(168,0.6619)(169,0.6622)(170,0.6678)(171,0.6703)(172,0.6685)(173,0.6718)(174,0.6735)(175,0.6672)(176,0.6712)(177,0.6717)(178,0.6767)(179,0.6734)(180,0.6734)(181,0.6818)(182,0.6777)(183,0.6802)(184,0.6815)(185,0.6815)(186,0.6837)(187,0.6807)(188,0.6816)(189,0.6795)(190,0.6821)(191,0.682)(192,0.6863)(193,0.6793)(194,0.6847)(195,0.6834)(196,0.6846)(197,0.6854)(198,0.6873)(199,0.687)(200,0.6915)(201,0.6915)(202,0.6877)(203,0.6952)(204,0.6948)(205,0.6959)(206,0.6926)(207,0.6944)(208,0.6954)(209,0.6932)(210,0.6902)(211,0.6997)(212,0.6987)(213,0.7021)(214,0.7016)(215,0.6986)(216,0.7032)(217,0.6976)(218,0.7046)(219,0.7012)(220,0.7078)(221,0.7054)(222,0.7111)(223,0.7006)(224,0.7061)(225,0.7062)(226,0.7104)(227,0.7015)(228,0.6986)(229,0.7032)(230,0.7121)(231,0.705)(232,0.709)(233,0.7082)(234,0.7148)(235,0.7134)(236,0.7089)(237,0.7087)(238,0.7161)(239,0.7117)(240,0.7124)(241,0.7144)(242,0.7076)(243,0.7105)(244,0.7094)(245,0.7195)(246,0.7193)(247,0.7162)(248,0.7174)(249,0.7175)(250,0.721)(251,0.7172)(252,0.7194)(253,0.7249)(254,0.7237)(255,0.7208)(256,0.719)(257,0.7246)(258,0.7191)(259,0.7247)(260,0.7201)(261,0.7275)(262,0.7207)(263,0.7288)(264,0.7284)(265,0.7274)(266,0.7239)(267,0.7196)(268,0.732)(269,0.7311)(270,0.7248)(271,0.7281)(272,0.728)(273,0.7277)(274,0.7321)(275,0.7273)(276,0.7266)(277,0.7317)(278,0.7302)(279,0.7283)(280,0.7262)(281,0.7232)(282,0.7286)(283,0.7306)(284,0.725)(285,0.731)(286,0.7314)(287,0.735)(288,0.7327)(289,0.7356)(290,0.7375)(291,0.733)(292,0.7284)(293,0.7287)(294,0.7328)(295,0.727)(296,0.7327)(297,0.7367)(298,0.7367)(299,0.7361)

}; \addlegendentry{$\alpha_t = 0.1, P = 43.5$ dBm}

\end{axis}
\end{tikzpicture}}
\subfigure[Imperfect CSI]{
\begin{tikzpicture}[scale=0.8]
\begin{axis}[
tick align=outside,
tick pos=left,
x grid style={white!69.0196078431373!black},
xlabel={Global rounds},
xmajorgrids,
xmin=0, xmax=300,
xtick style={color=black},
y grid style={white!69.0196078431373!black},
ylabel={Test accuracy},
ymajorgrids,
ymin=0, ymax=1,
ytick style={color=black},
grid=major,
scaled ticks=true,
legend pos=south east,	
grid style=densely dashed,
]

\addplot [semithick, color=black]
coordinates {
(0,0.0938)(1,0.1458)(2,0.1315)(3,0.2049)(4,0.205)(5,0.2376)(6,0.289)(7,0.2927)(8,0.3394)(9,0.3214)(10,0.3795)(11,0.3834)(12,0.4124)(13,0.4196)(14,0.4459)(15,0.4604)(16,0.4757)(17,0.4795)(18,0.4879)(19,0.4984)(20,0.5203)(21,0.5191)(22,0.5413)(23,0.5454)(24,0.5623)(25,0.56)(26,0.5714)(27,0.5744)(28,0.5847)(29,0.5873)(30,0.5973)(31,0.6067)(32,0.6078)(33,0.6152)(34,0.6166)(35,0.6165)(36,0.6257)(37,0.631)(38,0.6426)(39,0.6431)(40,0.6462)(41,0.6493)(42,0.653)(43,0.6607)(44,0.6616)(45,0.6657)(46,0.666)(47,0.6718)(48,0.6744)(49,0.6806)(50,0.6844)(51,0.6843)(52,0.6872)(53,0.6857)(54,0.695)(55,0.6949)(56,0.6995)(57,0.6979)(58,0.7007)(59,0.705)(60,0.7074)(61,0.7068)(62,0.7079)(63,0.7117)(64,0.7126)(65,0.714)(66,0.7134)(67,0.7194)(68,0.7164)(69,0.7189)(70,0.7194)(71,0.722)(72,0.7265)(73,0.7263)(74,0.7247)(75,0.7299)(76,0.7308)(77,0.7329)(78,0.7328)(79,0.737)(80,0.74)(81,0.7433)(82,0.7446)(83,0.7432)(84,0.7429)(85,0.7473)(86,0.7432)(87,0.7476)(88,0.7485)(89,0.7525)(90,0.7516)(91,0.7513)(92,0.7509)(93,0.7504)(94,0.7527)(95,0.7485)(96,0.7582)(97,0.7582)(98,0.7559)(99,0.7568)(100,0.758)(101,0.7584)(102,0.7621)(103,0.7601)(104,0.7643)(105,0.7632)(106,0.7663)(107,0.7665)(108,0.7662)(109,0.7685)(110,0.7662)(111,0.767)(112,0.7687)(113,0.7714)(114,0.7688)(115,0.7735)(116,0.7725)(117,0.771)(118,0.7716)(119,0.775)(120,0.7695)(121,0.7757)(122,0.7776)(123,0.7766)(124,0.7739)(125,0.7735)(126,0.7775)(127,0.7769)(128,0.7773)(129,0.775)(130,0.7758)(131,0.7784)(132,0.784)(133,0.7824)(134,0.7814)(135,0.7749)(136,0.7785)(137,0.7813)(138,0.7815)(139,0.7811)(140,0.7875)(141,0.7845)(142,0.7847)(143,0.781)(144,0.7821)(145,0.7846)(146,0.7841)(147,0.7879)(148,0.7853)(149,0.786)(150,0.7857)(151,0.7894)(152,0.7878)(153,0.7889)(154,0.7903)(155,0.786)(156,0.7889)(157,0.7865)(158,0.7884)(159,0.7919)(160,0.7925)(161,0.7916)(162,0.7834)(163,0.7885)(164,0.7871)(165,0.7859)(166,0.789)(167,0.7908)(168,0.7881)(169,0.7901)(170,0.7944)(171,0.7921)(172,0.792)(173,0.7894)(174,0.7918)(175,0.7886)(176,0.795)(177,0.7939)(178,0.7943)(179,0.7931)(180,0.7893)(181,0.795)(182,0.7941)(183,0.7955)(184,0.792)(185,0.7899)(186,0.7885)(187,0.791)(188,0.7919)(189,0.7921)(190,0.7941)(191,0.7945)(192,0.7935)(193,0.7956)(194,0.7978)(195,0.7985)(196,0.7902)(197,0.8001)(198,0.7954)(199,0.798)(200,0.7995)(201,0.7941)(202,0.7963)(203,0.7959)(204,0.7964)(205,0.7957)(206,0.7969)(207,0.7987)(208,0.7974)(209,0.7972)(210,0.7979)(211,0.7953)(212,0.8002)(213,0.8025)(214,0.7996)(215,0.7968)(216,0.7988)(217,0.8018)(218,0.7995)(219,0.798)(220,0.798)(221,0.7983)(222,0.7997)(223,0.7991)(224,0.7954)(225,0.7962)(226,0.7985)(227,0.7992)(228,0.7957)(229,0.7969)(230,0.7973)(231,0.7932)(232,0.7958)(233,0.7955)(234,0.8002)(235,0.8019)(236,0.8019)(237,0.7983)(238,0.7985)(239,0.7998)(240,0.7992)(241,0.804)(242,0.7975)(243,0.7985)(244,0.8012)(245,0.8054)(246,0.7976)(247,0.7996)(248,0.8036)(249,0.8032)(250,0.8007)(251,0.8032)(252,0.7992)(253,0.8008)(254,0.8025)(255,0.7964)(256,0.7948)(257,0.7996)(258,0.8014)(259,0.8008)(260,0.8049)(261,0.7992)(262,0.801)(263,0.8044)(264,0.7973)(265,0.8032)(266,0.8025)(267,0.8015)(268,0.7991)(269,0.8047)(270,0.8002)(271,0.7998)(272,0.7996)(273,0.7997)(274,0.8014)(275,0.797)(276,0.8036)(277,0.8057)(278,0.803)(279,0.8071)(280,0.799)(281,0.8044)(282,0.8017)(283,0.8063)(284,0.7985)(285,0.7983)(286,0.8054)(287,0.7968)(288,0.8006)(289,0.7971)(290,0.8044)(291,0.8042)(292,0.8043)(293,0.8025)(294,0.8005)(295,0.7989)(296,0.798)(297,0.8068)(298,0.8013)(299,0.8047)}; \addlegendentry{Perfect links}

\addplot [semithick, color =blue, mark =  square*,  mark size = 2, mark repeat = 30, mark phase = 10]
coordinates {
(0,0.0938)(1,0.1377)(2,0.1181)(3,0.1498)(4,0.1687)(5,0.226)(6,0.2309)(7,0.22)(8,0.2591)(9,0.3028)(10,0.2925)(11,0.3137)(12,0.3197)(13,0.3478)(14,0.3459)(15,0.3821)(16,0.3881)(17,0.4051)(18,0.4149)(19,0.4219)(20,0.4399)(21,0.4317)(22,0.4512)(23,0.4495)(24,0.4553)(25,0.4563)(26,0.4627)(27,0.4698)(28,0.4812)(29,0.473)(30,0.492)(31,0.4931)(32,0.4987)(33,0.5039)(34,0.4948)(35,0.5043)(36,0.5219)(37,0.5157)(38,0.5256)(39,0.5374)(40,0.5372)(41,0.5466)(42,0.5402)(43,0.545)(44,0.5544)(45,0.5516)(46,0.5568)(47,0.5681)(48,0.5631)(49,0.5638)(50,0.5653)(51,0.5785)(52,0.5819)(53,0.5904)(54,0.5865)(55,0.5865)(56,0.5946)(57,0.5905)(58,0.6048)(59,0.607)(60,0.6122)(61,0.6163)(62,0.612)(63,0.6175)(64,0.6248)(65,0.623)(66,0.625)(67,0.6331)(68,0.6397)(69,0.6431)(70,0.6372)(71,0.636)(72,0.6398)(73,0.6481)(74,0.6519)(75,0.656)(76,0.646)(77,0.6501)(78,0.6613)(79,0.6556)(80,0.6623)(81,0.6679)(82,0.6722)(83,0.6649)(84,0.674)(85,0.6615)(86,0.6746)(87,0.6802)(88,0.6809)(89,0.6868)(90,0.6833)(91,0.6871)(92,0.686)(93,0.6863)(94,0.6939)(95,0.6907)(96,0.6937)(97,0.6998)(98,0.6874)(99,0.6886)(100,0.693)(101,0.701)(102,0.704)(103,0.7032)(104,0.7045)(105,0.6996)(106,0.7055)(107,0.7048)(108,0.7097)(109,0.7139)(110,0.7125)(111,0.7104)(112,0.7136)(113,0.7183)(114,0.7195)(115,0.7138)(116,0.7225)(117,0.7163)(118,0.7234)(119,0.7272)(120,0.7224)(121,0.7224)(122,0.7291)(123,0.729)(124,0.7294)(125,0.7293)(126,0.7318)(127,0.7332)(128,0.7312)(129,0.729)(130,0.7335)(131,0.7386)(132,0.734)(133,0.7322)(134,0.7365)(135,0.7388)(136,0.7357)(137,0.7368)(138,0.7406)(139,0.7448)(140,0.7463)(141,0.7479)(142,0.7435)(143,0.7445)(144,0.7398)(145,0.739)(146,0.7458)(147,0.747)(148,0.7442)(149,0.7441)(150,0.7525)(151,0.753)(152,0.7545)(153,0.7479)(154,0.7547)(155,0.756)(156,0.7459)(157,0.7498)(158,0.7533)(159,0.7551)(160,0.7527)(161,0.7573)(162,0.7527)(163,0.757)(164,0.7522)(165,0.7564)(166,0.7557)(167,0.7538)(168,0.7549)(169,0.7596)(170,0.7577)(171,0.7573)(172,0.7551)(173,0.7607)(174,0.7593)(175,0.7588)(176,0.7608)(177,0.762)(178,0.7613)(179,0.7625)(180,0.7621)(181,0.7651)(182,0.7657)(183,0.7609)(184,0.7621)(185,0.7653)(186,0.7692)(187,0.7639)(188,0.7648)(189,0.7632)(190,0.7612)(191,0.7691)(192,0.7638)(193,0.7644)(194,0.7677)(195,0.7667)(196,0.7566)(197,0.7648)(198,0.7663)(199,0.7674)(200,0.7623)(201,0.7652)(202,0.7671)(203,0.7661)(204,0.77)(205,0.771)(206,0.767)(207,0.7686)(208,0.7752)(209,0.7716)(210,0.7754)(211,0.7708)(212,0.7714)(213,0.7727)(214,0.7717)(215,0.7711)(216,0.7709)(217,0.7716)(218,0.7724)(219,0.7734)(220,0.7786)(221,0.7763)(222,0.7746)(223,0.7715)(224,0.771)(225,0.773)(226,0.7753)(227,0.7705)(228,0.7786)(229,0.7731)(230,0.78)(231,0.77)(232,0.7763)(233,0.7757)(234,0.7824)(235,0.7765)(236,0.7803)(237,0.7807)(238,0.7799)(239,0.7804)(240,0.7795)(241,0.7808)(242,0.7776)(243,0.7824)(244,0.7762)(245,0.7773)(246,0.7868)(247,0.7782)(248,0.7818)(249,0.7793)(250,0.7876)(251,0.7775)(252,0.7792)(253,0.7843)(254,0.7803)(255,0.7806)(256,0.783)(257,0.7791)(258,0.7821)(259,0.7838)(260,0.7849)(261,0.7835)(262,0.7812)(263,0.7878)(264,0.7905)(265,0.7817)(266,0.7814)(267,0.7876)(268,0.7804)(269,0.785)(270,0.7815)(271,0.7866)(272,0.78)(273,0.7803)(274,0.7881)(275,0.7847)(276,0.7835)(277,0.784)(278,0.7844)(279,0.7899)(280,0.7885)(281,0.7766)(282,0.7856)(283,0.7847)(284,0.7828)(285,0.7852)(286,0.7867)(287,0.7819)(288,0.7875)(289,0.7861)(290,0.7817)(291,0.788)(292,0.7868)(293,0.786)(294,0.7898)(295,0.785)(296,0.7884)(297,0.7919)(298,0.7889)(299,0.7913)

}; \addlegendentry{$\alpha_t = 0.2, P = 49.7$ dBm}
\addplot  [semithick, color =red, mark = triangle*, mark size = 2, mark repeat = 30, mark phase = 10]
coordinates {
(0,0.0938)(1,0.1283)(2,0.1122)(3,0.1309)(4,0.1688)(5,0.1973)(6,0.2097)(7,0.2059)(8,0.209)(9,0.2528)(10,0.232)(11,0.2574)(12,0.2474)(13,0.2781)(14,0.2674)(15,0.2999)(16,0.2919)(17,0.3068)(18,0.3261)(19,0.3333)(20,0.3219)(21,0.3309)(22,0.3296)(23,0.3506)(24,0.3648)(25,0.3442)(26,0.3641)(27,0.3573)(28,0.3787)(29,0.3729)(30,0.3743)(31,0.3878)(32,0.4009)(33,0.3886)(34,0.3899)(35,0.4025)(36,0.4133)(37,0.4115)(38,0.4034)(39,0.4092)(40,0.4163)(41,0.409)(42,0.4309)(43,0.4266)(44,0.4283)(45,0.4294)(46,0.4425)(47,0.4284)(48,0.4366)(49,0.4526)(50,0.454)(51,0.4531)(52,0.4432)(53,0.4473)(54,0.4569)(55,0.4656)(56,0.4524)(57,0.4596)(58,0.4734)(59,0.4763)(60,0.4858)(61,0.4864)(62,0.4801)(63,0.4854)(64,0.4767)(65,0.5015)(66,0.4821)(67,0.4971)(68,0.4929)(69,0.4997)(70,0.5024)(71,0.5003)(72,0.5099)(73,0.494)(74,0.5111)(75,0.5116)(76,0.5019)(77,0.5161)(78,0.5141)(79,0.5187)(80,0.5289)(81,0.5275)(82,0.5255)(83,0.5251)(84,0.5201)(85,0.5255)(86,0.5378)(87,0.5299)(88,0.5353)(89,0.536)(90,0.5419)(91,0.5459)(92,0.5418)(93,0.5412)(94,0.5498)(95,0.5563)(96,0.5503)(97,0.5582)(98,0.5471)(99,0.5463)(100,0.5564)(101,0.5583)(102,0.5588)(103,0.5658)(104,0.5701)(105,0.5703)(106,0.5668)(107,0.5663)(108,0.5782)(109,0.5794)(110,0.5805)(111,0.5716)(112,0.5785)(113,0.5809)(114,0.5891)(115,0.5849)(116,0.584)(117,0.5882)(118,0.5886)(119,0.5993)(120,0.5926)(121,0.5938)(122,0.6038)(123,0.6024)(124,0.6081)(125,0.6039)(126,0.6074)(127,0.605)(128,0.6133)(129,0.6001)(130,0.6161)(131,0.6134)(132,0.6184)(133,0.6175)(134,0.6128)(135,0.6164)(136,0.626)(137,0.6266)(138,0.6185)(139,0.6257)(140,0.6255)(141,0.6301)(142,0.6312)(143,0.6233)(144,0.6254)(145,0.6302)(146,0.6299)(147,0.6286)(148,0.6313)(149,0.6388)(150,0.6362)(151,0.6382)(152,0.6412)(153,0.6401)(154,0.6437)(155,0.643)(156,0.6428)(157,0.6495)(158,0.6483)(159,0.6508)(160,0.6545)(161,0.6506)(162,0.6473)(163,0.6515)(164,0.6574)(165,0.6539)(166,0.6578)(167,0.6602)(168,0.6602)(169,0.6566)(170,0.6618)(171,0.6616)(172,0.6627)(173,0.6601)(174,0.6613)(175,0.6673)(176,0.6677)(177,0.667)(178,0.67)(179,0.6714)(180,0.6737)(181,0.6698)(182,0.6734)(183,0.6679)(184,0.675)(185,0.6717)(186,0.6812)(187,0.6747)(188,0.6854)(189,0.6809)(190,0.6773)(191,0.6844)(192,0.6805)(193,0.679)(194,0.6794)(195,0.6819)(196,0.6881)(197,0.6836)(198,0.6804)(199,0.6786)(200,0.6791)(201,0.6779)(202,0.6874)(203,0.6847)(204,0.6889)(205,0.6842)(206,0.6918)(207,0.6906)(208,0.6944)(209,0.689)(210,0.6963)(211,0.7032)(212,0.6951)(213,0.7001)(214,0.6958)(215,0.7024)(216,0.7002)(217,0.6966)(218,0.6964)(219,0.7019)(220,0.6991)(221,0.7029)(222,0.7015)(223,0.7028)(224,0.7016)(225,0.7068)(226,0.697)(227,0.6963)(228,0.699)(229,0.7008)(230,0.7095)(231,0.6986)(232,0.7045)(233,0.7051)(234,0.7108)(235,0.7067)(236,0.7113)(237,0.7106)(238,0.7079)(239,0.709)(240,0.7155)(241,0.7131)(242,0.7086)(243,0.7096)(244,0.7093)(245,0.7132)(246,0.7124)(247,0.7163)(248,0.7123)(249,0.7121)(250,0.7175)(251,0.7168)(252,0.7196)(253,0.7154)(254,0.7191)(255,0.7169)(256,0.7208)(257,0.7204)(258,0.7198)(259,0.7193)(260,0.7174)(261,0.7223)(262,0.726)(263,0.7228)(264,0.7229)(265,0.726)(266,0.718)(267,0.7254)(268,0.7242)(269,0.7322)(270,0.7253)(271,0.7264)(272,0.7241)(273,0.7242)(274,0.7319)(275,0.724)(276,0.719)(277,0.721)(278,0.723)(279,0.7259)(280,0.7287)(281,0.7259)(282,0.7321)(283,0.7296)(284,0.7265)(285,0.7281)(286,0.7337)(287,0.7333)(288,0.7325)(289,0.7324)(290,0.7318)(291,0.7343)(292,0.7336)(293,0.7369)(294,0.7386)(295,0.7319)(296,0.7323)(297,0.7363)(298,0.7332)(299,0.7326)

}; \addlegendentry{$\alpha_t = 0.1, P = 43.4$ dBm}

\end{axis}
\end{tikzpicture}}
\end{center}
\caption{Test accuracy vs. the number of global iterations using i.i.d. CIFAR10 dataset for $N=20$ and $\tau=5$.}
\label{cifar_sim}
\end{figure}
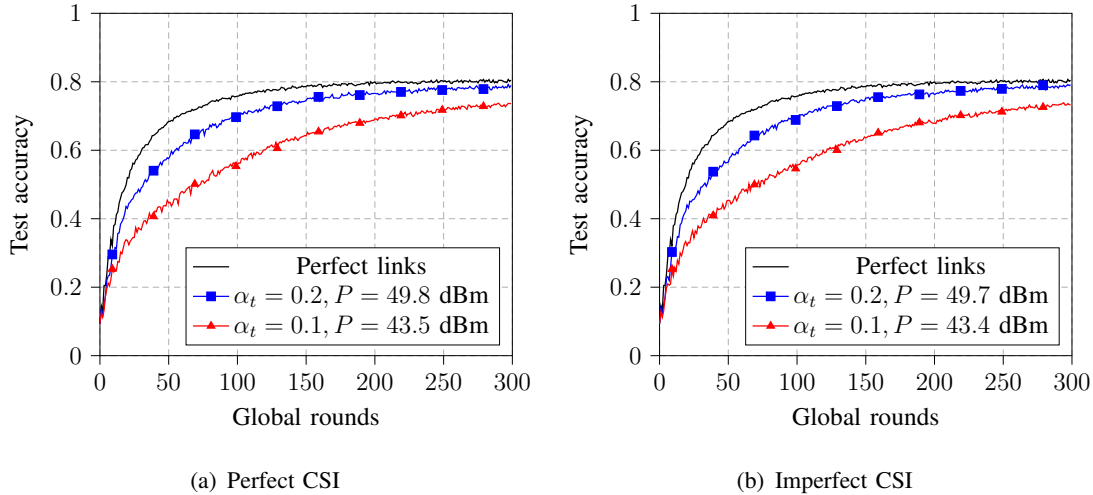

\section{Conclusion}
\label{conc}
In this paper, we have studied the performance of OTA-FL over scalable cell-free massive MIMO. Specifically, we have proposed a practical implementation that is shown to be convergent. The performance of the proposed implementation is studied analytically and experimentally, showing the advantages that cell-free massive MIMO can bring for OTA-FL. We have also compared the proposed scheme with cellular massive MIMO-based OTA-FL, where we have shown that the proposed implementation requires significantly lower energy to reach convergence. Future related research directions include studying different receiving techniques, such as the MMSE receiver, and investigating the convergence of second-order FL techniques, such as ADMM FL, in this context.  Fine-tuning the proposed OTA-FL over cell-free massive MIMO in terms of client scheduling and uplink power optimization is also an interesting future direction.

\section*{Appendix A\\ Proof of Theorem 1}
\label{proofthm1}
As mentioned in Section \ref{conv}, we assume without loss of generality that $d=2S$, which implies $K=1$ and thus we drop the dependency on $k$ for ease of notations. The received signal in \eqref{signal-y_imperfect1} can be rewritten as follows
\begin{align}
y^s(t) = \sum_{p = 1}^4y^s_p(t),
\end{align}
where
\begin{subequations}
\begin{align}
y^s_1(t)&= \frac{\alpha_t}{N}\sum_{n=1}^N \frac{\bh^{s}_{n}(t)^H\bD_n \bh^{s}_{n}(t)}{c_n }\left(\Delta \theta_{n,s}(t) +j \Delta \theta_{n,s+\frac{d}{2}}(t) \ \right)\\
y^s_2(t)&=\frac{\alpha_t}{N}\sum_{n=1}^N \sum_{\substack{n'=1\\n'\neq n}}^N \frac{\bh^{s}_{n}(t)^H\bD_n \bh^{s}_{n'}(t)}{c_n }  \left(\Delta \theta_{n',s}(t) +j \Delta \theta_{n',s+\frac{d}{2}}(t) \ \right)\\
y^s_3(t)&=\frac{\alpha_t}{N}\sum_{n=1}^N \sum_{\substack{n'=1}}^N \frac{ \tilde\bh^s_{n}(t)^H \bD_n \bh_{n'}^s(t)}{c_n}  \left(\Delta \theta_{n',s}(t) +j \Delta \theta_{n',s+\frac{d}{2}}(t) \ \right)\\
y^s_4(t)&=\frac{1}{N}\sum_{n=1}^N\frac{ \left(\bh^{s}_{n}(t)+\tilde\bh^s_{n}(t)\right)^H\bD_n\bz^s(t)}{c_n}.
\end{align}
\label{yis}
\end{subequations}
Accordingly, we define for $p\in[4]$
\begin{align}
\Delta{\widehat{\theta}}_{p,i}(t)&= \begin{cases} \frac{Re(y^i_p(t))}{\alpha_t}, \quad {\rm if} \quad 1\leq i\leq d/2,\\ \frac{Im(y^{i-d/2}_p(t))}{\alpha_t}, \quad {\rm if} \quad  d/2+1\leq i\leq {d},
\end{cases}
\label{theta_hat}
\end{align}
and the estimate of the $i$-the entry of the global model update can be written as
\begin{align*}
\Delta{\widehat{\theta}}_{i}(t)&= \sum_{p = 1}^4\Delta{\widehat{\theta}}_{p,i}(t).
\end{align*}
\noindent Define $\bnu(t)$ as
$$
\bnu(t+1) = \btt(t) + \Delta\btt(t),
$$
where 
$$
\Delta\btt(t) = \frac{1}{N} \sum_{n=1}^N\Delta\btt_n(t).
$$
Recall that
$$
\btt(t+1)  = \btt(t) + \Delta\widehat\btt(t).
$$
The objective is to bound $\ex\left\|\btt(t+1) -\btt^*\right\|^2$, where the expectation is taken over all randomness, namely the channels, the channel estimation error, the noise, and the stochastic gradients. We start by rewriting $\left\|\btt(t+1) -\btt^*\right\|^2$ as
\begin{align}
\left\|\btt(t+1) -\btt^*\right\|^2& = \left\|\btt(t+1) -\bnu(t+1)+\bnu(t+1) -\btt^*\right\|^2\nonumber\\
& = \left\|\btt(t+1) -\bnu(t+1)\right\|^2+\left\|\bnu(t+1) -\btt^*\right\|^2\label{eq1}\\
&\quad +2\left\langle \btt(t+1) -\bnu(t+1), \bnu(t+1) -\btt^*\right\rangle\nonumber.
\end{align}
It takes no effort to verify that 
\begin{align*}
\ex \left\langle \btt(t+1) -\bnu(t+1), \bnu(t+1) -\btt^*\right\rangle &= 
\ex \left\langle \Delta\widehat\btt(t) -\Delta\btt(t), \btt(t)+\Delta\btt(t) -\btt^*\right\rangle = 0, 
\end{align*}
where the last result is obtained by noting that the local model updates at global round $t$ are independent of the channels and noise during the same global round, and the expectation over $\bh_n^s(t)$, $\tilde \bh_n^s(t)$, and $\bz^s(t)$ of $\Delta\widehat\btt(t) $ is $\ex \Delta\widehat\btt(t) = \Delta\btt(t) $.

It remains to bound the first two terms of the RHS of \eqref{eq1}, which is the objective of the following two lemmas.
\begin{lemma}
\begin{align}
\ex\left\|\btt(t+1) -\bnu(t+1)\right\|^2 \leq \left(\frac{2}{N} +\frac{(\gamma+\tilde\gamma)}{2N} -\frac{\gamma}{2N^2} \right)\eta_t^2\tau^2G^2+ \frac{d}{2N\alpha_t^2}(\kappa +\tilde\kappa)\sigma_z^2.
\label{firstbound}
\end{align}
\label{firstboundlem}
\end{lemma}
\begin{proof}
See Appendix B.
\end{proof}
\begin{lemma}
\begin{align}
\ex\left\|\bnu(t+1) -\btt^*\right\|^2\leq &(1-\mu\eta_t(\tau-\eta_t(\tau-1)))\ex\left\|\btt(t) -\btt^*\right\|^2+2\eta_t(\tau-1)\Gamma\\&+(1+\mu(1-\eta_t))\eta_t^2G^2\frac{\tau(\tau-1)(2\tau-1)}{6}+\eta_t^2(\tau^2+\tau-1)G^2.
\label{secondbound}
\end{align}
\end{lemma}

\begin{proof}
The proof is classical in the convergence analysis of FL since the terms involved do not depend on the transmission scheme and can be found in several papers in literature such as \cite{amiri2021blind,amiri2020update}.
\end{proof}

\section*{Appendix B\\Proof of Lemma \ref{firstboundlem}}
\label{prooflem1}
\noindent First, we have
$$
\ex\left\|\btt(t+1) -\bnu(t+1)\right\|^2=  \ex\left\|\Delta\widehat\btt(t) -\Delta\btt(t)\right\|^2 = \sum_{i=1}^d\ex \left(\Delta\widehat\theta_i(t)-\Delta\theta_i(t)\right)^2,
$$
where $\Delta\theta_i(t)$ is the $i$-th entry of  $\Delta\btt(t)$. Note that 
$$
\Delta\widehat\theta_i(t) = \sum_{p=1}^4 \Delta\widehat\theta_{i,p}(t),
$$
where $\Delta\widehat\theta_{i,p}(t), \ p\in[4],$ are defined in \eqref{yis}.
Given that the channels and the noise are independent of the model updates at the global round $t$, it can be easily seen that
$$
\ex \left(\Delta\widehat\theta_i(t)-\Delta\theta_i(t)\right)^2= \ex \left(\Delta\widehat\theta_{i,1}(t)-\Delta\theta_i(t)\right)^2 +\sum_{p=2}^4\ex \left(\Delta\widehat\theta_{i,p}(t)\right)^2.
$$
\begin{lemma}
\begin{align}
 \sum_{i=1}^d\ex \left(\Delta\widehat\theta_{i,1}(t)-\Delta\theta_i(t)\right)^2 \leq \frac{\gamma}{N^2}\sum_{n=1}^N\ex \|\Delta\btt_n(t)\|^2,
\end{align}
with $\gamma = \max_{n,n'}\frac{1}{c_n^2}\tr(\bR_n\bD_{n} \bR_{n'} \bD_{n})$.
\label{term1}
\end{lemma}
\begin{proof}
First, substituting $\Delta\widehat\theta_{i,1}(t)$ by its expression, we have for $1\leq i \leq d/2$
\begin{align}
\ex \left[\left(\Delta\widehat\theta_{i,1}(t)-\Delta\theta_i(t)\right)^2\right] &= \ex \left[\left( \frac{1}{N} \sum_{n=1}^N  \left(\frac{1}{c_n}\bh_n^i(t)^H\bD_n\bh_n^i(t)-1\right)\Delta\theta_{n,i}(t)\right)^2\right]\nonumber\\
&\overset{(a)}{=}\frac{1}{N^2} \sum_{n=1}^N\ex \left[ \left(\frac{1}{c_n}\bh_n^i(t)^H\bD_n\bh_n^i(t)-1\right)^2\left(\Delta\theta_{n,i}(t) \ \right)^2\right]\nonumber\\
&\overset{(b)}{=}\frac{1}{N^2} \sum_{n=1}^N \frac{\tr\left(\bR_n\bD_n \bR_n \bD_n\right)}{c_n^2}\ex\left(\Delta\theta_{n,i}(t) \ \right)^2 \nonumber\\
&\overset{(c)}{\leq} \frac{\gamma}{N^2} \sum_{n=1}^N \ex\left(\Delta\theta_{n,i}(t)\right)^2,
\label{eq2}
\end{align}
where $(a)$ follows from the independence of channel vectors $\bh_n^i(t)$ from each other and from the model updates at global round $t$, and by noting that $\frac{1}{c_n}\ex\bh_n^i(t)^H\bD_n\bh_n^i(t) = 1, \ \forall n$. $(b)$ is obtained by computing the expectations over the channels as follows
\begin{align*}
\ex\left(\frac{1}{c_n}\bh_n^i(t)^H\bD_n\bh_n^i(t)-1\right)^2 &= \ex\left(\frac{1}{c_n}\bh_n^i(t)^H\bD_n\bh_n^i(t)\right)^2- 2 \ex\left(\frac{1}{c_n}\bh_n^i(t)^H\bD_n\bh_n^i(t)\right)+1\\
&= \frac{\tr (\bR_n\bD_n \bR_n \bD_n)+ \left(\tr(\bR_n\bD_n)\right)^2}{c_n^2}-\frac{2\tr(\bR_n\bD_n)}{c_n}+1\\
&=\frac{\tr (\bR_n\bD_n \bR_n \bD_n)}{c_n^2}.
\end{align*}
And $(c)$ follows from the definition of $\gamma$.
Similarly, we obtain for $d/2+1\leq i\leq d$,
\begin{align}
\ex \left[\left(\Delta\widehat\theta_{i,1}(t)-\Delta\theta_i(t)\right)^2\right] &= \ex \left[\left( \frac{1}{N} \sum_{n=1}^N  \left(\frac{1}{c_n}\bh_n^{i-d/2}(t)^H\bD_n\bh_n^{i-d/2}(t)-1\right)\Delta\theta_{n,i}(t)\right)^2\right]\nonumber\\
&{\leq} \frac{\gamma}{N^2} \sum_{n=1}^N \ex\left(\Delta\theta_{n,i}(t)\right)^2.
\label{eq3}
\end{align}
Lemma \ref{term1} follows readily from \eqref{eq2} and \eqref{eq3}.
\end{proof}

\begin{lemma}
\begin{align}
 \sum_{i=1}^d\ex \left[\left(\Delta\widehat\theta_{i,2}(t)\right)^2\right] \leq  \frac{\gamma(N-1)}{N^2}\sum_{n=1}^N  \ex \|\Delta\btt_n(t)\|^2
\end{align}
with $\gamma = \max_{n,n'}\frac{1}{c_n^2}\tr(\bR_n\bD_{n} \bR_{n'} \bD_{n})$.
\label{term2}
\end{lemma}
\begin{proof}
For $1\leq i\leq d/2$, replacing $\Delta\widehat\theta_{i,2}(t)$ by its expression given in \eqref{theta_hat}, we have
\begin{align*}
\ex \left[\left(\Delta\widehat\theta_{i,2}(t)\right)^2\right] 
 &=\!\ex\!\left( \frac{1}{N} \sum_{n=1}^N\sum_{\substack{n'=1\\n'\neq n}}^N Re\left\{ \left(\frac{1}{c_n}\bh_n^i(t)^H\bD_n\bh_{n'}^i(t)\right)\left(\Delta\theta_{n',i}(t)+ j\Delta\theta_{n',i+d/2}(t)\right)\right\}\right)^2\\
  & \overset{(a)}{=}\!\ex\!\left[\frac{1}{N^2}\!\sum_{n=1}^N\!\sum_{\substack{n'=1\\n'\neq n}}^N\!\left(\!Re\left\{ \left(\frac{1}{c_n}\bh_n^i(t)^H\bD_n\bh_{n'}^i(t)\!\right)\left(\Delta\theta_{n',i}(t)+ j\Delta\theta_{n',i+d/2}(t)\right)\right\}\!\right)^2
  \right.\\
  & \left. \quad + Re\left\{ \left(\frac{1}{c_n}\bh_n^i(t)^H\bD_n\bh_{n'}^i(t)\right)\left(\Delta\theta_{n',i}(t)+ j\Delta\theta_{n',i+d/2}(t)\right)\right\}\right. \\ & \left. \quad \quad Re\left\{ \left(\frac{1}{c_{n'}}\bh_{n'}^i(t)^H\bD_{n'}\bh_{n}^i(t)\right)\left(\Delta\theta_{n,i}(t)+ j\Delta\theta_{n,i+d/2}(t)\right)\right\}\right]\\
  & \overset{(b)}{=}\ex \left[\frac{1}{2N^2} \sum_{n=1}^N\sum_{\substack{n'=1\\n'\neq n}}^N \frac{1}{c_n^2}\tr(\bR_n\bD_n \bR_{n'} \bD_n)\left(\Delta\theta_{n',i}(t)^2+ \Delta\theta_{n',i+d/2}(t)^2\right)
  \right.\\
  & \left. \quad + \frac{1}{c_nc_{n'}}\tr(\bR_n\bD_{n} \bR_{n'} \bD_{n'})\left(\Delta\theta_{n,i}(t)\Delta\theta_{n',i}(t)-\Delta\theta_{n,i+d/2}(t)\Delta\theta_{n',i+d/2}(t)\right)\right]
\end{align*}
where $(a)$ follows from the independence of the channels $\bh_{n}^i(t), \ \forall n,i$, and $(b)$ is obtained by taking the expectation over the channels.
Similarly, for $ d/2+1\leq i\leq d$, following the same steps, we get
\begin{align*}
\ex \left[\left(\Delta\widehat\theta_{i,2}(t)\right)^2\right]&=\!\ex\!\left(\!\frac{1}{N} \sum_{n=1}^N\sum_{\substack{n'=1\\n'\neq n}}^N Im\left\{ \left(\frac{1}{c_n}\bh_n^i(t)^H\bD_n\bh_{n'}^i(t)\right)\left(\Delta\theta_{n',i-d/2}(t)+ j\Delta\theta_{n',i}(t)\right)\right\}\right)^2\\ 
&= \ex\!\left[\!\frac{1}{N^2}\!\sum_{n=1}^N\!\sum_{\substack{n'=1\\n'\neq n}}^N\!\left( Im\left\{ \left(\!\frac{1}{c_n}\bh_n^i(t)^H\bD_n\bh_{n'}^i(t)\!\right)\left(\Delta\theta_{n',i-d/2}(t)+ j\Delta\theta_{n',i}(t)\right)\right\}\right)^2
  \right.\\
  & \left. \quad + Im\left\{ \left(\frac{1}{c_n}\bh_n^i(t)^H\bD_n\bh_{n'}^i(t)\right)\left(\Delta\theta_{n',i-d/2}(t)+ j\Delta\theta_{n',i}(t)\right)\right\}\right. \\ & \left. \quad \quad Im\left\{ \left(\frac{1}{c_{n'}}\bh_{n'}^i(t)^H\bD_{n'}\bh_{n}^i(t)\right)\left(\Delta\theta_{n,i-d/2}(t)+ j\Delta\theta_{n,i}(t)\right)\right\}\right]\\
  &= \ex \left[\frac{1}{2N^2} \sum_{n=1}^N\sum_{\substack{n'=1\\n'\neq n}}^N \left(\frac{1}{c_n^2}\tr(\bR_n\bD_n \bR_{n'} \bD_n)\left(\Delta\theta_{n',i-d/2}(t)^2+ \Delta\theta_{n',i}(t)^2\right)
  \right. \right.\\
  & \left.\left. \quad +\frac{1}{c_n c_{n'}}\tr(\bR_n\bD_{n} \bR_{n'} \bD_{n'})\!\left(\Delta\theta_{n,i}(t)\Delta\theta_{n',i}(t)-\Delta\theta_{n,i-d/2}(t)\Delta\theta_{n',i-d/2}(t)\right)\right)\!\right].
\end{align*}
Summing over all $i\in[d]$, it follows that 
\begin{align*}
\sum_{i=1}^d\ex \left[\left(\Delta\widehat\theta_{i,2}(t)\right)^2\right]& =\frac{1}{N^2}\sum_{n=1}^N\sum_{\substack{n'=1\\n'\neq n}}^N  \frac{1}{c_n^2}\tr(\bR_{n'}\bD_n \bR_{n} \bD_n)\ex \|\Delta\btt_n(t)\|^2\\
& \leq \frac{\gamma(N-1)}{N^2}\sum_{n=1}^N  \ex \|\Delta\btt_n(t)\|^2,
\end{align*}
where the last inequality follows from the definition $\gamma$.
\end{proof}

\begin{lemma}
\begin{align}
 \sum_{i=1}^d\ex \left[\left(\Delta\widehat\theta_{i,3}(t)\right)^2\right] \leq \frac{\tilde{\gamma}}{N}\sum_{n=1}^N  \ex \|\Delta\btt_n(t)\|^2,
\end{align}
where $\tilde\gamma = \max_{n,n'}\frac{1}{c_n^2}\tr(\bC_n\bD_{n} \bR_{n'} \bD_{n})$.
\label{term4}
\end{lemma}
\begin{proof}
Applying the same techniques as in the proof of the previous lemma, we have for $1\leq i\leq d/2$,
\begin{align*}
\ex \left[\left(\Delta\widehat\theta_{i,3}(t)\right)^2\right] 
 &= \ex \left[\left( \frac{\alpha_t}{N}\sum_{n=1}^N \sum_{\substack{n'=1}}^N Re \left\{ \frac{ \tilde\bh_{n}^i(t)^H\bD_n \bh_{n'}^i(t)}{c_n}  \left(\Delta \theta_{n',i}(t) +j \Delta \theta_{n',i+\frac{d}{2}}(t)\right)\right\}\right)^2\right]\\
 &\overset{(a)}{=} \ex \left[\frac{\alpha_t}{N}\sum_{n=1}^N \sum_{\substack{n'=1}}^N \left( Re \left\{ \frac{ \tilde\bh_{n}^i(t)^H\bD_n \bh_{n'}^i(t)}{c_n}  \left(\Delta \theta_{n',i}(t) +j \Delta \theta_{n',i+\frac{d}{2}}(t)\right)\right\}\right)^2\right]\\
   &\overset{(b)}{=}\ex \left[\frac{1}{2N^2} \sum_{n=1}^N\sum_{\substack{n'=1}}^N \frac{1}{c_n^2}\tr(\bC_n\bD_n \bR_{n'} \bD_n)\left(\Delta\theta_{n',i}(t)^2+ \Delta\theta_{n',i+d/2}(t)^2\right)
\right].
 \end{align*}
Similarly, for $ d/2+1\leq i\leq d$,
\begin{align*}
\ex \left[\left(\Delta\widehat\theta_{i,3}(t)\right)^2\right] 
  &= \ex \left[\frac{1}{2N^2} \sum_{n=1}^N\sum_{\substack{n'=1}}^N \frac{1}{c_n^2}\tr(\bC_n\bD_n \bR_{n'} \bD_n)\left(\Delta\theta_{n',i-d/2}(t)^2+ \Delta\theta_{n',i}(t)^2\right)
\right].
\end{align*}
Summing over $i\in[d]$, we get
\begin{align*}
\sum_{i=1}^d\ex \left[\left(\Delta\widehat\theta_{i,3}(t)\right)^2\right]& =\frac{1}{N^2}\sum_{n=1}^N\sum_{\substack{n'=1}}^N  \frac{1}{c_n^2}\tr(\bC_{n}\bD_n \bR_{n'} \bD_n)\ex \|\Delta\btt_n(t)\|^2\\
& \leq \frac{\tilde{\gamma}}{N}\sum_{n=1}^N  \ex \|\Delta\btt_n(t)\|^2.
\end{align*}
\end{proof}

\begin{lemma}
\begin{align}
 \sum_{i=1}^d\ex \left[\left(\Delta\widehat\theta_{i,4}(t)\right)^2\right] \leq  \frac{d}{2N\alpha_t^2}(\kappa+\tilde\kappa)\sigma_z^2,
\end{align}
with $\kappa =\frac{1}{N}\sum_{n=1}^N \frac{1}{c_n}$ and $\tilde\kappa = \frac{1}{N}\sum_{n=1}^N \frac{ \tr(\bC_n\bD_n)}{c_n^2}$.
\label{term3}
\end{lemma}
\begin{proof}
For $1\leq i\leq d/2$, we have
\begin{align*}
\ex \left[\left(\Delta\widehat\theta_{i,4}(t)\right)^2\right] &= \ex\left(\frac{1}{N\alpha_t}\sum_{n=1}^N Re\left\{\frac{\left(\bh_{n}^i(t) +\tilde\bh_{n}^i(t)\ \right)^H\bD_n\bz^i(t)}{c_n }\right\}\right)^2\\
&\overset{(a)}{=}\ex\frac{1}{N^2\alpha_t^2}\sum_{n=1}^N\frac{1}{c_n^2} \left(Re\left\{{\bh_{n}^i(t)^H \bD_n\bz^i(t)}{}\right\}\right)^2+\left(Re\left\{{\tilde\bh_{n}^i(t)^H\bD_n\bz^i(t)}\right\}\right)^2\\
&\overset{(b)}{=} \frac{1}{2N\alpha_t^2}(\kappa+\tilde\kappa)\sigma_z^2,
\end{align*}
where $(a)$ follows from the independence of the channels $\bh_n^i(t)$, the channels error $\tilde\bh_n^i(t)$ and the noise, and $(b)$ follows from 
\begin{align*}
&\frac{1}{c_n^2}\ex\left(Re\left\{{\bh_{n}^i(t)^H\bD_n\bz^i(t)}\right\}\right)^2= \frac{\sigma_z^2}{2c_n^2}\tr(\bR_n\bD_n)= \frac{\sigma_z^2}{2c_n},\\
&\frac{1}{c_n^2}\ex\left(Re\left\{{\tilde\bh_{n}^i(t)^H\bD_n\bz^i(t)}{ }\right\}\right)^2= \frac{\sigma_z^2}{2c_n^2}\tr(\bC_n\bD_n).
\end{align*}
Similar analysis yields for for $d/2+1\leq i\leq d$
\begin{align*}
\ex \left[\left(\Delta\widehat\theta_{i,4}(t)\right)^2\right]
=\frac{1}{2N\alpha_t^2}(\kappa+\tilde\kappa)\sigma_z^2.
\end{align*}
\end{proof}
%
%
Combining the results of Lemmas \ref{term1}-\ref{term3}, it follows that
\begin{align}
\ex\left\|\btt(t+1) -\bnu(t+1)\right\|^2 &\leq \frac{\gamma+\tilde\gamma}{N} \sum_{n=1}^N  \ex \|\Delta\btt_n(t)\|^2 + \frac{d}{2N\alpha_t^2}(\kappa +\tilde\kappa)\sigma_z^2\nonumber\\
&\leq  \frac{\gamma+\tilde\gamma}{N}\eta_t^2\tau^2G^2+ \frac{d}{2N\alpha_t^2}(\kappa +\tilde\kappa)\sigma_z^2,
\end{align}
where the last result follows from
\begin{align}
\ex \|\Delta\btt_n(t)\|^2 = \ex \left\|\sum_{i=0}^{\tau-1} \eta_t  F_{n,\xi_{n,i}^t}'(\btt_{i}^n(t)) \right\|^2 \leq \eta_t^2\tau\sum_{i=0}^{\tau-1}\left\|   F_{n,\xi_{n,i}^t}'(\btt_{i}^n(t))\right\|^2 \leq \eta_t^2\tau^2G^2.
\end{align}

\bibliographystyle{IEEEtran}
\bibliography{IEEEabrv,IEEEconf,references}
\end{document}